%% file: Stability.tex
\documentclass[11pt]{article}
\usepackage{typearea}\typearea{12}
\usepackage{epsfig}
\input{stabMacro.tex}

\newcommand{\BoxedEPSF}[1]{\relax}
\makeatletter
\long\def\@makecaption#1#2{{\small
\advance\leftskip1cm
\advance\rightskip1cm
\vskip\abovecaptionskip
\sbox\@tempboxa{#1: #2}%
\ifdim \wd\@tempboxa >\hsize
 #1: #2\par
\else
\global \@minipagefalse
\hb@xt@\hsize{\hfil\box\@tempboxa\hfil}%
\fi
\vskip\belowcaptionskip}}
\makeatother
\begin{document}
\begin{flushright}
{\footnotesize
J. Stat. Phys. {\bf 84}, 535--653 (1996).
}
\end{flushright}
\noindent
{\Large\bf Stability of Ferromagnetism
in
Hubbard Models\\
with Nearly-Flat Bands}
\par\bigskip

\renewcommand{\thefootnote}{\fnsymbol{footnote}}
\noindent
Hal Tasaki\footnote{
Department of Physics, Gakushuin University, Mejiro, Toshima-ku, 
Tokyo 171-8588, Japan
}
\renewcommand{\thefootnote}{\arabic{footnote}}\setcounter{footnote}{0}

\begin{quote}\small
Whether spin-independent Coulomb interaction in an electron system can be 
the origin of ferromagnetism has been an open problem for a long time.
Recently, a ``constructive'' approach to this problem has been developed, 
and
the existence of ferromagnetism in the ground states of certain  
Hubbard models was established rigorously.
A special feature of these Hubbard models is that their lowest bands (in 
the corresponding single-electron problems) are completely flat.
Here we study models obtained by adding
small but arbitrary translation-invariant perturbation to the hopping 
Hamiltonian of these flat-band models.
The resulting models have nearly-flat lowest bands.
We prove that the ferromagnetic state is stable against a single-spin
flip provided that Coulomb interaction $U$ is sufficiently large.
(It is easily found that the same state is unstable against a single-spin 
flip if $U$ is small enough.)
We also prove upper and lower bounds for the dispersion relation of 
the lowest energy eigenstate with a single flipped spin, which bounds 
establish that the model has ``healthy'' spin-wave excitation.
It is notable that the (local) stability of ferromagnetism
is proved in 
non-singular Hubbard models, in which we must overcome
competition between the kinetic energy and the Coulomb interaction.
We also note that this is one of the very few rigorous
and robust results  which 
deal with truly 
nonperturbative phenomena in many electron systems.
The local stability strongly suggests that the Hubbard models with nearly 
flat bands have ferromagnetic ground states.
We believe that the present models can be studied as paradigm models for 
(insulating) ferromagnetism in itinerant electron systems.
\end{quote}
\tableofcontents
\newpage\Section{Introduction}
\label{SecIntro}
\subsection{Background}
\label{secback}
The origin of strong ferromagnetic ordering observed in some
materials has been a mystery in physical science for a
long time \cite{Mattis81}. 
Since non-interacting electron systems universally exhibit
paramagnetism, the origin of ferromagnetism
should be sought in electron-electron interaction.
In most solids, however,
the dominant part of interaction between electrons
is the Coulomb interaction, which is perfectly spin-independent.
(See Chapter~32, page 674 of \cite{AshcroftMermin76} for example.)
Therefore we are faced with a very interesting and 
fundamental problem in theoretical physics 
to determine {\em whether spin-independent
interaction in an itinerant 
electron system can be the origin of ferromagnetic 
ordering.}
This problem is important not only because ferromagnetism is 
a very common (and useful) phenomenon, but because 
it focuses on
a fundamental role of nonlinear interactions in many-body quantum
mechanical systems.

The present work is a continuation of our work \cite{92e,93d},
where we dealt with the above fundamental problem from a standpoint of 
``constructive condensed matter physics.''
Our goal is to provide concrete models in which the existence of 
ferromagnetic ordering can be established rigorously.
Such models should shed light on mechanisms by which Coulomb interaction 
generates ferromagnetic ordering.

It was Heisenberg
\cite{Heisenberg28} who first realized that ferromagnetism is an
intrinsically quantum mechanical phenomenon.
In Heisenberg's approach to ferromagnetism, one starts
from the picture that each electron (relevant to magnetism)
is almost localized at an atomic orbit.
By treating the effect of Coulomb interaction and
overlap between nearby atomic orbits in a perturbative manner,
Heisenberg concluded that there appears
``exchange interaction'' between nearby electronic spins
which determines magnetic properties of the system.

The validity of the Heisenberg's picture has been challenged
both from theoretical and from physical points of views. 
(See, for example, \cite{Herring66a}.)
It has been realized that, in most of the situations,
the exchange interaction is antiferromagnetic
rather than ferromagnetic.
Moreover conditions which would justify the basic assumption that 
electrons can be treated as localized at atomic sites are not
well understood\footnote{
This issue is closely related to the problem of Mott-Hubbard insulators.
}.

In a different approach to the problem of ferromagnetism,
which was originated by Bloch \cite{Bloch29}, one starts from
the quantum mechanical free electron gas, in which
electrons are in plane-wave like states.
One then treats the effect of Coulomb interaction perturbatively,
and tries to find instability against certain
magnetic ordering.
When combined with the Hartree-Fock approximation (or a mean-field
theory), this approach leads to the picture that there is an instability
against ferromagnetism when the density of states at the fermi energy
and the Coulomb interaction are sufficiently large.

It has been realized, however, that
the Hartree-Fock approximation drastically overestimates the 
tendency towards ferromagnetism, thus predicting the existence
of ferromagnetism in many situations where 
it does not take place.
From a theoretical point of view, the approximation is unsatisfactory
since it artificially replaces the fundamental $SU(2)$ symmetry 
(i.e., rotation symmetry in the spin space) of the
electron systems with a discrete ${\bf Z}_2$ symmetry.
Although there have appeared many improvements of the
simple Hartree-Fock theory, this approach does not provide
a conclusive answer to the fundamental problem about
the origin of ferromagnetism that we raised in the beginning
of the present subsection.
See, for example, \cite{Herring66b} for a review.

\subsection{Ferromagnetism in the Hubbard Model}
\label{secHubint}
A modern version of the problem about the origin 
of ferromagnetism was formulated
by Kanamori \cite{Kanamori63}, Gutzwiller
\cite{Gutzwiller63}, and Hubbard \cite{Hubbard63} in 1960's.
(The similar formulation was given earlier, for example, 
in \cite{Slater53}.) 
They
studied simple tight-binding models of electrons with
on-site Coulomb interaction\footnote{
It is sometimes argued that the originally long ranged Coulomb 
interaction becomes short ranged by the screening effect from 
electrons in the bands (or orbits) which are {\em not} taken into 
account in the Hubbard model.
But it is still true that the assumption that there is only on-site 
interaction is highly artificial.
} whose strength
is denoted as $U$. 
The model is usually called the
Hubbard model. 
When there is no
electron-electron interaction (i.e., $U=0$), 
the model 
exhibits paramagnetism as an inevitable consequence
of the Pauli exclusion principle. 
Among other things,
Kanamori, Gutzwiller, and Hubbard asked 
{\em whether the paramagnetism found for $U=0$ can
be converted into ferromagnetism when there is a sufficiently
large Coulomb interaction $U$.}
This is a concrete formulation of the fundamental problem
that we raised in the opening of the previous subsection.

It is worth noting that the on-site Coulomb interaction
itself is completely independent of electronic spins, and does
not favor any magnetic ordering. 
Therefore
one does not find any terms in the Hubbard Hamiltonian which 
explicitly favor ferromagnetism (or any other 
ordering).
Our theoretical goal will be to show that magnetic ordering arises as a 
consequence of subtle interplay between kinetic motion of electrons and 
the short-ranged Coulomb interaction. 
It is interesting to compare the situation with that in spin systems, 
where one is often given a Hamiltonian which favors some kind of magnetic
ordering, and the major task is to investigate if such ordering really 
takes place.
We can say that the Hubbard model formulation goes deeper 
into fundamental mechanisms of 
magnetism than that of spin systems.
It offers a
challenging problem
to theoretical physicists to derive magnetic interaction from 
models which do not explicitly contain such interactions.
Perhaps the best justification of the Hubbard model as a standard model 
of itinerant electron systems comes from such a theoretical 
consideration, rather than its accuracy in modeling narrow band electron 
systems. 
See also the introduction of \cite{93d,MullerHartmann95,StrackVollhardt94b} 
for discussions about ferromagnetism in the Hubbard model.

We stress that ferromagnetism is {\em not} a universal property of the 
Hubbard model.
The Hubbard model is believed to exhibit various phenomena including 
paramagnetism, antiferromagnetism, ferrimagnetism, ferromagnetism, or 
superconductivity, depending on various conditions.
Such drastic ``non-universality'' of the model motivated us to 
take the present ``constructive'' approach rather than to prove  
theorems which apply to general Hubbard models.

The problem of ferromagnetism in the Hubbard model
was extensively studied by using various heuristic methods.
The Hartree-Fock approximation discussed above leads one to the
so called Stoner criterion.
It says that the Hubbard model exhibits ferromagnetism if
one has $UD_{\rm F}>1$, where $D_{\rm F}$ is the density of 
states of the corresponding single-electron problem measured
at the fermi level (of the corresponding non-interacting system).
Although the criterion cannot be trusted literally, it guides us
to look for ferromagnetism in models with not too small $U$ and/or
large density of states.

The first rigorous result about ferromagnetism in the Hubbard model
was provided by Nagaoka \cite{Nagaoka66}, and independently
by Thouless \cite{Thouless65} in 1965.
It was proved that certain Hubbard models have ground states with
saturated ferromagnetism when there is exactly one hole and the
Coulomb repulsion $U$ is infinite.
See \cite{Lieb71,89c} for shorter proofs.
Whether the Nagaoka-Thouless ferromagnetism survives in the
models with finite density of holes and/or finite Coulomb repulsion
is a very interesting but totally unsolved problem
\cite{DoucotWen89,Shastry90,Suto91a,Suto91b,Toth91,%
HanischMullerHartmann93,Putikka93,LiangPang94,Kusakabe93}.
See also the introduction of \cite{MullerHartmann95} for a compact review 
of this subject.

Very recently, M\"{u}ller-Hartmann \cite{MullerHartmann95} 
argued that the Hubbard model with $U=\infty$ 
on a one-dimensional zigzag chain exhibits 
ferromagnetism\footnote{
Although M\"{u}ller-Hartmann's argument is quite interesting, it does not 
form a mathematically rigorous proof (as far as we can read off from  
\cite{MullerHartmann95}).
The argument involves an uncontrolled 
continuum limit of a strongly interacting system.
To make it into a rigorous proof seems to be a nontrivial task.
}. 
Interestingly, the 
geometry of the chain is similar to that of one-dimensional  
models studied in the present paper.

\Rem
It should be noted that the Hubbard model is by no means the 
unique formulation for studying strong correlation 
effects in narrow band itinerant electron systems.
If one recalls how a tight-binding model is derived (or supposed to be 
derived) from a continuum model, there is a good reason to consider 
models with more complicated interactions than mere on-site Coulomb 
repulsion. 
One can even include interactions which explicitly favor ferromagnetism, 
and still formulate interesting problems.
See \cite{StrackVollhardt94a} for an approach to ferromagnetism in such 
extended Hubbard models.

\subsection{Flat-Band Ferromagnetism}
In 1989, Lieb proved an important general theorem for the Hubbard model
at half filling on a bipartite lattice \cite{Lieb89}.
As a corollary of the theorem, Lieb showed that a rather general class of 
Hubbard model exhibits ferrimagnetism\footnote{
Ferrimagnetism is a kind of antiferromagnetism on a bipartite lattice 
such that the numbers of sites in two sublattices are different.
}.
See also \cite{ShenQiu94}.

In 1991, Mielke \cite{Mielke91a,Mielke91b} 
came up with a new class of rigorous 
examples of ferromagnetism in the Hubbard model.
He showed that the Hubbard models on a general class of line 
graphs have ferromagnetic ground states.
A special feature of Mielke's model is that the corresponding 
single-electron Schr\"{o}dinger equation\footnote{
Here (and throughout the present paper) we are talking about the genuine 
one-particle problem, not an effective 
(and uaully ill-defined) one-particle problem in interacting 
system which are often discussed in heuristic works.
} has highly degenerate ground states. 
In other words, Mielke's models have flat (or dispersionless) bands.
The original results of Mielke's were for the electron number which 
corresponds to the half-filling of the lowest flat band, but later it was 
extended to different electron densities in two dimensional 
models \cite{Mielke92}.

A similar but different class of examples of ferromagnetism in the 
Hubbard models were proposed in \cite{92e,93d}.
These models were defined on a class of decorated lattices, and were also 
characterized by flat bands at the bottom of the single-electron spectrum.
In a class of models in two and higher dimensions, it was proved that the 
ferromagnetism is stable against fluctuation of electron numbers
\cite{92e,93d}.

The examples of ferromagnetism in 
\cite{Mielke91a,Mielke91b,Mielke92,92e,93d} are
common in that they treat special models with flat lowest 
bands\footnote{
Lieb's examples also have flat bands in the middle of the 
single-electron spectra.
}. 
The ferromagnetism established for these models are now
called flat-band 
ferromagnetism \cite{Kusakabe93}. 
There is a general theorem due to Mielke \cite{Mielke93} which 
states a necessary and sufficient condition for a Hubbard model with a
flat 
lowest band to exhibit ferromagnetism when the flat-band is half-filled.
Although flat-band ferromagnetism sheds light on very important aspects 
of the role of strong interaction in itinerant electron systems, it  
relies on the rather singular assumption that the models have 
completely flat bands.
As we discuss in Section~\ref{secMM1d}, we do not have true 
``competition'' between the kinetic energy and the Coulomb interaction.

If one adds small perturbation to the hopping Hamiltonian of a flat-band 
model, one generically gets a model with slightly dispersive lowest band.
It was conjectured \cite{92e,93d} 
that such models with nearly flat-bands exhibit 
ferromagnetism provided that the Coulomb interaction $U$ is large enough.
Kusakabe and Aoki \cite{KusakabeAoki94a,KusakabeAoki94b} 
presented detailed study of this 
problem by numerical experiments and careful variational 
calculations. 
Their results provide strong support that the flat-band ferromagnetism 
is stable against small perturbations to the band structure.

We stress that this is a very delicate conjecture for the following reasons.
\begin{itemize}
	\item  
	When the ground states are ferromagnetically ordered, there inevitably 
	exist spin-wave (magnon, or Nambu-Goldstone) modes whose 
	excitation energies are of order $L^{-2}$, where $L$ denotes the linear 
	size of the lattice.
	The total energy of the perturbation, on the other hand, is always 
	proportional to the system volume $L^{d}$.
	This means that the total perturbation always exceeds the energy gap 
	when the system size becomes large.
	Such a situation can never be dealt with naive perturbation theories.
	
	\item  
	When the lowest band is non-flat, the model with $U=0$ exhibits Pauli 
	paramagnetism.
	It is strongly believed that, for sufficiently small $U$, the ground 
	states of the models (in finite volumes) are spin-singlet.
	Therefore one must have sufficiently large $U$ to get ferromagnetism.
	This means that the problem is a truly {\em nonperturbative} one.
\end{itemize}
In other words, one must directly face the notorious difficult problem of 
``competition'' between the kinetic energy and the Coulomb 
interaction.
Technically speaking, such natures of the problem inhibit one from making 
use of the common strategy to construct exact ground states by minimizing 
local Hamiltonians\footnote{
It turned out that there are exceptions to this statement \cite{95c}.
See the remark at the end of Section~\ref{secAbout}.
However we still believe that the above comment is true for generic 
models.
}.
This strategy has been used to derive exact ground states of various 
(extended) Hubbard models 
\cite{BrandtGiesekus92,StrackVollhardt94a,93b,94a,Boer95}, 
as well as in our early works \cite{92e,93d} on 
the flat-band Hubbard models.
During the successful history of mathematical physics,
there have been developed rigorous perturbation theories for various 
many-body problems, including classical and quantum spin systems and 
quantum field theories.
As far as we know, however, there is no general theory which enables one 
to control generic perturbation in models which exhibit continuous 
symmetry breaking.

\subsection{About the Present Paper}
\label{secAbout}

In the present paper, we report the first important step
towards the solution of the above problem about stability of flat-band 
ferromagnetism. 
We treat models with nearly flat bands
obtained by adding almost arbitrary perturbations
to the hopping matrices of the flat-band models.
For sufficiently large $U$, we prove that the ferromagnetic state
is locally stable.
More precisely we show that the lowest energy among ferromagnetic states 
is strictly less than the lowest energy among states with a single 
flipped spin. 
The local stability, 
along with the global stability established 
for the flat-band models,
strongly suggests that the ferromagnetic
states are the true ground states of the present models for 
sufficiently large $U$.
(See the remark below.)
We also prove that, in a certain range of the parameter space,
the spin-wave dispersion relations of the 
present models behave exactly as those in the Heisenberg ferromagnet.
This confirms the conjecture of Kusakabe and Aoki 
\cite{KusakabeAoki94b}.
These results were first announced in \cite{94c}.

As far as we know, this is the first time that 
the (local) stability of ferromagnetism is proved in truly
non-singular Hubbard models, overcoming the competition between the
kinetic energy and the Coulomb interaction.
We also note that this is one of the very few rigorous
and robust works in which 
nonperturbative aspects of many electron problems are treated.
Recently there have been remarkable progress in rigorous 
treatment of interacting many fermion systems based on renormalization 
group techniques.
However these treatments deal only with weak coupling phenomena such as 
the Tomonaga-Luttinger liquid \cite{Benfatto94}, and the 
superconductivity \cite{Feldman95}.

The present paper is organized as follows.
In Section~\ref{Sec1d}, we restrict ourselves to the simplest 
one-dimensional models, and discuss our main 
results and ideas behind the proof.
We have tried hard to make this section accessible to a wide range of 
readers.
In Section~\ref{SecDef}, we introduce general class of models in arbitrary 
dimensions, and state our rigorous results precisely.
Sections~\ref{SecSingle} to \ref{SecBasis} are devoted to the proof of 
our theorems. 
We have carefully organized the lengthy proof so that to make it as 
readable as possible.
One can read off the organization of these sections by taking 
a look at the table of contents.
In general, earlier sections contain physically interesting ideas, and 
later sections contain technical materials.
A browse through Sections~\ref{SecSingle} to \ref{SecProof} should give 
the reader a clear idea about the detailed structure of our proof.

\Rem
(April, 1995)
After the completion of the present paper, we have finally succeeded in 
proving the global stability of ferromagnetism in Hubbard models obtained 
by adding {\em special} perturbation to the flat-band models \cite{95c}.
We stress that this new result does not diminish the importance of the 
present work.
Even for the models treated in the new paper \cite{95c}, the only way 
(that we know of) to prove meaningful lower bounds for spin-wave 
excitation energy is via the machinery developed here.
The robustness of the present results (in the sense that we allow 
arbitrary weak translation-invariant perturbation) is also important.
\Section{Stability of Ferromagnetism in One-Dimensional Models}
\label{Sec1d}    
 In the present section, we discuss our main results and the basic ideas of
their proof in the context of simplest one-dimensional Hubbard models.
The advantage of restricting ourselves to one-dimensional models is
that we can discuss the essences of our theory without being bothered by
many technical details. 
In particular the analysis of the band structure
(Section~\ref{secband1d}) and the construction of localized bases
(Section~\ref{secbases1d}) can be carried out in explicit and elementary
manners, thanks to special features of the simple models. 
These explicit calculation will be
a good introduction to more elaborate analyses in the general class of
models.
Fortunately, the ideas developed in one-dimensional models can be used
in the study of the general models in higher dimensions with only 
technical modifications.

We have tried to make the present section self-contained, and  
accessible to a wide range of readers. 
We urge the readers to take a look at this section, no matter whether
he/she is planning to study the later sections.

We note that what we present in this section is far from a complete 
mathematical proof.
We often neglect ``small'' contributions without any justifications,
and some of the formulas are not perfectly correct (in view of the rigorous
analysis presented in the later sections).
Nevertheless we believe that the material presented here will give a clear
idea about the philosophy and the structure of our proof.

\subsection{Models and Main Results}
\label{secMM1d}
We define the simplest two-band models in one-dimension, and describe
what we can prove about the stability of ferromagnetism and the spin-wave
dispersion relations.
We stress that the restriction to one-dimension is by no means
essential.
All the results here extend to corresponding models in higher
dimensions (i.e., two, three, or even higher).
The reader who is {\em not\/} planning to study the later sections
is invited to take a brief look at Section~\ref{SecDef}, especially
at Figures~\ref{lattice2d}, \ref{band2df}, and \ref{band2dnf}
of two-dimensional lattices and band structures.

Let $L$ be a fixed odd integer, and denote by
\begin{equation}
\Lao=\cbk{-\frac{L-1}{2},\ldots,-1,0,1,\ldots,\frac{L-1}{2}}\subset{\bf Z}
\label{Lao1d}
\end{equation}
the length $L$ chain (identified with a set of integers). 
We also define
\begin{equation}
\La'=\Lao+\frac{1}{2}=\cbk{-\frac{L}{2}+1,\ldots,-\frac{1}{2},\frac{1}{2},
\frac{3}{2},\ldots,\frac{L}{2}},
\end{equation}
which is the chain obtained by shifting $\Lao$ by $1/2$ . 
Our lattice $\La$ is obtained by ``decorating'' the chain $\Lao$ by the sites from
$\La'$ as $\La=\Lao\cup\La'$.
See Figure~\ref{lattice1d}.
One may regard our lattice structure as mimicking that of an oxide, where
sites in $\Lao$ correspond to metallic atoms and sites in $\La'$ correspond to
oxygen atoms.
We have no intention of building models which are realistic from the
view point of condensed matter physics.
But this analogy proves to be helpful in understanding various aspects
of our work, including the basic mechanism of ferromagnetism.

\begin{figure}
\centerline{\includegraphics[width=14cm,clip]{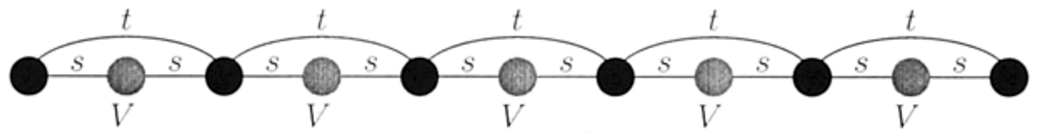}}
\caption{\captionA}
\label{lattice1d}
\end{figure}

We shall study the Hubbard model on $\La$ with the Hamiltonian
\begin{equation}
H=
t\sumtwo{x\in\Lao}{\sigma=\up,\dn}
\rbk{\cxs c_{x+1,\sigma}+\mbox{h.c.}} 
+
s\sumtwo{x\in\La}{\sigma=\up,\dn}
\rbk{\cxs c_{x+(1/2),\sigma}+\mbox{h.c.}} 
+
V\sumtwo{x\in\La'}{\sigma=\up,\dn}n_{x,\sigma}
+
U\sum_{x\in\La}n_{x,\up}n_{x,\dn},
\label{H1d}
\end{equation}
where we use periodic boundary conditions to identify $x$ 
with $x-L$ if necessary.
Here $\cxs$ and $\axs$
are the creation and the annihilation operators, respectively, 
of an electron
at site\footnote{
More precisely these operators correspond to an orbital state around the 
atom at $x$.
We have here assumed that each orbit 
is nondegenerate.
Usually models based on such an assumption are
referred to as single-band Hubbard models.
We find this terminology confusing since our model indeed has
multiple bands in its single-electron spectrum.
We think a better terminology is ``single-state'' Hubbard model.
Then the models we consider are classified as
``single-state multi-band Hubbard models''.
} $x\in\La$
with spin $\sigma=\up,\dn$.
They satisfy the standard fermionic anticommutation relations. 
(See (\ref{ac1}), (\ref{ac2}) for details.)
The corresponding number operator is $\nxs=\cxs\axs$.
Finally ``h.c.'' in (\ref{H1d}) stands for the 
hermitian conjugate.

 The real parameters $t$ and $s$ represent the amplitudes that an electron
hops between neighboring sites in $\Lao$
(separated by a distance $1$)
and between neighboring sites in $\La$
(separated by a distance $1/2$),
respectively.
The real parameter $V$
is the on-site potential energy for the sites in $\La'$.
See Figure~\ref{lattice1d}.
The first three terms in (\ref{H1d})
determine single-electron properties of the model.
The fourth term is the on-site Coulomb interaction characteristic in the
Hubbard model with the interaction energy $U>0$.

We consider many-electron states with the total electron number fixed to
$L$.
(See the end of Section~\ref{secHub} 
for an explicit construction of the Hilbert space.)
Since there are $2L$
sites in the lattice $\La$,
the present electron number corresponds to the quarter-filling of the whole
bands (or the half-filling of the lower band).
This electron number is natural if one imagines that each site in
$\Lao$ (which corresponds to a metallic atom) emits 
one electron to the band\footnote{
Of course one gets the same electron number if each site in $\La'$ emits
one electron.
But we want to insist on the present picture since it gives the desired
electron number for the general class of models studied later.
Moreover the picture to identify $\Lao$-sites as metallic atoms
is consistent with the nature of the ``ferromagnetic ground states''.
}.

The first result about ferromagnetism
deals with the so called flat-band Hubbard model.
To define the model, we introduce a parameter $\la>0$, and set\footnote{
The model studied here is obtained by setting $d=\nu=1$ in the general
class of models introduced in Section~\ref{secHub} and studied in the 
later sections.
In the Hamiltonian of the later sections, the energy is shifted by a constant
so that the lowest band in the flat-band models have vanishing energy.
}
\begin{equation}
s=\la t,\quad V=(\la^2-2)t.
\label{flatcond}
\end{equation}
Then the following strong result was proved in \cite{92e,93d}.
\begin{theorem}
[Flat-band ferromagnetism]
Let $t>0$ and $\la>0$
be arbitrary, and let $s$ and $V$ be determined by (\ref{flatcond}).
Then, for any $U>0$, the ground states of the Hamiltonian (\ref{H1d})
exhibit saturated ferromagnetism, and are nondegenerate apart from the
trivial spin degeneracy.
\label{flat1dTh}
\end{theorem}

More precisely,  a state is said to ``exhibit saturated ferromagnetism'' if
the total spin $\Stot$
of the state takes the maximum possible value $\Smax=L/2$.
See the end of Section~\ref{secHub} for a precise definition of $\Stot$.
See also Theorem~\ref{flatbandTh} for the general theorem, 
and Section~\ref{secflatproof} for 
a proof.
The flat-band ferromagnetism has been established for a general
class of models including those in higher dimensions \cite{Mielke93}.
In a class of models in two and higher dimensions, the existence
of ferromagnetism for lower electron densities, as well as the 
existence of a paramagnetism-ferromagnetism transition (as
the electron density is changed) are established rigorously
\cite{Mielke92,92e,93d}.

 A model determined by the conditions (\ref{flatcond}) with
$t>0$ and $\la>0$ has a very special
feature that the ground states of the corresponding
single-electron Schr\"{o}dinger equation
are $L$-fold degenerate.
In other words, the lower band (in its
single-electron spectrum) is dispersionless (or flat). 
We shall see this explicitly in Section~\ref{secband1d}.
See Figure~\ref{band1d}a.
As a consequence, the many-electron ground states in the non-interacting
model with $U=0$
are highly degenerate.
The total spin can take any of the allowed values $\Stot=1/2,3/2,\ldots,L/2$.
This is a kind of paramagnetism, but is certainly 
different from the Pauli paramagnetism which allows only unique (or
two-fold degenerate) ground state(s) with the minimum possible $\Stot$
(which is $0$ or $1/2$).

 The role of the Coulomb interaction $U$ 
in flat-band ferromagnetism
is to lift the above mentioned degeneracy, and to ``select'' only
the ferromagnetic states as ground states. 
This is why even infinitesimally small $U$ 
is sufficient for stabilizing ferromagnetism.
Although the flat-band ferromagnetism focuses on a nontrivial and
important effect caused by electron interactions, it avoids dealing 
with the
truly difficult problem about ``competition'' between the kinetic energy and
the electron interactions.

 Let us now turn to the models with nearly-flat bands obtained by perturbing
the above models.
In order to simplify the discussion, we consider the simplest possible
perturbation\footnote{
In the general treatment described in the
later sections, we allow completely general perturbations with
translation invariance and certain summability.
See Section~\ref{secHub}.}.
Instead of (\ref{flatcond}), let us set
\begin{equation}
s=\la t,\quad V=(\la^2-2+\kappa)t,
\label{nearflat}
\end{equation}
where the parameter $\kappa$
measures the strength of the perturbation.
As we see soon in Section~\ref{secband1d},
the lower band is no longer flat for $\kappa\ne0$.

Let $\Emin(\Stot)$
denote the lowest energy among the $L$-electron
states with a given total spin $\Stot$.
The Pauli exclusion principle implies that, for a model with 
$\kappa\ne0$ and $U=0$,
these energies satisfy the monotonicity inequality
\begin{equation}
\Emin(1/2)<\Emin(3/2)<\cdots<\Emin(\Smax-1)<\Emin(\Smax).
\label{mono}
\end{equation}
with $\Smax=L/2$.
This is nothing but the Pauli paramagnetism.

We want to examine if these strict inequalities can be reversed as a
consequence of on-site Coulomb interaction.
We stress that this is a truly nonperturbative problem in which one must
directly face the ``competition'' between the kinetic energy 
and the interaction.
In fact it is quite easy to see that we must have a sufficiently large $U$ 
to stabilize ferromagnetism.
\begin{theorem}
[Instability of ``ferromagnetic ground states'' for small $U$]
Let $t>0$ and $\la>0$
be arbitrary, and let $s$
and $V$ be determined by (\ref{nearflat}).
We assume $\kappa\ne0$.
Let 
$\bar{\ep}(\kappa)=(t/2)\abs{4-\{(\la^2+4)^2+2\kappa(\la^2-4)\}^{1/2}
+(\la^2+2\kappa\la^2)^{1/2}}
=\{4/(\la^2+4)\}t\akappa+O(\kappa^2)$
denote the band width of the lower band.
Then for $U$ satisfying $0\le U<\bar{\ep}(\kappa)$,
we have
\begin{equation}
\Emin(\Smax-1)<\Emin(\Smax).
\label{E<E1d}
\end{equation}
\label{inst1dTh}
\end{theorem}
This is the one-dimensional version of Theorem~\ref{instTh}.

 We call the states with $\Stot=\Smax$
which have the energy $\Emin(\Smax)$
the ``ferromagnetic ground states''\footnote{
This is a slight abuse of the word, since the states are not necessarily the
true ground states.}.
It is easily found that the ``ferromagnetic ground states'' are nondegenerate
apart from the trivial $(2\Smax+1)=(L+1)$-fold degeneracy.
(See Lemma~\ref{FerroGSLemma}.)

 Theorem~\ref{inst1dTh}
states that the ``ferromagnetic ground states'' are unstable against a
single-spin flip.
Although the inequality (\ref{E<E1d})
does not tell us what the ground state of the model is, 
it does establish that
the ``ferromagnetic ground states'' are {\em not}
the true ground states.

 Of course results like Theorem~\ref{inst1dTh} can be proved rather easily 
 by the standard
variational argument.
What is really interesting (and difficult) is to get a reversed 
inequality for models with larger values of $U$.
The following is the most important result of the present paper.
\begin{theorem}
[Local stability of ``ferromagnetic ground states'']
Let $t>0$, and let $s$ and $V$ be determined by (\ref{nearflat}).
We further assume that $\la\ge\la_2$, $\akappa\le\kappa_1$,
$\la\akappa\le p_1$,
and
\begin{equation}
U\ge K_1\la^2t\akappa,
\label{U>1d}
\end{equation}
where $\la_2$, $\kappa_1$, $p_1$, and $K_1$ are 
positive constants\footnote{
We use the same symbols for the constants  as in the 
later sections.
In general models, the constants depend on the basic model
parameters $d$, $\nu$, and $R$, but here they are simply constants.
}.
Then we have
\begin{equation}
\Emin(\Smax-1)>\Emin(\Smax).
\label{E>E1d}
\end{equation}
\label{st1dTh}
\end{theorem}
This is the one-dimensional version of Theorem~\ref{stabilityTh}.

 The bound (\ref{E>E1d}) states that the ``ferromagnetic ground states'' are
stable under a single-spin flip.
Clearly the most important condition for the above local stability theorem
is (\ref{U>1d}) which says we must have sufficiently large Coulomb
interaction (compared with the band width $\propto\akappa t$).
This is natural since the opposite inequality (\ref{E<E1d}) holds if $U$ is
small.
We can say that the above local stability theorem establishes
a truly nonperturbative result in which the ``competition''
between the kinetic energy and the electron interaction is
successfully dealt with. 

We recall the readers that both the energies $\Emin(\Smax-1)$ and 
$\Emin(\Smax)$ grow proportionally to the lattice size $L^{d}$, while 
their difference should be proportional to $L^{-2}$.
In such a situation, there seems to be little hope in proving the 
desired inequality (\ref{E>E1d}) for large $L$ by combining suitable 
lower bound for the left-hand side and upper bound for the right-had 
side.
However there are some nice features that save our task from being 
impossible.
In the subspace with $\Stot=\Smax$, the on-site Coulomb repulsion is 
completely irrelevant because of the Pauli principle.
Therefore the energy $\Emin(\Smax)$ in the right-hand side of  
(\ref{E>E1d}) is nothing but the ground state energy of the 
corresponding non-interacting spinless fermion, which energy is known 
exactly (at least formally).
In the subspace with $\Stot=\Smax-1$, the on-site repulsion does play 
a highly nontrivial role, but one can still imagine that its effect 
is (at most) of order 1 rather than of order $L^{d}$.
This is because (in a suitable representation) there is only one 
electron with down spin, and this single electron interact with the 
rest of electrons with up spin.
This intuitive observation is indeed the basic starting point of our 
proof.

 We are also able to establish rather strong results about the excitation
energy above the ``ferromagnetic ground states''.
Let
 $\calK=\{k=2\pi n/(L-1)\,\bigl|\,\mbox{$n\in{\bf Z}$ s.t. $\abs{n}\le(L-1)/2$}\}$
be the set of wave numbers allowed in the present model.
For $k\in\calK$,
we denote by $\calH_k$ the Hilbert space of the states which have a
definite crystal momentum $k$,
and which contain $(L-1)$ up-spin electrons and one down-spin electron.
(See (\ref{Hilbk}) for a precise definition.)
We let $\ESW$
be the lowest energy among the states in $\calH_k$.
Note that $\ESW$
can be interpreted as the energy of an elementary spin-wave excitation.
The following theorem essentially determines the behavior of $\ESW$.
\begin{theorem}
[Bounds on the spin-wave energy]
Let $t>0$, and let $s$ and $V$ be determined by (\ref{nearflat}).
Assume that $\la\ge\la_3$, $\akappa\le\kappa_0$, and
$K_2\la t\ge U\ge A_3 \la^2 t\akappa$,
where $\la_3$, $\kappa_0$, $K_2$ and $A_3$ are positive constants.
Then we have 
\begin{equation}
F_2\,\frac{4U}{\la^4}\rbk{\sin\frac{k}{2}}^2
\le\ESW-\Emin(\Smax)\le 
F_1\,\frac{4U}{\la^4}\rbk{\sin\frac{k}{2}}^2,
\label{ESW1d}
\end{equation}
with 
\begin{equation}
F_1 = 1 +\frac{A_4}{\la} +A_5\la\akappa 
+ \frac{A_6\la^2t\akappa^2}{U},
\end{equation}
and 
\begin{equation}
F_2 = 1 - A_1\akappa - \frac{A_2}{\la} 
- \frac{A_3\la^2t\akappa}{U},
\end{equation}
where $A_i(i=1,\ldots,6)$
are positive constants.
\label{SW1dTh}
\end{theorem}
This is the one-dimensional version of 
Theorems~\ref{Ek<Th} and \ref{Ek>Th}.

 It is remarkable that we have $F_1\simeq F_2\simeq1$
if $\la\gg1$, $\la\akappa\ll1$
and $U\gg\la^2t\akappa$.
In this case the bounds (\ref{ESW1d}) imply
\begin{equation}
\ESW-\Emin(\Smax)
\simeq
\frac{4U}{\la^4}\rbk{\sin\frac{k}{2}}^2
=2J_{\rm eff}\rbk{\sin\frac{k}{2}}^2,
\label{ESW1d2}
\end{equation}
which is nothing but the spin-wave dispersion relation for the ferromagnetic
Heisenberg spin system.
(See Section~\ref{secper}.)
This result is very important since it guarantees that our Hubbard model
develops low-lying excited states with the precise structure expected in a
``healthy'' ferromagnetic system.

Theorem~\ref{SW1dTh} is also meaningful when applied to the 
flat-band model with $\kappa=0$.
The theorem guarantees 
that the exchange interaction $J_{\rm eff}\simeq2U/\la^4$
(which appears in (\ref{ESW1d2})) remains finite even for the flat-band
models, thus confirming the 
conjecture of Kusakabe and Aoki \cite{KusakabeAoki94b}.
We can conclude that the ferromagnetism in the flat-band models is not
at all pathological\footnote{
We recall that the Nagaoka's example of ferromagnetism is known to have a 
pathological spin-wave dispersion relation \cite{Nagaoka66,Kusakabe93}.
As for the other rigorous examples, no results about spin-wave 
excitations are known.
} in spite of the rather artificial condition
imposed on the models.

The reader may notice that Theorem~\ref{SW1dTh},
unlike Theorem~\ref{st1dTh}, requires an upper
bound for the Coulomb interaction $U$.
There indeed is a {\em physical} reason for this limitation.
Our proof of Theorem~\ref{SW1dTh} is based on an explicit construction of
the state which approximates the elementary spin-wave excitation.

Our approximate excited state, however, takes into account the
effect of interaction $U$ in a rather crude way.
This inhibits us from getting precise estimate in the models with 
larger values of $U$.
That our analysis is not efficient for large $U$ can be easily seen from our
formula for the effective exchange interaction $J_{\rm eff}=2U\la^{-4}$,
which is proportional to $U$.
For larger values of $U$, we expect $J_{\rm eff}$ to be
``renormalized'' to a less increasing function of $U$.
In particular, Kusakabe and Aoki \cite{KusakabeAoki94b}
pointed out that $J_{\rm eff}$ remains finite even in the limit
$U\up\infty$.
A proof of this fascinating conjecture might be possible if one extends
the present work by devising a more efficient approximate excited 
state which takes into account the large-$U$ ``renormalization''
in a proper manner.
\subsection{Discussions and Open Problems}
\label{secopen}
The inequality (\ref{E>E1d}) stated in our main theorem~\ref{st1dTh}
only establishes the local
stability of the ``ferromagnetic ground states'', not the desired global
stability.
However the strong result for the flat-band models summarized in
Theorem~\ref{flat1dTh} 
suggests that the local stability (\ref{E>E1d}) implies that the
``ferromagnetic ground states'' are the true ground states.
In the course of constructing our proof of the local stability
theorem, we have developed a heuristic picture about the 
mechanism underlying ferromagnetism in our model.
The picture, which is briefly described in Section~\ref{secper},
also indicates that ferromagnetic states are the true ground 
states.
As we have noted in the remark at the end of Section~\ref{secAbout}, 
this conjecture has been verified for a special class of perturbations.

It is interesting to look at our rigorous results in the light of
traditional approaches to ferromagnetism discussed in
Sections~\ref{secback} and \ref{secHubint}.
In order to guarantee the existence of
ferromagnetism in our model, we  assumed
that $\rho$ is small enough so that the band is nearly-flat,
and the Coulomb interaction $U$ is large enough.
Since a nearly-flat band has large density of states,
our requirement shares something in common with 
the Stoner criterion.
Of course there is no hope that the criterion $UD_{\rm F}>1$ gives 
reliable conditions for the range of parameters where ferromagnetism
takes place.
The improved
criterion for ferromagnetism due to Kanamori \cite{Kanamori63}
and the accompanying formula for effective $U$
do not seem to coincide with our results.

If one looks into the
proof of the theorems, however, it becomes
clear that there is a picture quite similar to that developed
by Heisenberg.
We use basis in which each electron is treated as almost localized
at each lattice site in $\Lao$.
The basic mechanism for stabilizing ferromagnetism comes from the
``exchange'' part of the interaction Hamiltonian, which is in principle
the same as what Heisenberg treated.
See also Section~\ref{secper}.

It is amusing that the ferromagnetism in our model may
be understood in terms of the above two heuristic pictures.
Usually the band electron picture and the Heisenberg's localized
electron picture of ferromagnetism
are regarded as incompatible with each other.

 All the rigorous results summarized in the previous subsection
strongly suggest that our Hubbard model
exhibits non-pathological ferromagnetism in the vicinity 
of the flat-band
models characterized by (\ref{flatcond}).
However we are far from understanding precise (necessary and sufficient)
 condition for ferromagnetism.
We believe that the one-dimensional Hubbard model with the Hamiltonian
(\ref{H1d}) at quarter-filling can be studied as a paradigm model for
itinerant electron ferromagnetism (in insulators).
To determine the region (in the three dimensional parameter space spanned
by $s$, $V$, and $U$, as well as the sign of $t$) 
where ferromagnetism takes
place is a challenging and very illuminating problem 
that can be studied by various
methods, including numerical ones.

One might regard the models with only nearest neighbor
hoppings (obtained by setting $t=0$) as ``standard''.
However the Lieb-Mattis theorem \cite{LiebMattis62} ensures\footnote{
Rigorously speaking, this is true only for the models with
open boundary conditions.
} that such models do not exhibit ferromagnetism for any values
of $V$ and $U$.
This shows that the appearance of ferromagnetism is a rather delicate
phenomenon which cannot be determined by simple criteria like
the Stoner's.

There is a perturbative argument \cite{95un} (similar to that
in Section~\ref{secper}) which suggests that the Hubbard model
with Hamiltonian (\ref{H1d}) exhibits ferromagnetism in a finite
but not very large region including the flat-band models.
Perhaps this observation is consistent with the empirical fact
that most of the known insulators appear to be
antiferromagnets\footnote{
Recently there have appeared a few organic compounds which are
insulating ferromagnets.
}.

The electron number we have chosen corresponds to the
half-filling of the lower (nearly-flat) band.
This is also the case for the general class of models studied
in the later sections.
From the standard band theoretic point of view, an electron system
with such a filling becomes metallic.
When the Coulomb interaction $U$ is sufficiently large in our
models, however, the strong correlation makes the lower band  
(effectively) fully filled.
Since the lower band is separated by an energy gap from the upper band,
the system is expected to become an insulator.
In this sense, our models provide examples of Mott-Hubbard
insulators.
This is also true for the general models in higher dimensions.

We expect to get ferromagnetic metals by lowering the electron
density in the present models.
In the flat-band case \cite{93d}, we found that the model
must be at least two dimensional in order for ferromagnetism to 
be stable against the change of electron density.
We  argued that the one-dimensional flat-band model
exhibits ferromagnetism only when the lower band is exactly 
half-filled, and exhibits paramagnetism for any lower
electron densities\footnote{
We did not give a proof of the latter statement in \cite{93d}.
But we believe there is no essential difficulty in proving
it rigorously.
}.
We believe that this dimensional dependence is a special feature of 
the flat-band models in which electrons ``cannot move'' (in some sense).

We believe that our Hubbard models with nearly flat band in 
any dimensions 
with lower electron density are one of the best candidates 
of itinerant electron systems which exhibits metallic ferromagnetism.
Unfortunately we have no rigorous results in this direction.

Finally we recall that, in dimensions one or two,
ferromagnetism in any short-ranged model with a rotation
symmetry is inevitably destroyed by infinitesimally
small thermal fluctuation \cite{Ghosh71,92d}.
In order to have ferromagnetism stable at finite temperatures,
we must treat models in (at least) three dimensions.
We expect ferromagnetism in the three dimensional versions of
our models survive at finite temperatures,
but have no rigorous results\footnote{
We recall that the existence of a ferromagnetic order in the
ferromagnetic quantum Heisenberg model at low enough
temperatures is not yet proved \cite{DysonLiebSimon78}.
It is very likely that the corresponding problem in the Hubbard
model is much harder.
}.

When one recalls the fact that we are so familiar in our daily
lives with metallic ferromagnetism stable at room temperatures,
to prove the existence of metallic ferromagnetism (say, in our models
with lower electron densities) at low enough temperatures
may appear as a modest goal.
From theoretical and mathematical points of views, however, the 
problem looks formidably difficult.
It seems that not only mathematical techniques
but fundamental understanding of ``physics'' of itinerant electron
ferromagnetism is sill lacking.

\subsection{Band Structure in the Single-Electron Problem}
\label{secband1d}
 Before going into the full many-body problem, it is useful to investigate
the corresponding single-electron problem.
The single-electron Schr\"{o}dinger equation corresponding to the Hubbard
model (\ref{H1d}) with the parameterization (\ref{nearflat}) is written
as\footnote{
See Section~\ref{secband} if it is not clear how the single-electron 
Schr\"{o}dinger equation is derived.
}
\begin{equation}
\ep\,\phi_x=
\cases{
t(\phi_{x-1}+\phi_{x+1})+\la t(\phi_{x-(1/2)}+\phi_{x+(1/2)})&
if $x\in\Lao$;\cr
(\la^2-2+\kappa)t\,\phi_x+\la t(\phi_{x-(1/2)}+\phi_{x+(1/2)})&
if $x\in\La'$,\cr
}
\label{Sch1d}
\end{equation}
where $\ep$ is the energy eigenvalue.
By using the translation invariance of the equation (\ref{Sch1d}), 
we can
write an eigenstate $(\phi_x)_{x\in\La}$
in the form of the Bloch state as $\phi_x=e^{ikx}\,v_x(k)$
with $k\in\calK$,
and $v_x(k)$ such that $v_{x+1}(k)=v_x(k)$ for any $x\in\La$.
The Schr\"{o}dinger equation in $k$-space which determines $\ep$ and $v_x(k)$
is
\begin{equation}
\ep\rbk{\matrix{v_0(k)\cr v_{1/2}(k)\cr}}=
\rbk{\matrix{2t\cos k & 2\la t\cos\frac{k}{2}\cr
2\la t\cos\frac{k}{2} & (\la^2-2+\kappa)t \cr}}
\rbk{\matrix{v_0(k)\cr v_{1/2}(k)\cr}}.
\label{Schk1d}
\end{equation}

The eigenvalue problem (\ref{Schk1d}) can be solved easily, and for each
$k\in\calK$, we find two energy eigenvalues
\begin{eqnarray}
&&
\ep_{1,2}(k)=
\ret
&&
=\frac{t}{2}\rbk{
\la^2-2+\kappa+2\cos k
\pm\sbk{\cbk{\la^2-2(1+\cos k)+\kappa}^2+4\rbk{2\la\cos\frac{k}{2}}^2}^{1/2}
},
\label{e1,2}
\end{eqnarray}
where $1$, $2$ are the band index with $1$ ({\em resp.} $2$) 
corresponding to the $-$ ({\em resp.} $+$) sign.
The energy $\ep_j(k)$,
as a function of $k$, is usually called the dispersion relation of the
$j$-th band.
When $\kappa=0$, (\ref{e1,2})
become $\ep_1(k)=-2t$
and $\ep_2(k)=\la^2t+2t\cos k$.
Note that the lower band is completely flat (dispersionless), and there is an
energy gap $\la^2t$ between the two bands as in Figure~\ref{band1d}a.
When the perturbation to the flat-band model is sufficiently small (i.e.
$\akappa t\ll\la^2t$), the lower band is nearly flat, and there remains a gap
close to $\la^2t$
as in Figure~\ref{band1d}b.
See Lemma~\ref{gapLemma}.

\begin{figure}
\centerline{\includegraphics[width=10cm,clip]{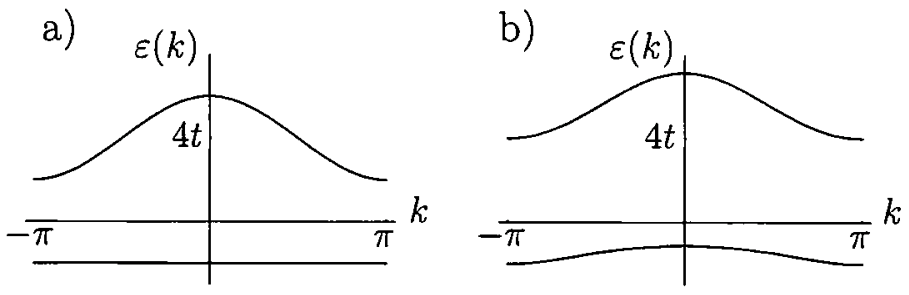}}
\caption{\captionB}
\label{band1d}
\end{figure}

We choose an eigenvector $\vvv{0}=(\vks{0}{0},\vks{0}{1/2})$
corresponding to the eigenvalue $\ep_1(k)$ as
\begin{equation}
\vvv{0}
=
\rbk{\matrix{\vks{0}{0}\cr\vks{0}{1/2}\cr}}=
\rbk{\matrix{\frac{1}{2}\cbk{F(k)+\sqrt{F(k)^2+4(2\la^{-1}\cos\frac{k}{2})^2}}\cr
-2\la^{-1}\cos\frac{k}{2}\cr}},
\label{v0k}
\end{equation}
where $F(k)=1-2(1+\cos k)/\la^2+\kappa/\la^2$.
Note that we did not normalize the vector $\vvv{0}$.
The eigenvector $\vvv{1/2}$ which corresponds to the eigenvalue
$\ep_2(k)$
can be written in terms of $\vvv{0}$ as
\begin{equation}
\vvv{1/2}
=
\rbk{\matrix{\vks{1/2}{0}\cr\vks{1/2}{1/2}\cr}}
=
\rbk{\matrix{-\vks{0}{1/2}\cr\vks{0}{0}\cr}}.
\label{v1/2k}
\end{equation}

\subsection{Localized Bases for Single-Electron States}
\label{secbases1d}
 The band structure discussed above plays a fundamental role in the
corresponding many-body problem as well.
But the $k$-space picture, which was very useful in analyzing the band
structure, turns out to be not quite effective in treating strong short-range
interactions.
This dilemma (which originates from the wave-particle dualism in quantum
mechanics) suggests the need for a new description of electronic states
which takes into account the band structure and, at the same time, treats
electrons as ``particles'' rather than ``waves''.

Let $\mu(x)=0$ for $x\in\Lao$ and $\mu(x)=1/2$ for $x\in\La'$.
For $x,y\in\La$, we define
\begin{equation}
\phis{y}{x}=(2\pi)^{-1}\int dk\,e^{ik(x-y)}\vks{\mu(y)}{\mu(x)},
\label{phiDef1d}
\end{equation}
where $\int dk(\cdots)$
is a shorthand for the sum $(2\pi/L)\sum_{k\in\calK}(\cdots)$.
Suppose that $y\in\Lao$ is fixed.
Then one can regard\footnote{
In the symbol like $\phis{y}{x}$, the upper index $y$ is the ``name'' given to
the state while the lower index $x$ is the argument in the standard 
wave function representation.
When we refer to the state itself, we write $\phi^{(y)}$.
Such a notation is used throughout the present paper.
} $\phi^{(y)}=(\phis{y}{x})_{x\in\La}$
as (the wave function of) a single-electron state which is a superposition
of the Bloch states $e^{ikx}\,\vks{0}{\mu(x)}$
with various $k$.
This means that, for any $y\in\Lao$,
the state $\phi^{(y)}$
belongs to the Hilbert space of the lower band.
By examining the definition (\ref{phiDef1d}), it follows that the collection
$\cbk{\phi^{(y)}}_{y\in\Lao}$
forms a  (nonorthonormal) basis of the Hilbert space
corresponding to the lower band.
Similarly the collection $\cbk{\phi^{(y)}}_{y\in\La'}$
forms a basis of the Hilbert space of the upper band.

Moreover the states $\phi^{(y)}=(\phis{y}{x})_{x\in\La}$
has a rather nice localization properties.
When $\kappa=0$, an explicit calculation shows that, for $y\in\Lao$,
$\phis{y}{y}=1$, $\phis{y}{x}=-1/\la$
if $\abs{x-y}=1/2$, and $\phis{y}{x}=0$ 
otherwise.
(See Section~\ref{seclocbases} where $\phis{y}{x}$ with $\kappa=0$
is denoted as $\psis{y}{x}$.)
These are the strictly localized basis states constructed and used in
\cite{92e,93d}.

For $\kappa\ne0$, the basis states $\phi^{(y)}$ are no longer strictly
localized.
Expanding the term $\sqrt{F(k)^2+4(2\la^{-1}\cos\frac{k}{2})^2}$
in (\ref{v0k})
into a power series in $\la^{-2}$ and $(\kappa/\la^2)$,
we can still prove that the state $\phi^{(y)}=(\phis{y}{x})_{x\in\La}$
is almost localized at the site $y$.
More precisely, we have 
\begin{equation}
\phis{y}{x}\simeq\cases{
1&if $x=y$;\cr
\pm1/\la&if $\abs{x-y}=1/2$;\cr
O(\akappa/\la^{2})&if $\abs{x-y}=1$;\cr
\mbox{smaller and decays exponentially}&
for $\abs{x-y}>1$,\cr
}
\label{phi1d}
\end{equation}
when $\la\gg 1$ 
and $\abs{\kappa}/\la^2 \ll 1$.
(We take $+$ sign if $y\in\La'$ and $-$ sign if $y\in\Lao$.)
This sharp localization property of the states ${\phi^{(y)}}$
plays a fundamental role throughout our proof.

Since the states $\phi^{(y)}$ with different reference sites $y$  are not
necessarily orthogonal with each other, it is useful to introduce the dual of
the basis $\cbk{\phi^{(y)}}$.
We shall construct the dual basis states 
$\phitil^{(y)}=(\phit{y}{x})_{x\in\La}$
so that 
$\sum_{x\in\La}\rbks{\phit{y}{x}}\phis{y'}{x}=\delta_{y,y'}$
holds. 
(It also holds that 
$\sum_{y\in\La}\rbks{\phit{y}{x}}\phis{y}{x'}=\delta_{x,x'}$.)
Then $\cbk{\phitil^{(y)}}_{y\in\Lao}$ 
and $\cbk{\phitil^{(y)}}_{y\in\La'}$
automatically form bases of the Hilbert spaces for the upper and
lower bands, respectively.
See Sections~\ref{seckspace} and \ref{secrealspace}
for concrete procedure for constructing $\phitil^{(y)}$
from the vectors $\vvv{0}$ and $\vvv{1/2}$.

 For $\la\gg 1$
and $\abs{\kappa}/\la^2\ll 1$,
the dual basis state $\phitil^{(y)}$
is localized at the site $y$ as 
\begin{equation}
\phit{y}{x}\simeq\cases{
1&if $x=y$;\cr
\pm1/\la&if $\abs{x-y}=1/2$;\cr
-1/\la^{2}&if $\abs{x-y}=1$;\cr
\mbox{smaller and decays exponentially}&
for $\abs{x-y}>1$,\cr
}
\label{phit1d}
\end{equation}
where the $\pm$ sign is chosen as in (\ref{phi1d}).
It should be noted that the states $\phitil^{(y)}$
are only moderately localized as compared with the sharp localization of 
$\phi^{(y)}$.
Even for $\kappa=0$,
$\phit{y}{x}$
has nonvanishing exponentially decaying tail.

\Remark
It is interesting to compare the states $\phi^{(y)}$ and $\phitil^{(y)}$
with Wannier functions \cite{Kohn73}.
Wannier functions are the standard machinery in condensed matter
physics which provide particle-like picture of electronic states by also
taking into account band structures.

 The Wannier functions $\omega^{(y)}=(\omega^{(y)}_x)_{x\in\La}$
are constructed as in (\ref{phiDef1d}),
but with the vectors $\vvv{u}$ (with $u=0,1/2$)
replaced by their normalized versions $\vvv{u}/\abs{\vvv{u}}$.
As a consequence $\cbk{\omega^{(y)}}_{y\in\Lao}$
and $\cbk{\omega^{(y)}}_{y\in\La'}$
form {\em orthonormal\/} bases of the Hilbert spaces 
for the upper and
lower bands, respectively.
As for the localization property, we have
\begin{equation}
\omega^{(y)}_{x}\simeq\cases{
1&if $x=y$;\cr
\pm1/\la&if $\abs{x-y}=1/2$;\cr
-1/(2\la^2)&if $\abs{x-y}=1$;\cr
\mbox{smaller and decays exponentially}&
for $\abs{x-y}>1$,\cr
}
\label{wdecay}
\end{equation}
which is, roughly speaking, intermediate between 
those of $\phi^{(y)}$ and $\phitil^{(y)}$.

Although the orthonormality of the Wannier basis is a clear advantage of this
machinery, the poor localization property (\ref{wdecay}) is not
optimal for our analysis of the Hubbard model.
The sharp localization (\ref{phi1d})
is so important for us that we can give up the orthonormality of bases.

It is interesting that, in the context of band calculation, Anderson
\cite{Anderson68} suggested to use non-orthonormal basis states
which are more sharply localized than the Wannier states.
One can regard our $\phi^{(y)}$ as a concrete (and typical) example of
Anderson's ultralocalized functions, used in mathematical proofs of
ferromagnetism rather than in band calculations.

In \cite{94c}, where the main results of the present paper were
first announced, we claimed that the proof of local stability of
ferromagnetism is impossible if we use the Wannier states instead
of the sharply localized states $\phi^{(y)}$.
We now feel, however, that similar proof based on the Wannier functions
may be constructed if we are careful enough in estimating various 
matrix elements.

\subsection{Representation of the Hamiltonian
in terms of the Localized Basis}
 In order to analyze many-body problems by using 
the particle-like picture
developed above, we introduce the fermion operators
\begin{equation}
\ad_{x,\sigma}=\sum_{y\in\La}\rbks{\phis{x}{y}}\cd_{y,\sigma},\quad
b_{x,\sigma}=\sum_{y\in\La}\phit{x}{y}c_{y,\sigma},
\label{abDef}
\end{equation}
for $x\in\La$ and $\sigma=\up,\dn$.
It turns out that these operators obey the
standard anticommutation relations such
as $\cbk{\ad_{x,\sigma},b_{y,\tau}}=\delta_{x,y}\delta_{\sigma,\tau}$.
This means that the ``right'' annihilation operator to be used with 
$\ad_{x,\sigma}$ is $b_{x,\sigma}$,
not $a_{x,\sigma}=(\ad_{x,\sigma})^\dagger$.

As we show in Section~\ref{secrepHam}, we can rewrite the Hamiltonian
(\ref{H1d}) in terms of these new operators as
\begin{equation}
H=
\sumtwo{x,y\in\Lao}{\sigma=\up,\dn}
\tau_{x,y}\,\ad_{x,\sigma}b_{y,\sigma}
+
\sumtwo{x,y\in\La'}{\sigma=\up,\dn}
\tau_{x,y}\,\ad_{x,\sigma}b_{y,\sigma}
+
\sum_{y,v,w,z\in\La}
\Ut_{y,v;w,z}\,
\ad_{y,\up}\ad_{v,\dn}b_{w,\dn}b_{z,\up}.
\label{Hint1d}
\end{equation}
Note that there is no hopping between $\Lao$ and $\La'$
in the hopping parts of $H$.
This is because the operators $\ad_{x,\sigma}$ and $b_{y,\sigma}$
``know'' about the band structure.
As for the properties of the effective hopping $\tau_{x,y}$,
we only need to note that $\tau_{x,x+1}=O(\akappa t)$
for $x\in\Lao$,
and $\tau_{x,x}\simeq\la^2 t$ for $x\in\La'$.

Note that the interaction term in (\ref{Hint1d}) is no longer on-site.
This fact is essential for the appearance of ferromagnetism in these 
models. 
The effective (four fermi) coupling $\Ut_{y,v;w,z}$,
in (\ref{Hint1d}) is given by
\begin{equation}
\Ut_{y,v;w,z}=U\sum_{x\in\La}
\phit{y}{x}\phit{v}{x}(\phis{w}{x}\phis{z}{x})^*.
\label{Ut1d}
\end{equation}
This expression means that the coupling 
function $\Ut_{y,v;w,z}$ is determined
by the overlap between the four states 
$\phitil^{(y)}$,  $\phitil^{(v)}$,
$\phi^{(w)}$, and $\phi^{(z)}$,
where the former two states are created and the latter two states are
annihilated.
Since each state $\phi^{(y)}$ or $\phitil^{(y)}$
is localized at the reference site $y$, 
we find that $\Ut_{y,v;w,z}$ is also short ranged.
We can say that our representation successfully took into account the
particle-like nature of electrons.
We also note that the coupling function satisfies 
the translation invariance
$\Ut_{y,v;w,z}=\Ut_{y+p,v+p;w+p,z+p}$ for any $p\in{\bf Z}$.

Let us assume $\la\gg 1$ and $\abs{\kappa}/\la^2\ll 1$.
Then we can substitute the properties (\ref{phi1d}) and (\ref{phit1d})
of the basis states into (\ref{Ut1d}) and
evaluate $\Ut_{y,v;w,z}$ explicitly as
\begin{equation}
\Ut_{0,0;0,0}\simeq
U\,\phit{0}{0}\phit{0}{0}(\phis{0}{0}\phis{0}{0})^*
\simeq U,
\label{Ueff1}
\end{equation}
\begin{equation}
\Ut_{0,1;0,1}\simeq
U\,\phit{0}{1/2}\phit{1}{1/2}(\phis{0}{1/2}\phis{1}{1/2})^*
\simeq\frac{U}{\la^4},
\label{Ueff2}
\end{equation}
\begin{equation}
\Ut_{0,1;0,0}\simeq
U\,\phit{0}{0}\phit{1}{0}(\phis{0}{0}\phis{0}{0})^*
\simeq-\frac{U}{\la^2},
\label{Ueff3}
\end{equation}
\begin{eqnarray}
\Ut_{0,0;0,1}&\simeq&
U\cbk{\phit{0}{0}\phit{0}{0}(\phis{0}{0}\phis{1}{0})^*
+\phit{0}{1/2}\phit{0}{1/2}(\phis{0}{1/2}\phis{1}{1/2})^*}
\ret
&\simeq&O\rbk{\frac{U\akappa}{\la^2}}+\frac{U}{\la^4},
\label{Ueff4}
\end{eqnarray}
and 
\begin{equation}
\Ut_{0,1;1/2,1/2}\simeq\Ut_{1/2,1/2;0,1}\simeq\frac{U}{\la^2}.
\label{Ueff5}
\end{equation}
These are the components of $\Ut$ which play important roles
when we investigate low-lying excited states of our Hubbard model.
Note that $\Ut_{0,0;0,1}$
and $\Ut_{0,1;0,0}$
are drastically different.
This asymmetry, which originates from the difference in the localization
properties (\ref{phi1d}), (\ref{phit1d}) of the states
$\phi^{(y)}$ and $\phitil{(y)}$,
is important for our proof.

\subsection{Perturbative Analysis and Effective Spin
Hamiltonian}
\label{secper}
 At this stage we shall 
develop a heuristic theory which reveals why our Hubbard
model exhibits a stable ferromagnetism.
This subsection is different from all the 
others in that it is devoted to
arguments which are not yet made rigorous.
This, however, allows us to go beyond our technical limitation, 
and discuss the
stability of ferromagnetism beyond a single-spin flip.

Here we focus on the region of parameters characterized as 
$\abs{\kappa}t\ll U\ll \la^2 t$.
Recall that $\abs{\kappa}t$, $U$, and $\la^2 t$
roughly represent the band width of the lower band, local Coulomb
interaction, and the band gap, respectively.
By examining the representation (\ref{Hint1d}) of the
Hamiltonian, we extract the most dominant part as the ``unperturbed''
Hamiltonian
\begin{equation}
H_0=\sumtwo{x,y\in\La'}{\sigma=\up,\dn}
\tau_{x,y}\,\ad_{x,\sigma}b_{y,\sigma}
+
\sum_{x\in\La}\Ut_{x,x;x,x}\,\ntil_{x,\up}\ntil_{x,\dn}.
\label{Hu}
\end{equation}
Here we introduced the pseudo number operator 
$\ntil_{x,\sigma}=\ad_{x,\sigma}b_{x,\sigma}$.
Although $\ntil_{x,\sigma}$ is not hermitian, it works exactly the same 
as the standard
number operator as long as one uses $\ad_{x,\sigma}$ and $b_{x,\sigma}$ 
as creation and annihilation operators, respectively.
By recalling $\tau_{x,x}\simeq\la^2 t$ 
for $x\in\La'$ 
and $\Ut_{x,x;x,x}\simeq U$,
we find that the conditions for minimizing $H_0$ 
are i)~there are only electrons from the lower band, i.e., those created by 
$\ad_{x,\sigma}$ with $x\in\Lao$,
and ii)~there are no doubly occupied  sites in the language  of
$\tilde{n}_{x,\sigma}$. 
Since the number of electrons $L$ is the same as the number of the sites in
$\Lao$, such states can be written as 
linear combinations of the states
\begin{equation}
\Phi_\sigma=\rbk{\prod_{x\in\Lao}\ad_{x,\sigma(x)}}\vac.
\label{Phisigma}
\end{equation}
Here the multi-index $\sigma=(\sigma(x))_{x\in\Lao}$ 
with $\sigma(x)=\up,\dn$ 
represents spin configurations.
Clearly we have $H_0\Phi_{\sigma}=0$ for any $\sigma$.
The unperturbed Hamiltonian 
$H_0$ has $2^{L}$-fold degenerate ground states.

Let us examine how the degeneracy is lifted when we consider the
remainder of the Hamiltonian,
which is 
\begin{equation}
H_{\rm pert}=\sumtwo{x,y\in\Lao}{\sigma=\up,\dn}
\tau_{x,y}\,\ad_{x,\sigma}b_{y,\sigma}
+\sumtwo{y,v,w,z\in\La}{({\rm except\ }y=v=w=z)}
\Ut_{y,v;w,z}\,\ad_{y,\up}\ad_{v,\dn}b_{w,\dn}b_{z,\up}.
\label{Hpert1d}
\end{equation}
We wish to develop a standard first order perturbation theory, 
but with using
the non-orthonormal basis consisting of the states 
$\rbk{\prod_{x\in A}\ad_{x,\up}}\rbk{\prod_{x\in B}\ad_{x,\dn}}\vac$
where $A$, $B$ are arbitrary subsets of $\La$.
Let $\Po$
be the projection operator\footnote{
The procedure for defining $\Po$ is as follows. 
Given a many-electron state $\Phi$,
one (uniquely) expands it in terms of the basis states.
Then one throws away all the basis states which are not of the form 
$\Phi_\sigma$
(\ref{Phisigma}).
The resulting state is $\Po\Phi$.
Note that $\Po$ is not an orthogonal projection.}
(defined with respect to the present basis)
onto the $2^ {L}$-dimensional ground state space spanned by
$\Phi_{\sigma}$.
The basic object in the first order degenerate perturbation is then the
effective Hamiltonian $H_{\rm eff}={\Po}H_{\rm pert}{\Po}$.
Note that $H_{\rm eff}$ is not a self-adjoint operator.
This is inevitable since we are developing a perturbation theory based on 
a non-orthonormal basis. 
Since the standard perturbation theory can be applied to 
non-hermitian matrices as well, the situation is by no means pathological.
There is a similar perturbation theory that uses orthonormal basis 
constructed from the Wannier states \cite{95un}.

Obviously a term contributing to $H_{\rm eff}$
should not affect the locations of the electrons.
As a consequence, contributions come from the so-called ``exchange'' 
terms
(and the diagonal elements of $\tau_{x,y}$) as 
\begin{equation}
H_{\rm eff}=
\cbk{E_0+
\sumtwo{x,y\in\Lao}{(x\ne y)}
\Ut_{x,y;x,y}\rbk{\ad_{x,\up}\ad_{y,\dn}b_{y,\dn}b_{x,\up}
+\ad_{y,\up}\ad_{x,\dn}b_{y,\dn}b_{x,\up}}
}\Po,
\label{Heff1}
\end{equation}
where $E_0=\sum_{x\in\Lao}\tau_{x,x}$
turns out to be the energy of the ``ferromagnetic ground states''.
(Figure~\ref{process1} illustrates how the ``exchange'' terms act
on a state.)

It turns out that the ``exchange'' term is the ultimate origin
of ferromagnetism in our Hubbard model.
In the present model, 
the ``exchange'' takes place between the spins of
two electrons in neighboring $\Lao$-sites (metallic atoms).
By recalling that there is a $\La'$-site (oxygen atom) in between
them, one might prefer to call the present process
``superexchange'' \cite{Anderson63}.
We think this terminology also possible, 
but wish to stress that 
we never get
ferromagnetism if the direct hopping between $\Lao$
sites (represented by $t$ in the Hamiltonian (\ref{H1d}))
are absent as we discussed in Section~\ref{secopen}.
We think there are much more delicate
mechanism going on here than what one would naively expects
from a ``superexchange'' process.

Let us define the pseudo spin operators by 
$\St{j}{x}=\sum_{\sigma,\tau=\up,\dn}
\ad_{x,\sigma}\,p^{(j)}_{\sigma,\tau}\,b_{x,\tau}/2$
for $j=1,2,$ and $3$ where $p^{(j)}_{\sigma,\tau}$
are the Pauli matrices (\ref{pauli}).
Again these operators are not hermitian, but work exactly 
the same as the
standard spin operators.
Then the effective Hamiltonian is rewritten as 
\begin{eqnarray}
H_{\rm eff}&=&\sbk{
E_0-
\sumtwo{x,y\in\Lao}{(x\ne y)}
\Ut_{x,y;x,y}\cbk{\rbk{\sum_{j=1}^3\St{j}{x}\St{j}{y}}-\frac{3}{4}}
}\Po
\ret
&\simeq&
\sbk{
E_0-
\frac{2U}{\la^4}\sum_{x\in\Lao}
\cbk{\rbk{\sum_{j=1}^3\St{j}{x}\St{j}{x+1}}-\frac{3}{4}}
}\Po,
\label{Heff3}
\end{eqnarray}
where we used the estimate (\ref{Ueff2}) for $\Ut$ to get
the final line.
The right-hand side of (\ref{Heff3}) is nothing but the Hamiltonian of
the nearest-neighbor Heisenberg chain with the ferromagnetic
interaction $J_{\rm eff}\simeq{2U/\la^4}$.
We have successfully derived a ferromagnetic spin system 
starting from
the Hubbard model for itinerant electrons.

If we believe in this first order perturbation theory, 
then we can conclude from
the ``spin Hamiltonian'' (\ref{Heff3}) that the ground states
of the present
Hubbard model are the ferromagnetic states given by
\begin{equation}
\UP=\rbk{\prod_{x\in\Lao}\ad_{x,\up}}\vac,
\label{Phiup1d}
\end{equation}
and its $SU(2)$ rotations.
Moreover low-energy excitations of the Hubbard model should 
coincide with those of the ferromagnetic Heisenberg model (\ref{Heff3}).
The elementary spin-wave excitation should then have the dispersion
relation
\begin{equation}
\ESW-E_0=2J_{\rm eff}\rbk{\sin\frac{k}{2}}^2
\simeq\frac{4U}{\la^4}\rbk{\sin\frac{k}{2}}^2.
\label{ESW1d3}
\end{equation}
Note that this heuristic estimate exactly coincides with our rigorous result
(\ref{ESW1d2})!

It should be stressed, however, that the above naive perturbation 
theory remains
to be justified in many aspects.
We have been neglecting so many contributions without giving any
estimates.
The most important contribution that has been neglected comes from the
second order perturbation from the hopping terms or 
the effective hopping terms
(as is illustrated in Figure~\ref{process2}).
Since such a perturbation lowers the energy of electron pairs in a
spin-singlet, it weakens the tendency towards ferromagnetism.
Fortunately, a rough estimate shows that this effect is small 
provided that $\akappa t\ll U\ll\la^2t$.

We do not argue here that the validity of the present perturbation 
theory can be established.
By comparing it with our rigorous results about local stability of
ferromagnetism and the spin-wave excitation, however, it seems rather
likely that this treatment gives sensible conclusions about low energy
properties of our Hubbard model.
In \cite{95un}, we further discuss
about the derivation of low energy
effective spin Hamiltonians in the Hubbard models.

\subsection{Sketch of the Proof}
\label{secidea}
We will now illustrate how
the theorems discussed in Section~\ref{secMM1d}
are proved.
The heart of the proof is to construct rather accurate 
trial states for 
the spin-wave excitations, and carefully examine the action of
the Hamiltonian on them.

To begin with, we note that one of the ``ferromagnetic ground states''
(defined in Section~\ref{secMM1d} as the lowest energy states within the
sector with $\Stot=\Smax$)
can be written as in (\ref{Phiup1d}).
Note that (\ref{Phiup1d}) is nothing but the state obtained by ``completely filling'' the
(single-electron) states in the lower band with up-spin
electrons.
The energy of $\UP$ is given by $E_0=\sum_{x\in\Lao} \tau_{x,x}$.

As a candidate for the spin-wave excitation, we shall consider the state
in which a single down-spin propagates  in $\Phi_{\up}$ with a momentum 
$k$ as 
\begin{equation}
\Ok=\al(k)^{-1}\sum_{x\in\Lao}e^{ikx}\,\Gamma_x,
\label{Omega1d}
\end{equation}
where we introduced
\begin{equation}
\Gamma_x=\ad_{x,\dn}b_{x,\up}\UP.
\label{Gammax}
\end{equation}
The normalization $\al(k)$ will be determined later.

Since the annihilation operator $b_{x,\up}$
properly cancels out with the creation operator $\ad_{x,\up}$, the state
$\Gamma_x$ is such that (one of) $\ad_{x,\up}$ in (\ref{Phiup1d})
is replaced with $\ad_{x,\dn}$.
Recall that, as can be seen from the definition (\ref{abDef}),
the operator
$\ad_{x,\sigma}$ creates an electron in the sharply localized state
$\phi^{(x)}$.
As the localization property (\ref{phi1d}) of $\phi^{(x)}$ shows, 
two neighboring states $\phi^{(x)}$
and $\phi^{(x+1)}$ have very small overlap (of order $1/\la^2$).
This means that the down-spin electron inserted in (\ref{Omega1d}) costs
very small energy due to the Coulomb repulsion 
$U\sum_{x\in\La}n_{x,\up}n_{x,\dn}$.
At the same time it costs small kinetic energy since it only contains
(single-electron) states from the lower band.
These observations suggests that $\Ok$ (\ref{Omega1d})
are good trial states for low-lying excitations 
in which both the kinetic energy and the Coulomb
interaction are properly taken into account.
See Figure~\ref{opict}.

\begin{figure}
\centerline{\includegraphics[width=10cm,clip]{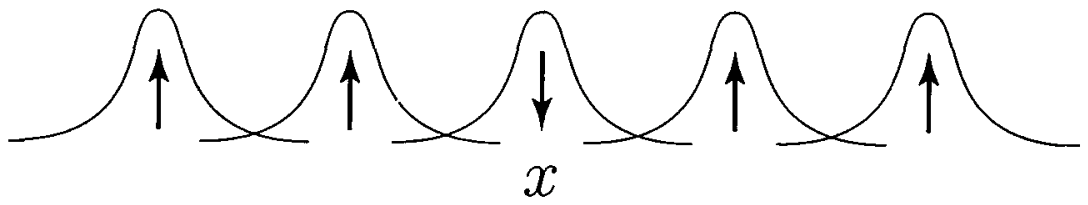}}
\caption{\captionC}
\label{opict}
\end{figure}

To prove the upper bound for the spin-wave dispersion 
relation in (\ref{ESW1d}),
we employ the standard variational inequality (see (\ref{Evar})), 
and calculate
the expectation value of $H$ in the state $\Ok$.
We then find that the main contribution comes from the ``exchange'' Hamiltonian 
(\ref{Heff3}), which leads us to the desired upper bound.
See Section~\ref{SecSW<} for details.

To further investigate the accuracy of the trial state,  
and to get the lower
band in (\ref{ESW1d}), we apply the Hamiltonian onto  $\Ok$.
Although there can appear enormous number of terms, the major
contributions\footnote{
In the later sections, we of course control all the 
possible contributions.
} come from three basic short range processes 
that we now describe.
The three processes are represented by the following three 
operators (\ref{H1})-(\ref{H3}) which are extracted from the
Hamiltonian in the form (\ref{Hint1d}).
The first process is the nearest neighbor ``exchange'' discussed in
Section~\ref{secper}, which is represented by 
\begin{equation}
H_1=\sumtwo{x\in\Lao}{\sigma=\up,\dn}\Ut_{x,x+1;x,x+1}
(\ad_{x,\sigma}\ad_{x+1,-\sigma}b_{x+1,-\sigma}b_{x,\sigma}
+\ad_{x+1,\sigma}\ad_{x,-\sigma}b_{x+1,-\sigma}b_{x,\sigma}).
\label{H1}
\end{equation}
The second process is the nearest neighbor hopping\footnote{
The first operator annihilates an electron at $x+r$ with 
spin $-\sigma$, and then creates the same thing.
Therefore its action is the same as the second operator
$\ad_{x+r,\sigma}b_{x,\sigma}$ provided that there is a
spin $-\sigma$ electron at $x+r$.
}
represented by
\begin{equation}
H_2=\sumthree{x\in\Lao}{r=\pm1}{\sigma=\up,\dn}
\cbk{
\Ut_{x+r,x+r;x+r,x}\,
\ad_{x+r,\sigma}\ad_{x+r,-\sigma}
b_{x+r,-\sigma}b_{x,\sigma}
+
\tau_{x+r,x}\,\ad_{x+r,\sigma}b_{x,\sigma}
}.
\label{H2}
\end{equation}
The third process is represented by
\begin{equation}
H_3=\sumtwo{x\in\Lao}{\sigma=\up,\dn}
\Ut_{x+(1/2),x+(1/2);x,x+1}\,
\ad_{x+(1/2),\sigma}\ad_{x+(1/2),-\sigma}
b_{x,-\sigma}b_{x+1,\sigma},
\label{H3}
\end{equation}
which lets two electrons in neighboring $\Lao$
sites $x$, $x+1$ hop simultaneously to the site in between them.
Note that $\ad_{x+(1/2),\sigma}$ creates an electron in the upper band.

Let us investigate the action of 
these partial Hamiltonians (\ref{H1})-(\ref{H3}) onto the
state $\Ok$.
It is useful to first consider the action on the state
$\Gamma_y$ defined in (\ref{Gammax}) which contains a down-spin
electron at site $y$.
By operating the ``exchange'' Hamiltonian (\ref{H1}) onto 
$\Gamma_y$, we find
\begin{eqnarray}
H_1\Gamma_y
&=&
\Ut_{y-1,y;y-1,y}(\Gamma_y-\Gamma_{y-1})+
\Ut_{y,y+1;y,y+1}(\Gamma_y-\Gamma_{y+1})
\ret
&=&
\Ut_{0,1;0,1}(2\Gamma_y-\Gamma_{y-1}-\Gamma_{y+1}),
\label{H1Phix}
\end{eqnarray}
where the minus signs come from fermion ordering for the ``exchanged''
configurations.
Figure~\ref{process1} illustrates how these four terms arise.
We also used the translation invariance of $\Ut$, which is indeed essential
for the present proof.
Recalling (\ref{Omega1d}), we get the expected result
\begin{equation}
H_1\,\Ok=\cbk{E_0+4\,\Ut_{0,1;0,1}\rbk{\sin\frac{k}{2}}^2}\Ok.
\label{H1O}
\end{equation}

\begin{figure}
\centerline{\includegraphics[width=11cm,clip]{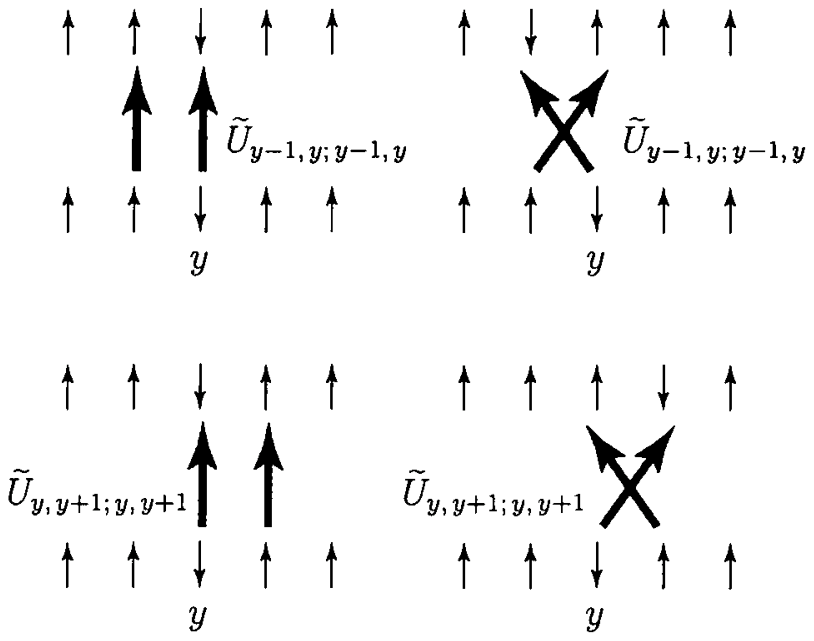}}
\caption{\captionD}
\label{process1}
\end{figure}

To see the role of $H_2$ (\ref{H2}), let us operate it onto
$\Gamma_y$ to get
\begin{eqnarray}
H_2\,\Gamma_y
&=&
\sum_{r=\pm1}\cbk{
(\Ut_{y+r,y+r;y+r,y}+\tau_{y+r,y})\Gammat{r}{y+r}
-
(\Ut_{y,y;y,y-r}+\tau_{y,y-r})\Gammat{r}{y}
}
\ret
&=&
(\Ut_{r,r;r,0}+\tau_{r,0})\rbk{\Gammat{r}{y+r}-\Gammat{r}{y}},
\label{H2Py}
\end{eqnarray}
where $\Gammat{r}{y}=\ad_{y,\dn}b_{y-r,\up}\UP$ is the state 
with an empty site
$y-r$ and a doubly occupied site $y$.
We also used the translation invariance of $\Ut$ and $\tau$.
In Figure~\ref{process2}, we illustrate the action of a part of
$H_2$ onto $\Gamma_y$.
Again from (\ref{Omega1d}), we find that
\begin{eqnarray}
H_2\,\Ok
&=&
\sum_{r=\pm1}\frac{e^{-ikr}-1}{\al(k)}\,
(\Ut_{r,r;r,0}+\tau_{r,0})\,\Xi_r(k)
\ret
&=&
\sum_{r=\pm1}\frac{e^{-ikr}-1}{\al(k)}\,
(\Ut_{1,1;1,0}+\tau_{1,0})\,\Xi_r(k),
\label{H2O}
\end{eqnarray}
where we used the reflection symmetry\footnote{
Such a symmetry exists in the present model.
In the general class of models studied later, we do not
assume reflection or rotation symmetries.
} to get the final line.
Here
\begin{equation}
\Xi_r(k)=\sum_{x\in\Lao}e^{ikx}\,\Gammat{r}{x}
=\sum_{x\in\Lao}e^{ikx}\,\ad_{x,\dn}b_{x-r,\up}\UP,
\label{Xirk}
\end{equation}
is the state in which a bound pair of an empty site $x-r$ and a
doubly occupied site $x$ is propagating with momentum $k$.
We shall abbreviate $\Xi_{\pm1}(k)$ as $\Xi_{\pm}(k)$.

\begin{figure}
\centerline{\includegraphics[width=11cm,clip]{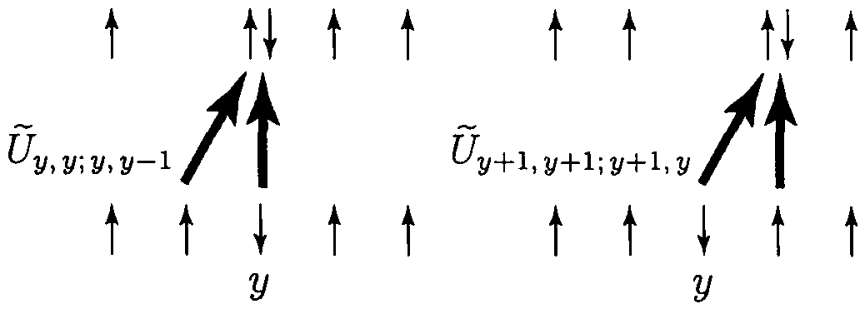}}
\caption{\captionE}
\label{process2}
\end{figure}

Similarly we obtain 
\begin{equation}
H_3\,\Ok=\frac{e^{ik}-1}{\al(k)}\,\Ut_{1/2,1/2;0,1}\,\Theta(k),
\label{H3O}
\end{equation}
where 
\begin{equation}
\Theta(k)=\sum_{x\in\Lao}e^{ikx}\,
\ad_{x+(1/2),\dn}\ad_{x+(1/2),\up}
b_{x,\up}b_{x+1,\up}\UP,
\label{Thetak}
\end{equation}
is the state in which two adjacent empty sites in $\Lao$ and a doubly 
occupied site
in between them are forming a bound state and propagating with 
momentum $k$.
See Figure~\ref{3states} for schematic pictures 
of the states\footnote{
In the general notation used in the latter sections, the states
$\Xi_r(k)$ and $\Theta(k)$ are denoted as
$\Phik{0,r}$ and $\Phik{1/2,1/2,0,1}$, respectively.
See Section~\ref{secmatrix}.
}
$\Ok$, $\Xi_+(k)$, and $\Theta(k)$.
We also note here that these states all belong to 
the Hilbert space $\calH_k$, which we introduced
just above Theorem~\ref{SW1dTh}.

The relations (\ref{H2O}) and (\ref{H3O}) clearly show 
that our trial state
$\Ok$ cannot be the exact eigenstate of the Hamiltonian.
To investigate low-lying spectrum of $H$, 
we have to consider (at least)
the subspace spanned by the states 
$\Ok$, $\Xi_\pm(k)$, and $\Theta(k)$.
As before we calculate the action of $H$ 
onto the latter states to find
\begin{eqnarray}
H\,\Xi_\pm(k)&\simeq&(E_0+\Ut_{0,0;0,0})\,\Xi_\pm(k)
+\al(k)\,(\Ut_{0,1;1,1}+\tau_{0,1})(e^{\mp ik}-1)\,\Ok
\ret
&&+(\mbox{other states}),
\label{HXi}
\end{eqnarray}
and
\begin{eqnarray}
H\Theta(k)&\simeq&
(E_0-2\,\tau_{0,0}+2\,\tau_{1/2,1/2}+\Ut_{1/2,1/2;1/2,1/2})\Theta(k)
\ret
&&+\al(k)(e^{-ik}-1)\Ut_{0,1;1/2,1/2}\,\Ok+(\mbox{other states}).
\label{HTheta}
\end{eqnarray}
Although it might not be clear at this stage, 
it turns out that the ``other states''
in (\ref{HXi}) and (\ref{HTheta}) do not play essential roles.
We leave such estimates (as well as the precise definition of the 
``other states'') to the latter sections, and simply neglect them here.

The equations (\ref{H1O}), (\ref{H2O}), (\ref{H3O}), (\ref{HXi}), 
and (\ref{HTheta})
provide, for each $k\in\calK$, the representation of the 
Hamiltonian in the four
dimensional subspace of $\calH_k$ spanned by the states
$\Ok$, $\Xi_\pm(k)$, and $\Theta(k)$.
We now read off the matrix elements\footnote{
As usual, matrix elements $h[\Psi,\Phi]$ are defined by the unique
expansion $H\Phi=\sum_\Psi h[\Psi,\Phi]\Psi$.
See also (\ref{matrixelement}).
}
from these equations, and then 
use the estimates (\ref{Ueff1})-(\ref{Ueff5}) 
of $\Ut$ to evaluate them as
\begin{equation}
h[\Ok,\Ok]
\simeq
E_0+4\,\Ut_{0,1;0,1}\rbk{\sin\frac{k}{2}}^2
\ret
\simeq
E_0+4\frac{U}{\la^4}\rbk{\sin\frac{k}{2}}^2,
\label{hOO1d}
\end{equation}
\begin{equation}
h[\Xi_\pm(k),\Xi_\pm(k)]\simeq E_0+\Ut_{0,0;0,0}\simeq E_0+U,
\label{hXX}
\end{equation}
\begin{equation}
h[\Theta(k),\Theta(k)]
\simeq
E_0-2\,\tau_{0,0}+2\,\tau_{1/2,1/2}+\Ut_{1/2,1/2;1/2,1/2}
\ret
\simeq
E_0+2\la^2t+U,
\label{hTT}
\end{equation}
\begin{eqnarray}
h[\Xi_\pm(k),\Ok]&\simeq&
\al(k)^{-1}(\Ut_{1,1;1,0}+\tau_{1,0})(e^{\pm ik}-1)
\ret
&\simeq&
\al(k)^{-1}\rbk{\frac{cU\kappa}{\la^2}+\frac{U}{\la^4}+c'\kappa t}
(e^{\pm ik}-1),
\label{hXO}
\end{eqnarray}
\begin{eqnarray}
h[\Ok,\Xi_\pm(k)]&\simeq&
\al(k)(\Ut_{0,1;1,1}+\tau_{0,1})(e^{\mp ik}-1)
\ret
&\simeq&\al(k)\rbk{-\frac{U}{\la^2}+c'\kappa t}(e^{\mp ik}-1),
\label{hOX}
\end{eqnarray}
\begin{equation}
h[\Theta(k),\Ok]\simeq\al(k)^{-1}\,\Ut_{1/2,1/2;0,1}(e^{ik}-1)
\simeq\al(k)^{-1}\frac{U}{\la^2}(e^{ik}-1),
\label{hTO}
\end{equation}
and
\begin{equation}
h[\Theta(k),\Ok]\simeq\al(k)\,\Ut_{0,1;1/2,1/2}(e^{-ik}-1)
\simeq\al(k)\frac{U}{\la^2}(e^{-ik}-1),
\label{hOT}
\end{equation}
where the approximate values are obtained for
$\la\gg1$ and $\akappa/\la^2\ll1$,
and $c$, $c'$ are constants.
Reflecting the use of the non-orthonormal basis, these matrix elements
are not symmetric.
In particular drastic difference between the elements 
$h[\Xi_\pm(k),\Ok]$ (\ref{hXO}) and $h[\Ok,\Xi_\pm(k)]$ (\ref{hOX})
plays fundamental role in our proof.
Figure~\ref{3states} shows  
these matrix elements.

\begin{figure}
\centerline{\includegraphics[width=11cm,clip]{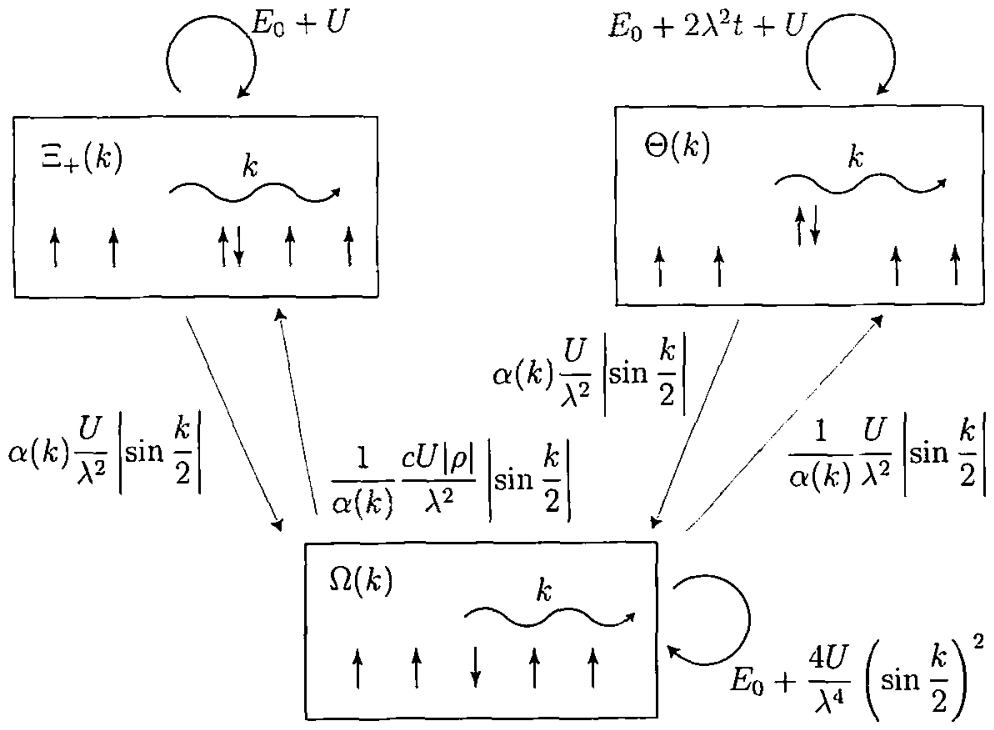}}
\caption{\captionF}
\label{3states}
\end{figure}

In order to bound the excitation energy from below, 
we use the following
wellknown fact about the minimum eigenvalue of a matrix.
Let $(h_{i,j})_{i,j=1,\ldots,N}$ be an $N\times N$ matrix with real
eigenvalues.
Then the lowest eigenvalue $h_0$ of the matrix satisfies
$h_0\ge\min_{i=1,\ldots,N}D_i$ with 
$D_i={\rm Re}[h_{i,i}]-\sum_{j\ne i}|h_{i,j}|$.
This is almost trivial, but see Lemma~\ref{DPhiLemma} for a proof.
We stress that this simple-minded inequality is expected to yield 
physically meaningful results only when one uses a basis which ``almost 
diagonalizes'' the low energy part of the Hamiltonian.

We now apply this inequality to the $4\times4$ matrix representation
of $H$ in each sector with a fixed $k\in\calK$.
The quantities corresponding to $D_i$ are evaluated for each state as
\begin{eqnarray}
&&
D[\Ok]
=
h[\Ok,\Ok]-\sum_\pm\abs{h[\Ok,\Xi_\pm(k)]}
-\abs{h[\Ok,\Theta(k)]}
\ret
&&
\ge
E_0+\frac{4U}{\la^4}\rbk{\sin\frac{k}{2}}^2
-4\al(k)\rbk{\frac{U}{\la^2}+c'\akappa t}\abs{\sin\frac{k}{2}}
-2\al(k)\frac{U}{\la^2}\abs{\sin\frac{k}{2}},
\label{DOmega}
\end{eqnarray}
\begin{eqnarray}
D[\Xi_\pm(k)]&=&h[\Xi_\pm(k),\Xi_\pm(k)]-\abs{h[\Xi_\pm(k),\Ok]}
\ret
&\ge&
E_0+U-2\al(k)^{-1}\rbk{\frac{cU\akappa}{\la^2}+\frac{U}{\la^4}+c'\akappa t}
\abs{\sin\frac{k}{2}},
\label{DXi}
\end{eqnarray}
and
\begin{eqnarray}
D[\Theta(k)]&=&h[\Theta(k),\Theta(k)]-\abs{h[\Theta(k),\Ok]}
\ret
&\ge&
E_0+2\la^2t+U-2\al(k)^{-1}\frac{U}{\la^2}\abs{\sin\frac{k}{2}},
\label{DTheta}
\end{eqnarray}
where we used the estimates (\ref{hOO1d})-(\ref{hOT}) to get
the lower bounds.

At this stage, we choose\footnote{
The choice of $\al(k)$ here is different from that in 
the full proof in
Section~\ref{SecProof}.
(See (\ref{alpha}).)
This is because the actual estimate of $D[\Xi_\pm(k)]$
in the latter sections take into account various 
small terms
which are simply neglected here.
}
the constant $\al(k)$ as
\begin{equation}
\al(k)=4\rbk{\frac{c\akappa}{\la^2}+\frac{1}{\la^4}+\frac{c'\akappa t}{U}}
\abs{\sin\frac{k}{2}}.
\label{alpha1d}
\end{equation}
This choice makes the bound (\ref{DXi}) into
\begin{equation}
D[\Xi_\pm(k)]\ge E_0+\frac{U}{2},
\label{DXi>}
\end{equation}
and the bound (\ref{DTheta}) into
\begin{eqnarray}
D[\Theta(k)]&\ge&E_0+2\la^2t+U-\frac{U}{2\la^2}
\rbk{\frac{c\akappa}{\la^2}+\frac{1}{\la^4}+\frac{c'\akappa t}{U}}^{-1}
\ret
&\ge&E_0+2\la^2t+U-\frac{\la^2U}{2}
\ret
&\ge&E_0+U,
\label{DTheta>}
\end{eqnarray}
where we have further assumed\footnote{
The upper bound required for $U$ depends sensitively on 
the choice of $\al(k)$.
The requirement $U\le K_2\la t$ made in Theorem~\ref{SW1dTh} 
(and which appears in the full proof)
is somewhat different from the present one.
}  $U\le4t$.
Finally we substitute (\ref{alpha1d}) into the bound 
(\ref{DOmega}) to get
\begin{eqnarray}
D[\Ok]&\ge&E_0+\frac{4U}{\la^4}\rbk{\sin\frac{k}{2}}^2
\ret
&&
-4\rbk{\frac{c\akappa}{\la^2}+\frac{1}{\la^4}
+\frac{c'\akappa t}{U}}
\rbk{\frac{6U}{\la^2}+4c'\akappa t}\rbk{\sin\frac{k}{2}}^2
\ret
&=&
E_0+\frac{4U}{\la^4}\cbk{1-
\rbk{c\akappa+\frac{1}{\la^2}+\frac{c'\akappa\la^2t}{U}}
\rbk{6+\frac{4c'\akappa\la^2t}{U}}
}
\rbk{\sin\frac{k}{2}}^2
\ret
&\ge&
E_0+\frac{4U}{\la^4}\cbk{
1-A_1\akappa-\frac{A_2}{\la}-\frac{A_3\la^2t\akappa}{U}
}
\rbk{\sin\frac{k}{2}}^2,
\label{DOmega>}
\end{eqnarray}
with constants\footnote{
Here the term $A_2/\la$ in the right-hand side of (\ref{DOmega>})
can be replaced with $A_2/\la^2$
if we simply equate the above expression.
Since the actual matrix elements have many ``small'' terms
that are neglected here, what we can prove (in the latter sections
with perfect
rigor) is the bound in terms of the quantity in the
right-hand side of (\ref{DOmega>}).
} $A_1$, $A_2$, and $A_3$.
Since the lower bounds (\ref{DXi}) and (\ref{DTheta}) for 
$D[\Xi_\pm(k)]$ and $D[\Theta(k)]$ are strictly larger than that
for $D[\Ok]$, we find that the right-hand side of (\ref{DOmega>}) gives
the desired lower bound for the lowest eigenvalue of the Hamiltonian
in the space $\calH_k$.
Therefore the lower bound for the spin wave excitation in (\ref{ESW1d})
(which is the main statement of Theorem~\ref{SW1dTh})
has been derived.

In the remainder, we sketch how we get Theorem~\ref{st1dTh}
about the local stability of ferromagnetism
from the above lower bounds.
The lower bounds in (\ref{ESW1d}) gives strict bounds
$\ESW>E_0$ for all $k\in\calK$ except for $k=0$.
This means that the desired local stability inequality 
(\ref{E>E1d}) has been 
proved except in the translation invariant sector with $k=0$.
To deal with the $k=0$ sector is easy once we realize that
$\Omega(0)$ is nothing but an $SU(2)$ rotation of the 
``ferromagnetic ground state'' $\UP$.
By simply repeating the above arguments for the {\em three} dimensional
subspace spanned by $\Xi_\pm(0)$ and $\Theta(0)$, we easily find that
the desired bound  (\ref{E>E1d}) also holds in the sector with $k=0$.
It only remains to extend the parameter region in which
the statement is valid.
This is easily done by a general consideration about the monotonicity of
energies as a function of $U$.
See Section~\ref{secproofstab}.


\Section{Definitions and Main Results}
\label{SecDef}
In the present section, we define the general class of models treated in 
the present paper, and precisely state our main theorems.
\subsection{Lattice}
We describe the lattice on which our Hubbard model is defined. 
The lattice
is characterized by the dimension of the lattice $d=1,2,3,\cdots$, the
dimension of ``cells'' $\nu=1,2,\cdots,d$, and the linear size $L$ which is
taken to be an odd integer. 
Throughout the present paper we
assume that the three parameters $d$, $\nu$, and $L$ are fixed to allowed
values.
All the bounds proved in the present paper are independent of the system 
size $L$.
 
Let $\Lao$ be the $d$-dimensional $L\times\ldots\times L$ hypercubic
lattice with periodic boundary conditions; 
\begin{equation}
\Lao = 
\set
{x = (x_1,\ldots, x_d)}
{\mbox{$x_i\in{\bf Z}$, $\abs{x_i}\le(L-1)/2$ for $i=1,\ldots,d$}}.
\label{Lambdao}
\end{equation}
We  ``decorate'' the lattice $\Lao$ by adding sites taken at the
center   of   each $\nu$-dimensional cell. 
Let $\calU'$ be the set of vectors defined as
\begin{equation}
\calU'
=
\set
{u=(u_1,\ldots,u_d)}
{\mbox{$u_i=0$ or $1/2$, and $2\sum_{i=1}^du_i=\nu$}}.
\label{Uprime}
\end{equation}
Note that each $u\in\calU'$ has the length $\abs{u}=\sqrt{\nu}/2$.
For each $u\in\calU'$, we let
\begin{equation}
\Lau = \set{x+u}{x\in\Lao}.
\label{Lambdau}
\end{equation}
By introducing the unit cell $\calU$ of the lattice by 
\begin{equation}
\calU = \cbk{o}\cup\calU',
\label{U}
\end{equation}
where $o=(0,\cdots,0)$, our decorated hypercubic lattice is defined as
\begin{equation}
\La = \bigcup_{u\in\calU}\Lau.
\label{Lambda}
\end{equation}
We often decompose $\La$ as $\La=\Lao\cup\La'$
where
\begin{equation}
\La' = \bigcup_{u\in\calU'}\Lau.
\label{Lambdaprime}
\end{equation}
As is discussed in Section~\ref{secMM1d}, we can imagine that sites in 
$\Lao$ represent metallic atoms and sites in $\La'$ represent oxygen atoms.
The numbers of sites (vectors) in the unit cell\footnote{
Throughout the present paper, $\abs{S}$ denotes the number of elements in 
a set $S$.
}
\begin{equation}
b = \abs{\calU} = \dnu+1
\label{bandnumber}
\end{equation}
is important, since it gives the number of the bands in the
corresponding single-electron problem.

For $d=1$, the only possible choice of $\nu$ is $\nu=1$, 
and we get the chain with two kinds of
atoms discussed in Section~\ref{Sec1d}.
(See Figure~\ref{lattice1d}.)
For $d=2$, we can either set $\nu=1$ to get the lattice  
in Figure~\ref{lattice2d}a
with the band number $b=3$, or
set $\nu=2$ to get the lattice in Figure~\ref{lattice2d}b with $b=2$.
For $d=3$, there are three choices for $\nu$.
The lattices with $\nu=2$ and $\nu=3$ have the structures similar to 
the fcc and the bcc lattices, respectively.

\begin{figure}
\centerline{\includegraphics[width=7cm,clip]{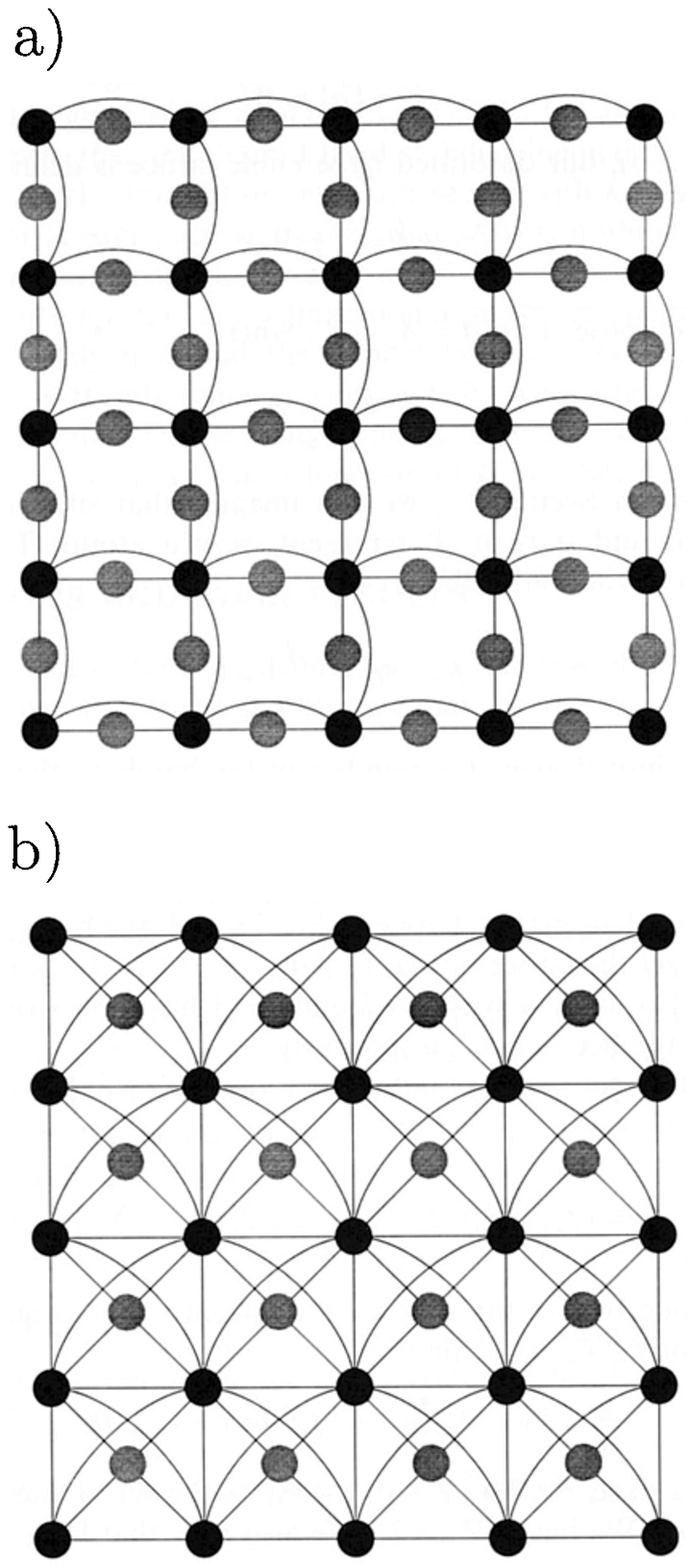}}
\caption{\captionG}
\label{lattice2d}
\end{figure}

We introduce some sets of lattice vectors which will become useful.
We define
\begin{equation}
\calF_o
=\set
{f=(f_1,\ldots,f_d)}
{\mbox{$f_i=0$ or $\pm1/2$, and $2\sum_{i=1}^d|f_i|=\nu$}},
\label{Fo}
\end{equation}
which is the collection of the sites in $\La'$ adjacent  to the
origin $o$.
We have 
$\abs{\calF_o}=2^\nu\dnu$.
For 
$f\in\calF_o$,
we define 
\begin{equation}
\calF_f = 
\set{g\in\calF_o}{\mbox{$\abs{g_i}=\abs{f_i}$ for $i=1,\ldots,d$}}.
\label{Ff}
\end{equation}
Note that for a fixed $f\in\calF_o$,
$\set{f+g}{g\in\calF_f}$
is the set of sites in $\Lao$ which are adjacent  to $f$.
We have 
$\abs{\calF_f}=2^\nu$.
We also note that for $g\in\calF_f$, we have 
\begin{equation}
\calF_g=\calF_f.
\label{Fg=Ff}
\end{equation}

\subsection{Hubbard Model}
\label{secHub}
We define the Hubbard model on the decorated hypercubic lattice $\La$.
As usual, we denote by $\cxs$ and $\axs$ the creation and the annihilation 
operators, respectively, of an electron at site $x\in\La$
with spin $\sigma=\up$, $\dn$.
These operators satisfy the standard anticommutation relations
\begin{equation}
\cbk{\axs,\ayt} = \cbk{\cxs,\cyt} = 0,
\label{ac1}
\end{equation}
and
\begin{equation}
\cbk{\axs,\cyt} = \delta_{x,y}\,\delta_{\sigma,\tau},
\label{ac2}
\end{equation}
for any $x,y\in\La$ and $\sigma,\tau=\up,\dn$,
where $\cbk{A,B}=AB+BA$.
The number operator for an electron at site $x$
with spin $\sigma$ is defined as
\begin{equation}
\nxs = \cxs\axs.
\label{nxs}
\end{equation}

We consider the standard Hubbard Hamiltonian
\begin{equation}
H = \Hhop + \Hint.
\label{Ham1}
\end{equation}
The interaction Hamiltonian is
\begin{equation}
\Hint = U \sum_{x\in\La} n_{x,\up} n_{x,\dn},
\label{Hint}
\end{equation}
where $U>0$ is the on-site Coulomb repulsion energy.
The hopping Hamiltonian is further decomposed as
\begin{equation}
\Hhop = \Hhopo + \kappa \Hpert,
\label{Hhop}
\end{equation}
where $\Hhopo$ is the hopping Hamiltonian of the flat-band model
defined as
\begin{equation}
\Hhopo
=
t
\sum_{\sigma=\up,\dn}
\sum_{x\in\La'}
\rbk{\la\, \cxs + \sumtwo{y\in\Lao}{|x-y|=\sqrt{\nu}/2}\cys}
\rbk{\la\, \axs + \sumtwo{y\in\Lao}{|x-y|=\sqrt{\nu}/2}\ays},
\label{Hhop01}
\end{equation}
where $t>0$ and $\la>0$ are parameters.
It is, of course, possible to represent the Hamiltonian
in the ``standard'' form as
\begin{equation}
\Hhopo
=
\sum_{\sigma=\up,\dn}
\sum_{x,y\in\La}
t_{x,y}^{(0)} \,\cxs\ays,
\label{Hhop02}
\end{equation}
where the hopping matrix elements are given by
\begin{equation}
t_{x,y}^{(0)} = t_{y,x}^{(0)} = \cases{
\la^2 t &if $x=y\in\La'$;\cr
\la t & 
if $x\in\Lao$, $y\in\La'$ with $\abs{x-y}=\sqrt{\nu}/2$;\cr
2^{(\nu-\mu)}{d-\mu \choose \nu-\mu}t &
if $x,y\in\Lao$ with $\abs{x-y}=\sqrt{\mu}$
where $\mu =0,1,\ldots,\nu$;\cr
0 & otherwise.\cr
}
\label{txy0}
\end{equation}
The representation (\ref{Hhop01}) shows that the
hopping Hamiltonian $\Hhopo$ is characterized by mean-field-like
hoppings within each $\nu$-dimensional cell which consists of
$x\in\La'$ and the sites $y\in\Lao$ adjacent to it.
This rather artificial choice of the hopping produces the single-electron
spectrum with a completely flat band.
See Section~\ref{secband}.

The perturbation Hamiltonian, on the other hand, is rather arbitrary.
The magnitude of the perturbation is controlled by the real parameter
$\kappa$.
The Hamiltonian $\Hpert$ has the standard form
\begin{equation}
\Hpert 
=
\sum_{\sigma=\up,\dn}
\sum_{x,y\in\La}
t'_{x,y}\,\cxs\ays.
\label{Hpert}
\end{equation}
The hopping matrix elements $t'_{x,y}=t'_{y,x}\in{\bf R}$ 
are arbitrary except for the following conditions.
We require the translation invariance
\begin{equation}
t'_{x,y}=t'_{x+z,y+z}
\label{trans}
\end{equation}
for any $z\in\Zd$ and any $x,y\in\La$,
and the summability
\begin{equation}
\sum_{y\in\La}\abs{t'_{x,y}}\le t,
\label{tp<t}
\end{equation}
\begin{equation}
\sum_{y\in\La}\abs{x-y}\abs{t'_{x,y}}\le t\, R,
\label{txp<tR}
\end{equation}
for any $x\in\La$.
Here $t$ is the same as before, and $R$ is a constant which measures
the range of the hopping $\cbk{t'_{x,y}}$.
When $R$ chosen to optimize (\ref{txp<tR}) is less than $\sqrt{\nu}/2$,
we redefine it\footnote{
This is done for a purely technical reason to make some formulas simple.
See (\ref{sumxI<}).
} as $R=\sqrt{\nu}/2$.

The Hilbert space of the model is spanned by the basis states
\begin{equation}
\Phi_{A,B} = 
\rbk{\prod_{y\in A}c^\dagger_{y,\up}}
\rbk{\prod_{z\in B}c^\dagger_{z,\dn}}
\vac,
\label{generalbasis}
\end{equation}
where $A$, $B$ are subsets of $\La$, and
$\vac$ is the unique vacuum state characterized by
$\axs\vac=0$ for any $x\in\La$ and $\sigma=\up,\dn$.
In the preset work, we only consider the Hilbert space $\calH$
with the electron number fixed to $L^d=\abs{\Lao}$,
which is spanned by the basis states (\ref{generalbasis})
with $\abs{A}+\abs{B}=L^d$.

\subsection{Local Stability Theorem}
\label{secmainres}
The total spin operator of the Hubbard model is defined as 
usual by
\begin{equation}
S^{(\alpha)}_{\rm tot} = 
\frac{1}{2}\sum_{x\in\La}\sum_{\sigma,\tau=\up,\dn}
\cxs(p^{(\alpha)})_{\sigma,\tau}\,c_{x,\tau}
\label{Stot}
\end{equation}
for $\alpha=1,2,3$, where $p^{(\alpha)}$ are the Pauli matrices
\begin{equation}
p^{(1)}=\rbk{\matrix{0&1\cr 1&0\cr}},\quad
p^{(2)}=\rbk{\matrix{0&-i\cr i&0\cr}},\quad
p^{(3)}=\rbk{\matrix{1&0\cr 0&-1\cr}}.
\label{pauli}
\end{equation}

An explicit calculation shows that the spin operators
$({\bf S}_{\rm tot})^2=\sum_{\alpha=1,2,3}(S^{(\alpha)}_{\rm tot})^2$,
$\Sztot$, and the Hamiltonian $H$ commute with each other.
This means that we can find simultaneous eigenstates of these operators.
The eigenvalue of $({\bf S}_{\rm tot})^2$ is denoted as 
$\Stot(\Stot+1)$ where $\Stot$ can take values 
$1/2, 3/2, \ldots, \Smax$ with $\Smax=L^d/2$.

We are now able to state the theorem due to Tasaki \cite{92e}
and Mielke-Tasaki \cite{93d}.
\begin{theorem}
[Flat-band ferromagnetism]
Consider the Hubbard model with the Hamiltonian (\ref{Ham1}).
Assume $t>0$, $\la>0$, and $\kappa=0$ to get a model with
a flat band.
Then, for any $U>0$, the ground states of the Hamiltonian $H$
has $\Stot=\Smax$, and are non-degenerate apart from the 
trivial $(2\Smax+1)$-fold spin degeneracy.
\label{flatbandTh}
\end{theorem}
This theorem is desirable in the sense that it completely determines
the ground states of the model.
But the result is not robust since it applies only to the models with
a completely flat band.
Since the references \cite{92e,93d} only discusses the models with 
$\nu=1$, we will prove the theorem in Section~\ref{secflatproof}.

We now describe the new {\em robust} results for 
the models with a nearly flat band.
For $\Stot = 1/2, 3/2, \ldots, \Smax$, we denote by
$\Emin(\Stot)$ the lowest eigenvalue of the Hamiltonian
(\ref{Ham1}) in the sector which consists of the states 
$\Phi$ such that $({\bf S}_{\rm tot})^2\Phi=\Stot(\Stot+1)\Phi$.
Then we have the following simple lemma for the 
sector with $\Stot=\Smax$.
\begin{lemma}
[``Ferromagnetic ground states'']
Assume that $t>0$, $\la\ge\la_1$,
and 
\newline
$|\kappa|\la^{-2}\le r_1$,
where $\la_1$ and $r_1$ are finite constants which depend only on
the dimensions $d$ and $\nu$.
(See Lemma~\ref{gapLemma} for explicit formulas of 
$\la_1$ and $r_1$.)
Then for arbitrary $U$, the states $\Phi$
such that 
$({\bf S}_{\rm tot})^2\Phi=\Smax(\Smax+1)\Phi$
and
$H\Phi=\Emin(\Smax)\Phi$
are non-degenerate apart from the 
trivial $(2\Smax+1)$-fold spin degeneracy.
\label{FerroGSLemma}
\end{lemma}
This lemma is almost trivial, but will be proved in
Section~\ref{seceasy}.
For convenience, we call the state $\Phi$ characterized by the above lemma
the ``ferromagnetic ground states''.
These states are the energy minimizers in the sector
with $\Stot=\Smax$, and are not necessarily the true ground states.
We shall remind the readers about this abuse of terminology by
always putting the ``ferromagnetic ground states'' into quotation marks.

The first theorem establishes the {\em instability} of the
``ferromagnetic ground states'' against a single-spin flip.
Let $\bar{\varepsilon}$ denote the band width of the lowest band.
(See Section~\ref{secband}.)
For $\kappa\ne0$ and a generic choice of $\cbk{t'_{x,y}}$, the band width
$\bar{\varepsilon}$ is strictly positive.
\begin{theorem}
[Instability of the ``ferromagnetic ground states'']
Assume the conditions for Lemma~\ref{FerroGSLemma}.
We also assume that $\bar{\varepsilon}>0$ (which is generically true 
if $\kappa\ne0$), and
\begin{equation}
0\le U<\bar{\varepsilon}.
\end{equation} 
Then the ``ferromagnetic ground states''
are unstable under a single-spin flip in the sense that
\begin{equation}
\Emin(\Smax-1)<\Emin(\Smax).
\label{E<E}
\end{equation}
\label{instTh}
\end{theorem}
The theorem will be proved in Section~\ref{seceasy}.

Theorem~\ref{instTh} shows that one can lower the energy by flipping a 
single spin in the ``ferromagnetic ground states''.
It only shows that the ``ferromagnetic ground states'' are not
the true ground states.
To identify the 
true ground states (for $U\ne0$) in this situation 
is a highly nontrivial and interesting question.

To show the instability of a certain state (as in the above theorem)
is not a hard task since one can rely on the standard 
variational argument.
A really important (and difficult) part of the present work is to 
show the following theorem which states the {\em stability}.
\begin{theorem}
[Local stability of the ``ferromagnetic ground states'']
Consider the Hubbard model with the Hamiltonian (\ref{Ham1}).
Assume that the parameters satisfy
\begin{equation}
\la \ge\la_2,\quad
\akappa\le\kappa_1,\quad
\la\akappa\le p_1
\label{st1}
\end{equation}
and 
\begin{equation}
U\ge K_1\la^2t\abs{\kappa},
\label{st2}
\end{equation}
where $\la_2$, $\kappa_1$, $p_1$, and $K_1$ are positive constants which depend only on
the basic parameters $d$, $\nu$, and $R$.
Then the ``ferromagnetic ground states'' are stable under a single-spin flip
in the sense that
\begin{equation}
\Emin(\Smax-1)>\Emin(\Smax).
\label{stabineq}
\end{equation}
\label{stabilityTh}
\end{theorem}

We stress that the problem of stability against a single-spin flip is
already a highly nontrivial many-body problem.
The restriction to the sector with $\Stot=\Smax-1$ does {\em not}
reduce the problem to that of a single-particle (such as a magnon)
since there are plenty of spaces for the electrons to move around.
Moreover there is no way of expressing the eigenstates as Slater determinant
states since there are both up-spin and down-spin electrons interacting
via local Coulomb repulsion.
See also the discussion after Theorem~\ref{st1dTh}.

\subsection{Bounds for the Spin-Wave Excitation Energy}
\label{secsw}
Finally we describe our results about the elementary spin-wave
excitation.
The lower bound for the spin-wave energy 
in Theorem~\ref{Ek>Th} is closely related to 
the above local stability theorem.

For $x\in\Zd$, we let $T_x$ denote the translation operator
acting on the Hilbert space $\calH$ as
\begin{equation}
T_x\sbk{
\rbk{\prod_{y\in A}c^\dagger_{y,\up}}
\rbk{\prod_{z\in B}c^\dagger_{z,\dn}}
\vac}
=
\rbk{\prod_{y\in A}c^\dagger_{y+x,\up}}
\rbk{\prod_{z\in B}c^\dagger_{z+x,\dn}}
\vac,
\label{Tx}
\end{equation}
where $A$ and $B$ are arbitrary (ordered)
subsets of $\La$ as in (\ref{generalbasis}).
Let us define the space of wave number vectors by
\begin{equation}
\calK=
\set{k=(k_1,\ldots,k_d)}
{\mbox{$k_i=2\pi n_i/L$ with $n_i\in{\bf Z}$ such that $\abs{n_i}\le(L-1)/2$}}.
\label{Kdef}
\end{equation}
For each $k\in\calK$, we denote by $\calH_k$ the Hilbert space of the
states with the crystal momentum $k$, and with $L^d-1$  up-spin electrons
and one down-spin electron.
More precisely, we set
\begin{equation}
\calH_k
=
\set{\Phi\in\calH}
{\mbox{$T_x[\Phi]=\emik{x}\Phi$ for any $x\in\Zd$, and
$\Sztot\Phi=(\Smax-1)\Phi$}}.
\label{Hilbk}
\end{equation}

We can now define the energy $\ESW$ of the elementary spin-wave
excitation with the wave number $k\in\calK$ as the lowest energy
among the states in $\calH_k$.
Then we have the following two theorems.
\begin{theorem}
[Upper bound for the spin-wave energy]
Assume that $\la\ge\la_0$, and $\akappa\le\kappa_0$,
where $\la_0$ and $\kappa_0$ are positive constants which depend only
on $d$, $\nu$, and $R$.
Then we have
\begin{equation}
\ESW-\Emin(\Smax)
\le
F_1 \frac{U}{\la^4}G(k),
\label{ESW<}
\end{equation}
where
\begin{equation}
G(k) = 
2\sum_{f\in\calF_o}\sum_{g\in\calF_f}
\rbk{\sin\frac{k\cdot(f+g)}{2}}^2.
\label{Gk}
\end{equation}
The prefactor $F_1$ can be written as 
\begin{equation}
F_1 = 1 +\frac{A_4}{\la} +A_5\la\akappa 
+ \frac{A_6\la^2t\akappa^2}{U}
\label{F1}
\end{equation}
with the constants $A_i$ ($i=4,5,6$) which depend only on
$d$, $\nu$, and $R$.
\label{Ek<Th}
\end{theorem}
Like Theorem~\ref{instTh} about the instability of the 
``ferromagnetic ground
states'', the above theorem is proved by the standard 
variational argument.
See Section~\ref{SecSW<}.

The major achievement in the present paper is the lower bound which
corresponds to the above (\ref{ESW<}).
\begin{theorem}
[Lower bound for the spin-wave energy]
Assume that $\la\ge\la_3$, $\akappa\le\kappa_0$, and
$K_2\la t\ge U\ge A_3 \la^2 t\akappa$,
where $\la_3$, $\kappa_0$, $K_2$ and $A_3$ are positive constants which depend only
on $d$, $\nu$, and $R$.
Then we have
\begin{equation}
\ESW-\Emin(\Smax)
\ge
F_2 \frac{U}{\la^4}G(k),
\label{ESW>}
\end{equation}
with $G(k)$ defined in (\ref{Gk}).
The prefactor $F_2$ can be written as 
\begin{equation}
F_2 = 1 - A_1\akappa - \frac{A_2}{\la} 
- \frac{A_3\la^2t\akappa}{U}
\label{F2}
\end{equation}
with the constants $A_1$, $A_2$, and $A_3$ which depend only on
$d$, $\nu$, and $R$.
\label{Ek>Th}
\end{theorem}

Note that (\ref{F1}) and (\ref{F2}) imply that $F_1\simeq F_2\simeq1$
when $\la$ is large and $\akappa$ is (very) small.
In this case the dispersion relation $\ESW$ of the elementary spin-wave
excitation is given by
\begin{equation}
\ESW-\Emin(\Smax)\simeq\frac{U}{\la^4}G(k).
\end{equation}
This dispersion relation is exactly what one expects in the 
ferromagnetic Heisenberg quantum spin system 
defined on the hypercubic lattice $\Lao$
with the exchange interaction $J_{\rm eff}=2U\la^{-4}$.

As we have already stressed in Section~\ref{secMM1d},  
Theorem~\ref{Ek>Th} requires an upper bound for the Coulomb interaction 
$U$. 
By noting that $\ESW$ is increasing in $U$, however,
it is easy to prove nontrivial lower bounds for $\ESW$
for larger values of $U$.
\begin{coro}
Assume that $\la\ge\la_3$, $\akappa\le\kappa_0$,
$A_3\la\akappa/K_2\le1$, and $U\ge K_2\la t$.
Then we have
\begin{equation}
\ESW-\Emin(\Smax)
\ge
F_3 G(k),
\label{ESW>2}
\end{equation}
with $G(k)$ defined in (\ref{Gk}).
The prefactor $F_3$ can be written as 
\begin{equation}
F_3 = \rbk{1 - A_1\akappa - \frac{A_2}{\la} 
- \frac{A_3\la\akappa}{K_2}}
\frac{K_2t}{\la^3}.
\label{F3}
\end{equation}
\label{Ek>Coro}
\end{coro}
\begin{proof}{Proof}
The first three conditions assumed here guarantee that we can use 
Theorem~\ref{Ek>Th} when $U=K_2\la t$.
We claim that, for each $k\in\calK$, $\ESW$ is an increasing 
function of $U$.
This is because both $H$ and $\Hint$ commute with $T_{x}$ and 
$\Sztot$, and 
$\ESW$ is defined to be the lowest energy in the sector with 
the fixed momentum $k$ and the fixed eigenvalue of $\Sztot$.
Then it is trivial that $\ESW$ for $U\ge K_2\la t$ satisfies the
desired bound (\ref{ESW>2}), where the right-hand side of 
(\ref{ESW>2}) is obtained by substituting $U=K_2\la t$
into (\ref{ESW>}).
\end{proof}

\Section{Single-Electron Problem}
\label{SecSingle}
We shall investigate the properties of the single-electron
system corresponding to our Hubbard model.
A careful study of single-electron properties is
indispensable when we work with interacting many-electron systems.
\subsection{Band Structure of the Model}
\label{secband}
If there is only a single electron with, say, up-spin in the whole
system, a general state can be written as
\begin{equation}
\Phi(\phi)
=
\sum_{x\in\La}\phi_x \,c^\dagger_{x,\up}\, \vac
\label{Phiphi}
\end{equation}
with $\phi_x\in{\bf C}$.
As in the standard quantum mechanics, we regard 
the collection $\phi=(\phi_x)_{x\in\La}$ as a vector in a
$\abs{\La}$-dimensional complex linear space 
$\Hsing\cong{\bf C}^{\abs{\La}}$,
which we call the single-electron Hilbert space.

Since it obviously holds that $\Hint\,\Phi(\phi)=0$,
the Schr\"{o}dinger equation $H\Phi(\phi)=\epsilon\Phi(\phi)$
reduces to
\begin{equation}
\sum_{y\in\La}t_{x,y}\,\phi_y = \epsilon\,\phi_x,
\label{Sch1}
\end{equation}
where $t_{x,y}=t^{(0)}_{x,y}+\kappa t'_{x,y}$, and we denote the
(single-electron) energy eigenvalue as $\epsilon$.

By rewriting the expression (\ref{Hhop01}) for $\Hhop$ as
\begin{equation}
\Hhop
=
t\sum_{\sigma=\up,\dn}\sum_{u\in\calU'}\sum_{x\in\Lau}
\rbk{\la\cxs+\sum_{f\in\calF_u}c^\dagger_{x+f,\sigma}}
\rbk{\la\axs+\sum_{f\in\calF_u}c_{x+f,\sigma}},
\label{Hhop03}
\end{equation}
we can write down (\ref{Sch1}) in a concrete form as
\begin{equation}
\epsilon\,\phi_x
=
t\sum_{f\in\calF_o}
\rbk{\la\,\phi_{x+f}+\sum_{g\in\calF_f}\phi_{x+f+g}}
+\kappa\sum_{y\in\La}t'_{x,y}\,\phi_y,
\label{Sch2}
\end{equation}
and
\begin{equation}
\epsilon\,\phi_{x+u}
=
\la^2t\,\phi_{x+u}+\la t\sum_{f\in\calF_u}\phi_{x+u+f}
+\kappa\sum_{y\in\La}t'_{x+u,y}\,\phi_y,
\label{Sch3}
\end{equation}
where $x\in\Lao$ and $u\in\calU'$.
We recall that $\calU$ is the unit cell of the lattice, 
and\footnote{
For any sets $A$ and $B$, $A\backslash B$ denotes the set
$\set{x\in A}{x\not\in B}$.
} $\calU'=\calU\backslash\cbk{o}$.

Since the hopping matrix elements $t'_{x,y}$
are invariant under the translation by any integer
vector $z\in\Zd$,
we can use the Bloch theorem to write an eigenstate
of (\ref{Sch1}) as
\begin{equation}
\phi_x = \eikx\,\vk{x},
\label{Blochth}
\end{equation}
with $k\in\calK$ (see (\ref{Kdef})), and $\vk{x}$ satisfying
\begin{equation}
\vk{x}=\vk{x+y},
\label{vtransinv}
\end{equation}
for any $y\in\Zd$.
With the translation invariance (\ref{vtransinv}) in mind,
we can identify, for each fixed $k\in\calK$,
the function $\vk{x}$ (of $x$) with a $b$-dimensional vector\footnote{
In the present paper the bold face symbols are reserved to indicate
elements of the $b$-dimensional vector space
introduced here.
}
\begin{equation}
\vv=\rbk{\vk{u}}_{u\in\calU}
\in{\bf C}^b,
\label{vvec}
\end{equation}
where $b=\abs{\calU}=\dnu+1$ will turn out to be the number of 
the bands in the Schr\"{o}dinger equation (\ref{Sch1}).

By substituting the representation (\ref{Blochth})
into the Schr\"{o}dinger equation (\ref{Sch2}), (\ref{Sch3}),
we find the  equation (the Schr\"{o}dinger equation
in $k$-space) 
\begin{equation}
\epsilon\,\vv=(\la^2t \,\Mt + \kappa t \,\Qt)\vv,
\label{SchVec1}
\end{equation}
which determines, for each $k\in\calK$,
the eigenvalue $\epsilon$ of the original 
Schr\"{o}dinger equation (\ref{Sch1}).
Here $\Mt=(\Mk{u,u'})_{u,u'\in\calU}$ and 
$\Qt=(\Qk{u,u'})_{u,u'\in\calU}$ are $b\times b$ matrices\footnote{
Sanserif symbols denote matrices in the $b$-dimensional vector space.
}.
They are defined by
\begin{equation}
\Mk{u,u'}=\Mk{u',u}=
\cases{
{A(k)}/{\la^2} & if $u=u'=o$;\cr
{C_u(k)}/{\la} & if $u\in\calU'$ and $u'=o$;\cr
0 & if $u,u'\in\calU'$ and $u\ne u'$;\cr
1 & if $u=u'\in\calU$,\cr
}
\label{Mkdef}
\end{equation}
and
\begin{equation}
\Qk{u,u'}=\frac{1}{t}\sum_{y\in\La_{u'}}t'_{u,y}
\,e^{ik\cdot(y-u)}.
\label{Qdef}
\end{equation}
Here we have introduced
\begin{equation}
C_f(k)=\sum_{g\in\calF_f}e^{ik\cdot g}
\label{Ckdef}
\end{equation}
for $f\in\calF_o$, and
\begin{equation}
A(k)=\sum_{f\in\calF_o}\sum_{g\in\calF_f}
e^{ik\cdot(f+g)}
=\sum_{u\in\calU'}\rbk{C_u(k)}^2.
\label{Akdef}
\end{equation}

Since the matrix $(\la^2t\,\Mt+\kappa t\,\Qt)$ is hermitian,
it generically has $b$ eigenvalues and eigenstates for each $k\in\calK$.
We denote these eigenvalues as $\ep_j(k)$, where the band index
$j=1,2,\ldots,b$ is assigned so that 
$\ep_j(k)\le\ep_{j+1}(k)$.
When viewed as a function of $k$, the eigenvalues $\ep_j(k)$ are usually
called the dispersion relations of the $j$-th band.

When $\kappa=0$, the eigenvalue problem (\ref{SchVec1}) can be solved
easily, and we obtain the dispersion relations
\begin{equation}
\ep_j(k)=
\cases{
0 & for $j=1$;\cr
\la^2t & for $j=2,\ldots,b-1$;\cr
\la^2t+tA(k) & for $j=b$.\cr
}
\label{flatdispersion}
\end{equation}
Note that the model has a rather singular band structure where most
of the bands have constant energies (i.e., are flat), and all the bands
with $j=2,\ldots,b-1$ are completely degenerate.
Another important feature of (\ref{flatdispersion}) is 
that the lowest band
($j=1$) is separated from the higher bands by an energy gap $\la^2t$.
See Figure~\ref{band1d}a for the dispersion relation in $d=1$.
We have also drawn the dispersion relation of the flat-band model
with $d=2$ and $\nu=1$ in Figure~\ref{band2df}.

\begin{figure}
\centerline{\includegraphics[width=7cm,clip]{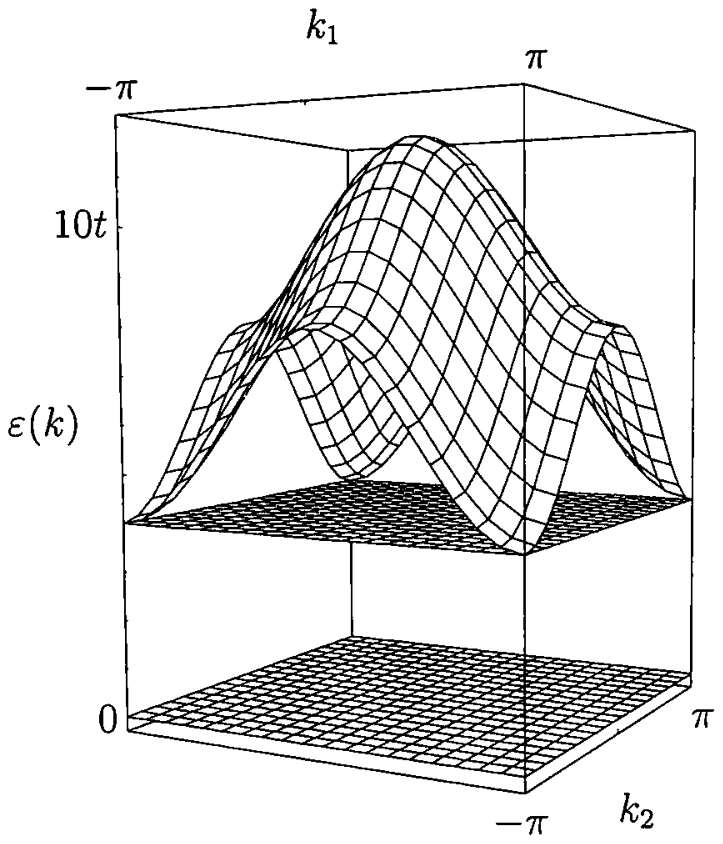}}
\caption{\captionH}
\label{band2df}
\end{figure}

\begin{figure}
\centerline{\includegraphics[width=7cm,clip]{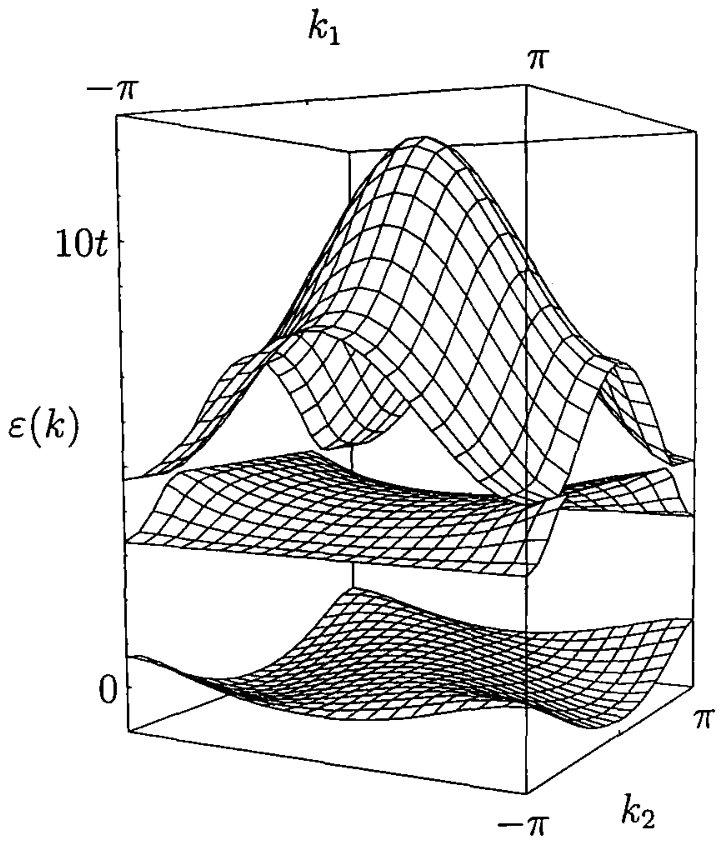}}
\caption{\captionI}
\label{band2dnf}
\end{figure}

For $\kappa\ne0$, with a generic choice of $\Hpert$,
the dispersion relations $\ep_j(k)$ become $k$-dependent,
and the bands are no longer flat.
See Figure~\ref{band1d}b for the dispersion relation 
with slightly perturbed model
in $d=1$, and Figure~\ref{band2dnf} for that in $d=2$.
The degeneracy between the bands
with $j=2,\ldots,b-1$ is also likely to be lifted 
(unless the perturbation
has certain symmetry).
Actual band structures depend delicately on the choice of the perturbation, and are
not easy to calculate.
It generically holds, however, that the lowest band is 
still separated from
the rest of the bands by an energy gap, provided that 
$\akappa$ is not too large.
We present the following crude estimate, which 
is sufficient for our purpose.
\begin{lemma}
Assume that $\la\ge\la_1=2^\nu\sqrt{b-1}$, and
$\akappa\la^{-2}\le r_1=9\times10^{-3}/b$.
Then we have
\begin{equation}
\ep_1(k)\le\frac{\la^2t}{4},
\label{e1k<}
\end{equation}
and
\begin{equation}
\ep_j(k)\ge\frac{3\la^2t}{4}
\label{ejk>}
\end{equation}
for $j=2,3,\ldots,b$.
\label{gapLemma}
\end{lemma}
\begin{proof}{Proof}
The statement is almost trivial, but we give a proof for completeness.
Since the eigenvalues of $\la^2t\,\Mt$ are either $=0$ or $\ge\la^2t$,
we have
\begin{equation}
\rbk{\la^2t\,\Mt-\frac{\la^2t}{2}}^2\ge
\rbk{\frac{\la^2t}{2}}^2.
\end{equation}
Consider the similar quantity for the perturbed matrix,
and note that
\begin{eqnarray}
&&\sbk{
\cbk{\la^2t\,\Mt+\kappa t\,\Qt}-\frac{\la^2t}{2}
}^2
\ret
&=&
\rbk{\la^2t\,\Mt-\frac{\la^2t}{2}}^2+\kappa^2t^2(\Qt)^2
\ret
&+&
\rbk{\la^2t\,\Mt-\frac{\la^2t}{2}}\kappa t\,\Qt
+\kappa  t\,\Qt\rbk{\la^2t\,\Mt-\frac{\la^2t}{2}}
\ret
&\ge&
\rbk{\frac{\la^2t}{2}}^2-\kappa^2t^2\norm{\Qt}^2
-2\akappa t\norm{\Qt}\rbk{\la^2t\norm{\Mt}+\frac{\la^2t}{2}}.
\label{tMtQ}
\end{eqnarray}
By substituting the assumed bound for $\akappa$, and the bounds
$\norm{\Qt}\le b$, 
$\norm{\Mt}\le1+\abs{A(k)}\la^{-2}\le1+(b-1)4^\nu\la^{-2}\le2$,
we observe that the right-hand side of (\ref{tMtQ}) is not less than
$(\la^2t/4)^2$.
This means that the Schr\"{o}dinger equation (\ref{SchVec1}) cannot have
eigenvalues in the range $\la^2t/4\le\ep\le(3/4)\la^2t$.
Since the eigenvalue $\ep_j(k)$ with $j$ and $k$ fixed is continuous in 
$\kappa$, the statement of the lemma follows.
\end{proof}

From now on, we assume that the condition for lemma~\ref{gapLemma}
is satisfied.
The existence of a gap allows us to treat the lowest band in a special manner.
Let us decompose the single-electron Hilbert space as
\begin{equation}
\Hsing=\Hsingo\oplus\Hsingp
\label{Hsingdecomp}
\end{equation}
where $\Hsingo$ is the Hilbert space corresponding to the
lowest band (with the band index $j=1$).
It is spanned by the eigenstates of (\ref{Sch1})
with the eigenvalue $\ep_1(k)$ for $k\in\calK$.
$\Hsingp$ is the orthogonal complement of $\Hsingo$.
The dimensions of the spaces $\Hsingo$ and $\Hsingp$ are
$L^d=\abs{\Lao}$ and $(b-1)L^d=\abs{\La'}$, respectively.

Let $P^{(1)}_{\rm single}$ be the orthogonal projection onto the space
$\Hsingo$, 
and denote by $T$ the hopping operator
on $\Hsing$, whose matrix representation
is given by $(t_{x,y})_{x,y\in\La}$.
We define the modified hopping operator $\Ttil$ by
\begin{equation}
\Ttil=TP^{(1)}_{\rm single}+\frac{3}{4}\la^2t(1-P^{(1)}_{\rm single}),
\label{Ttil}
\end{equation}
and denote by $(\ttil_{x,y})_{x,y\in\La}$ the matrix representation 
of $\Ttil$.
Note that the bound (\ref{ejk>}) implies the operator inequality 
$T\ge\Ttil$.
Define the modified hopping Hamiltonian by
\begin{equation}
\Hhopt=\sum_{\sigma=\up,\dn}\sum_{x,y\in\La}
\ttil_{x,y}\,\cxs\ays,
\label{Hhoptil}
\end{equation}
which also satisfies
\begin{equation}
\Hhop\ge\Hhopt.
\label{H>Htil}
\end{equation}
Although the introduction of $\Hhopt$ is not essential for our proof,
it considerably simplifies the required estimates.

\subsection{Localized Bases}
\label{seclocbases}
We introduce bases for the single-electron spaces
$\Hsingo$ and $\Hsingp$, in which each basis state is localized at 
a lattice site.
The use of such localized bases enable us to treat electrons as
``particles'' but with taking into account the band structure of 
the model.
The actual construction of the bases will be presented in 
Section~\ref{SecBasis}.

We start from the easy case with $\kappa=0$, i.e., the flat-band models.
For $x\in\Lao$, we define the state 
$\psi^{(x)}=(\psis{x}{y})_{y\in\La}\in\Hsing$ by
\begin{equation}
\psis{x}{y}=
\cases{
1&if $x=y$;\cr
-{1}/{\la}&if $y\in\La'$ and $\abs{x-y}=\sqrt{\nu}/2$;\cr
0&otherwise.\cr
}
\label{psi1}
\end{equation}
An explicit calculation shows that 
$\sum_{y\in\La}t^{(0)}_{x,y}\psis{x}{y}=0$ for $x\in\Lao$.
This can be done by using (\ref{txy0}), but it is easier
to use (\ref{Hhop01}).
This means that $\psi^{(x)}\in\Hsingo$ since the lowest band has a constant
energy $\ep=0$ when $\kappa=0$ as in (\ref{flatdispersion}).
Since the states $\psi^{(x)}$ with $x\in\Lao$ are 
linearly independent,
and $\abs{\Lao}$ is equal to the dimension of $\Hsingo$, we find that 
the collection of the states $\cbk{\psi^{(x)}}_{x\in\Lao}$ form a
 basis of $\Hsingo$.

For $x\in\La'$, we similarly define $\psi^{(x)}\in\Hsing$ as
\begin{equation}
\psis{x}{y}=
\cases{
1&if $x=y$;\cr
{1}/{\la}&if $y\in\Lao$ and $\abs{x-y}=\sqrt{\nu}/2$;\cr
0&otherwise.\cr
}
\label{psi2}
\end{equation}
It is evident that $\psi^{(x)}$ with $x\in\La'$ and 
$\psi^{(x')}$ with $x'\in\Lao$ are orthogonal with each other.
This means that $\psi^{(x)}$ with any $x\in\La'$ is orthogonal to
the space $\Hsingo$.
By counting the dimension, it then follows that 
$\cbk{\psi^{(x)}}_{x\in\La'}$ form a
 basis of $\Hsingp$.

Both the bases $\cbk{\psi^{(x)}}_{x\in\Lao}$ and 
$\cbk{\psi^{(x)}}_{x\in\La'}$ are not orthonormal,
but the states in the bases are sharply localized at lattice sites.
The introduction and the use of the localized basis for $\Hsingo$
was essential in the study of the flat-band Hubbard models in
\cite{92e,93d}.

To deal with non-flat band models,
we shall construct similar bases for the models  with $\kappa\ne0$.
Since this is a problem of perturbation theory in a one-body
quantum mechanics, there is no essential difficulty when 
the strength of the perturbation $\akappa t$ is sufficiently smaller
than the energy gap $\la^2t$.
In Section~\ref{SecBasis}, we prove the following.
\begin{lemma}
Suppose that $\la\ge\la_0$ and $\akappa\la^{-2}\le r_0$,
where $\la_0$ and $r_0$ are positive constants which depend only on
the dimensions $d$, $\nu$.
Then we can take for each $x\in\La$ a state 
$\phi^{(x)}=(\phis{x}{y})_{y\in\La}\in\Hsing$ such that
$\phis{x}{y}=\phis{x+z}{y+z}$ holds for any $z\in\Zd$.
The collections of the states
$\cbk{\phi^{(x)}}_{x\in\Lao}$ and $\cbk{\phi^{(x)}}_{x\in\La'}$
form  bases of $\Hsingo$ and $\Hsingp$, respectively.
These basis states are summable as
\begin{equation}
\sum_{y\in\La}\abs{\phis{x}{y}-\psis{x}{y}}
\le B_1\frac{\akappa}{\la^2},
\label{reg1}
\end{equation}
\begin{equation}
\sum_{y\in\La}\abs{x-y}\abs{\phis{x}{y}-\psis{x}{y}}
\le B_1R\frac{\akappa}{\la^2},
\label{reg2}
\end{equation}
\begin{equation}
\sum_{x\in\La}\abs{\phis{x}{y}-\psis{x}{y}}
\le B_1\frac{\akappa}{\la^2},
\label{reg3}
\end{equation}
and
\begin{equation}
\sum_{x\in\La}\abs{x-y}\abs{\phis{x}{y}-\psis{x}{y}}
\le B_1R\frac{\akappa}{\la^2},
\label{reg4}
\end{equation}
where $B_1$ is a positive constant which depend only on $d$ and $\nu$.
\label{basisLemma}
\end{lemma}
The bounds (\ref{reg1})-(\ref{reg4}) imply that each state $\phi^{(x)}$
is sharply localized at the reference site $x$.
The bounds also show that the states $\phi^{(x)}$ are
chosen so that they become identical to $\psi^{(x)}$ when $\kappa=0$.

Let us investigate how the modified hopping operator 
$\Ttil=(\ttil_{x,y})_{x,y\in\La}$ introduced in (\ref{Ttil})
acts on these basis states.
From the definition (\ref{Ttil}), it is obvious that
\begin{equation}
\Ttil\phi^{(x)}=\frac{3}{4}\la^2t\,\phi^{(x)}
\label{Ttilphi1}
\end{equation}
if $x\in\La'$.
For $x\in\Lao$, the basis state is transformed as
\begin{equation}
\Ttil\phi^{(x)}=\sum_{y\in\Lao}\tau_{y,x}\,\phi^{(y)},
\label{Ttilphi2}
\end{equation}
where the effective hopping matrix elements $\tau_{y,x}$ are given by
\begin{equation}
\tau_{y,x} = (2\pi)^{-d}\int dk\,e^{ik\cdot(y-x)}\ep_1(k),
\label{tauDef}
\end{equation}
where $\ep_1(k)$ is the dispersion relation of the lowest band discussed in
Section~\ref{secband}, and $\int dk(\cdots)$ is a shorthand for the sum
$(2\pi/L)^d\sum_{k\in\calK}(\cdots)$.
Note that only $y$ in $\Lao$ appear in the right-hand side of
(\ref{Ttilphi2}), reflecting the band structure.

The precise form of $\tau_{y,x}$ depends on specific perturbation.
But the following general bound is sufficient for our purpose.
\begin{lemma}
When $\la\ge\la_0$ and $\akappa\la^{-2}$, we have
\begin{equation}
\sum_{x\in\Lao}\abs{\tau_{y,x}}
=\sum_{y\in\Lao}\abs{\tau_{y,x}}
\le B_1\akappa t,
\label{regtau1}
\end{equation}
and
\begin{equation}
\sum_{x\in\Lao}\abs{x-y}\abs{\tau_{y,x}}
=\sum_{y\in\Lao}\abs{x-y}\abs{\tau_{y,x}}
\le B_1R\akappa t.
\label{regtau2}
\end{equation}
\label{tauLemma}
\end{lemma}

Since the bases 
$\cbk{\phi^{(x)}}_{x\in\Lao}$ and $\cbk{\phi^{(x)}}_{x\in\La'}$
are not orthonormal, it is convenient to introduce the bases
which are dual to them.
The dual bases are constructed uniquely by a standard procedure
(in Section~\ref{SecBasis}), and we can prove the following.
\begin{lemma}
Suppose that $\la\ge\la_0$ and $\akappa\la^{-2}\le r_0$.
Then we can take for each $x\in\La$ a state 
$\phitil^{(x)}=(\phit{x}{y})_{y\in\La}\in\Hsing$ such that
$\phit{x}{y}=\phit{x+z}{y+z}$ holds for any $z\in\Zd$.
The collections of the states
$\cbk{\phitil^{(x)}}_{x\in\Lao}$ and $\cbk{\phitil^{(x)}}_{x\in\La'}$
form  bases of $\Hsingo$ and $\Hsingp$, respectively.
They are dual of the 
bases $\cbk{\phi^{(x)}}_{x\in\Lao}$ and $\cbk{\phi^{(x)}}_{x\in\La'}$
in the sense that we have
\begin{equation}
\sum_{y\in\La}\rbk{\phit{x}{y}}^*\phis{x'}{y}=\delta_{x,x'}
\label{duality1}
\end{equation}
for any $x,x'\in\La$, and
\begin{equation}
\sum_{x\in\La}\rbk{\phit{x}{y}}^*\phis{x}{y'}=\delta_{y,y'}
\label{duality2}
\end{equation}
for any $y,y'\in\La$.
These basis states are summable as
\begin{equation}
\sum_{y\in\La}\abs{\phit{x}{y}-\psis{x}{y}}
\le B_1\frac{\akappa}{\la^2}+\frac{B_2}{\la^2},
\label{reg5}
\end{equation}
\begin{equation}
\sum_{y\in\La}\abs{x-y}\abs{\phit{x}{y}-\psis{x}{y}}
\le B_1R\frac{\akappa}{\la^2}+\frac{B_2}{\la^2},
\label{reg6}
\end{equation}
\begin{equation}
\sum_{x\in\La}\abs{\phit{x}{y}-\psis{x}{y}}
\le B_1\frac{\akappa}{\la^2}+\frac{B_2}{\la^2},
\label{reg7}
\end{equation}
and
\begin{equation}
\sum_{x\in\La}\abs{x-y}\abs{\phit{x}{y}-\psis{x}{y}}
\le B_1R\frac{\akappa}{\la^2}+\frac{B_2}{\la^2},
\label{reg8}
\end{equation}
where $B_2$ is a positive constant which depend only on $d$ and $\nu$.
\label{dualbasisLemma}
\end{lemma}
Note that the right-hand side of (\ref{reg5})-(\ref{reg8}) do not vanish when
$\kappa=0$.
This is because the dual basis state $\phit{x}{y}$ has nonvanishing
exponentially decaying tail even in the flat-band model.
This remarkable asymmetry between the states $\phi^{(x)}$ and
their dual $\phitil^{(x)}$ plays a fundamental role in our work.

\Section{Localized Basis for the Hubbard Model}
\label{SecLocbasis}
In the present section, we discuss 
the framework for describing many-electron systems by using the 
localized basis introduced in Section~\ref{seclocbases}.
Elementary statements about the ``ferromagnetic ground states''
and the theorem for flat-band ferromagnetism
are also proved.
\subsection{Fermion Operators for the Localized Bases}
\label{secFermiloc}
We rewrite the Hubbard Hamiltonian (\ref{Ham1}) by using
the new fermion operators.
The new representation turns out to be suitable
for our purpose to take into account both the particle-like
nature of electrons and the band structure of the model.

We first define the creation operator corresponding to the
basis state $\phi^{(x)}$ as
\begin{equation}
\adxs = \sum_{y\in\La}\rbk{\phis{x}{y}}^*\cys
\label{aDef}
\end{equation}
for $x\in\La$ and $\sigma=\up,\dn$.
Similarly we define the annihilation operator
corresponding to the dual basis state $\phitil^{(x)}$ as
\begin{equation}
\bxs = \sum_{y\in\La}\phit{x}{y} \ays
\label{bDef}
\end{equation}
for $x\in\La$ and $\sigma=\up,\dn$.

By using the basic anticommutation relations (\ref{ac1}), (\ref{ac2}),
the definitions (\ref{aDef}), (\ref{bDef}), and
the duality relation (\ref{duality1}),
we find that these operators satisfy the anticommutation relations
\begin{equation}
\cbk{\adxs,\ad_{y,\tau}}=\cbk{\bxs,b_{y,\tau}}=0,
\label{ac3}
\end{equation}
and
\begin{equation}
\cbk{\adxs,b_{y,\tau}}=\delta_{x,y}\delta_{\sigma,\tau}
\label{ac4}
\end{equation}
for any $x,y\in\La$ and $\sigma,\tau=\up,\dn$.
Note that (\ref{ac3}) and (\ref{ac4}) have exactly the same forms as
the canonical anticommutation relations.

By using the other duality relation (\ref{duality2}), we can invert
(\ref{aDef}) and (\ref{bDef}) to get
\begin{equation}
\cxs=\sum_{y\in\La}\phit{y}{x}\ad_{y,\sigma},
\label{ainv}
\end{equation}
and
\begin{equation}
\axs=\sum_{y\in\La}\rbk{\phis{y}{x}}^*b_{y,\sigma}.
\label{binv}
\end{equation}

\subsection{Representation of the Hamiltonian}
\label{secrepHam}
We shall rewrite the Hamiltonian using the operators
$\adxs$ and $\bxs$.
As for the hopping part, we treat the modified Hamiltonian
$\Hhopt$ defined in (\ref{Hhoptil}), rather than the original
$\Hhop$.
By substituting (\ref{ainv}) and (\ref{binv}) into (\ref{Hhoptil}),
we find that
\begin{eqnarray}
\Hhopt &=& 
\sum_{\sigma=\up,\dn}\sum_{x,y,v,w\in\La}
\phit{v}{x}\ttil_{x,y}\rbk{\phis{w}{y}}^*\ad_{v,\sigma}b_{w,\sigma}
\ret
&=&
\sum_{\sigma=\up,\dn}
\cbk{
\sum_{x,y\in\Lao}\tau_{x,y}\adxs b_{y,\sigma} 
+ \frac{3\la^2t}{4} \sum_{x\in\La'}\adxs b_{x,\sigma}
},
\label{HtilhopRep}
\end{eqnarray}
where we have used (\ref{Ttilphi1}) and (\ref{Ttilphi2})
which determine the action of $\Ttil$, and the duality
relation (\ref{duality1}).
The representation (\ref{HtilhopRep}) makes the band structure
manifest.

Similarly we can rewrite the interaction Hamiltonian (\ref{Hint})
as
\begin{eqnarray}
\Hint &=&
U\sum_{x,y,v,w,z\in\La}
\rbk{\phit{y}{x}\ad_{y,\up}} 
\rbk{\rbk{\phis{z}{x}}^*b_{z,\up}}
\rbk{\phit{v}{x}\ad_{v,\dn}} 
\rbk{\rbk{\phis{w}{x}}^*b_{w,\dn}}
\ret
&=&
\sum_{y,v,w,z\in\La}
\Ut_{y,v;w,z}\,\ad_{y,\up}\ad_{v,\dn}b_{w,\dn}b_{z,\up},
\label{HintRep}
\end{eqnarray}
where the effective interaction is given by
\begin{equation}
\Ut_{y,v;w,z} = U\sum_{x\in\La}
\phit{y}{x}\phit{v}{x}\rbk{\phis{w}{x}\phis{z}{x}}^*.
\label{Util}
\end{equation}
Note that the interaction Hamiltonian $\Hint$ in the new
representation (\ref{HintRep}) is no longer local.

\Rem
It is also possible to write down the  representation
similar to (\ref{HtilhopRep})
for the original hopping Hamiltonian
\begin{equation}
\Hhop=\sum_{\sigma=\up,\dn}
\cbk{
\sum_{x,y\in\Lao}\tau_{x,y}\adxs b_{y,\sigma}
+\sum_{x,y\in\La'}\tau_{x,y}\adxs b_{y,\sigma}
},
\label{HhopRep}
\end{equation}
with properly defined $\tau_{x,y}$ for $x,y\in\La'$.
\par\bigskip

\subsection{Elementary Facts about the ``Ferromagnetic Ground States''}
\label{seceasy}
We can now prove the basic statement about the
``ferromagnetic ground states''.
\bigno
\begin{proof}{Proof of Lemma~\ref{FerroGSLemma}}
Since we are interested in states with $\Stot=\Smax$, we can 
concentrate on the sector with $\Sztot=\Smax$.
States (with $\Stot=\Smax$) in other sectors can be obtained by 
suitably applying the total spin lowering operator.
Clearly $\Hint$ annihilates a state with $\Sztot=\Smax$ as it 
contains only up-spin electrons.

Because the conditions for Lemma~\ref{gapLemma} are satisfied, there is 
a finite energy gap between the lowest band and the remaining bands.
In order to make the eigenvalue of $\Hhop$ small, we need to use as
many states from the lowest band.
Since the electron number $L^d$ is identical to the dimension of the 
Hilbert space $\Hsingo$ for the lowest band, this can be done in a unique way, 
and we find
\begin{equation}
\UP = \rbk{\prod_{x\in\Lao}\ad_{x,\up}}\vac
\label{Phiup}
\end{equation}
is the desired ``ferromagnetic ground state.''
By operating $\Hhop$ in the representation (\ref{HhopRep}),
we find that $H\UP=E_0\UP$ with
\begin{equation}
E_0=\sum_{x\in\Lao}\tau_{x,x}=L^d\tau_{o,o}.
\label{E0}
\end{equation}
\end{proof}

We also prove the theorem about the instability
of the ``ferromagnetic ground states.''
The proof is based on the standard variational argument.
\bigno
\begin{proof}{Proof of Theorem~\ref{instTh}}
Let $d^\dagger_{k,\sigma}$ be the creation operator for the
Bloch state (\ref{Blochth}) in the lowest band with the wave number vector 
$k\in\calK$, and let $\ep_1(k)$ be the corresponding energy eigenvalue.
Let $k_{\rm min}$ and $k_{\rm max}$ be such that 
\begin{equation}
\ep_1(k_{\rm min})\le\ep_1(k)\le\ep_1(k_{\rm max})
\end{equation}
holds for any $k\in\calK$.
The band width is given by
$\bar{\ep}=\ep_1(k_{\rm max})-\ep_1(k_{\rm min})$.
Take a variational state
\begin{equation}
\Phi_{\rm var} = d^\dagger_{k_{\rm min},\dn}d_{k_{\rm max},\up}\UP.
\end{equation}
The energy expectation value of the state $\Phi_{\rm var}$ is
easily shown to satisfy
\begin{equation}
\frac{(\Phi_{\rm var},H\Phi_{\rm var})}{(\Phi_{\rm var},\Phi_{\rm var})}
\le E_0-\bar{\ep}+U=\Emin(\Smax)-\bar{\ep}+U.
\end{equation}
The claimed instability follows when $\bar{\ep}>U$.
\end{proof}

\subsection{Flat-Band Ferromagnetism}
\label{secflatproof}
In \cite{92e,93d} Theorem~\ref{flatbandTh}, which establishes 
flat-band ferromagnetism, was proved for the models with $\nu=1$.
Although the extension to the general case is not hard, we present it 
here for completeness.

\begin{proof}{Proof of Theorem~\ref{flatbandTh}}
The flat-band model is characterized by $\tau_{x,y}=0$ for any 
$x,y\in\Lao$.
Then it is easily verified (from, say, (\ref{HhopRep})) that 
$\Hhop\ge0$.
We also know $\Hint\ge0$, and hence $H\ge0$.
From (\ref{E0}), on the other hand, one finds that the 
``ferromagnetic ground state'' $\UP$ (\ref{Phiup}) has vanishing 
energy, and hence is a ground state of $H$.
The remaining task is to determine all the other ground states.

Let $\Phi$ be an arbitrary ground state with $L^{d}$ electrons.
We obviously have 
\begin{equation}
	\Hhop\Phi=0,
	\label{HhopP=0}
\end{equation}
and
\begin{equation}
	\Hint\Phi=0,
	\label{HintP=0}
\end{equation}
which mean that $\Phi$ is at the same time a ground state of $\Hhop$ 
and of $\Hint$.
As we discussed in Section~\ref{secMM1d}, this is a special feature 
of flat-band models.

Since $\Hint$ (\ref{Hint}) is a sum of nonnegative terms, 
(\ref{HintP=0}) implies $n_{x,\up}n_{x,\dn}\Phi=0$ for each $x\in\La$.
Since 
$n_{x,\up}n_{x,\dn}=(c_{x,\up}c_{x,\dn})^{\dagger}(c_{x,\up}c_{x,\dn})$,
this further implies 
$c_{x,\up}c_{x,\dn}\Phi=0$ for each $x\in\La$.
By using the inversion formula (\ref{binv}), and noting that 
$(\phi^{(y)}_{x})^{*}=\psi^{(y)}_{x}$ for the flat-band models (see 
Section~\ref{seclocbases}, especially (\ref{psi1})), this reduces to 
the following useful condition.
\begin{equation}
	\sum_{y,z\in\La}\psi^{(y)}_{x}\psi^{(z)}_{x}
	b_{y,\up}b_{z,\dn}
	\Phi
	=0
	\label{bbP=0}
\end{equation}

The relation (\ref{HhopP=0}) implies that the state $\Phi$ consists 
only of the single-electron states from the lowest (flat) band.
Therefore we expand it as
\begin{equation}
	\Phi=\sum_{A,B\subset\Lao}
	f(A,B)
	\rbk{\prod_{x\in A}a^{\dagger}_{x,\up}}
	\rbk{\prod_{x\in B}a^{\dagger}_{x,\dn}}
	\vac,
	\label{Phiexpand}
\end{equation}
where the sum is taken over all subsets $A,B\subset\Lao$ such 
that $|A|+|B|=L^{d}$, and $f(A,B)$ are coefficients.

For $x\in\Lao$ and $\Phi$ of the form (\ref{Phiexpand}), the condition 
(\ref{bbP=0}) becomes
\begin{equation}
	b_{x,\up}b_{x,\dn}
	\Phi
	=0,
	\label{bbP=02}
\end{equation}
because of the definition (\ref{psi1}) of the $\psi$ states.
By using the anticommutation relation (\ref{ac4}), (\ref{bbP=02}) 
implies that $f(A,B)=0$ whenever $A\cap B\ne\emptyset$.
Thus the expansion (\ref{Phiexpand}) can be reorganized as
\begin{equation}
	\Phi=\sum_{\sigma}g(\sigma)
	\rbk{\prod_{x\in\Lao}a^{\dagger}_{x,\sigma(x)}}
	\vac,
	\label{Phiexpand2}
\end{equation}
where the sum is now taken over all the possible ``spin configurations'' 
$\sigma=(\sigma(x))_{x\in\Lao}$ with $\sigma(x)=\up,\dn$.

For $x\in\La'$ and $\Phi$ of the form (\ref{Phiexpand2}), the 
condition (\ref{bbP=0}) becomes
\begin{equation}
	\sumtwo{y,z\in\Lao(x)}{y>z}
	(b_{y,\up}b_{z,\dn}-b_{z,\up}b_{y,\dn})\Phi=0,
	\label{bbP=04}
\end{equation}
where $\Lao(x)=\set{y\in\Lao}{|y-x|=\sqrt{\nu}/2}$, and we ordered 
this set in an arbitrary manner.
Since any site $x\in\Lao$ is ``occupied'' in the representation 
(\ref{Phiexpand2}), the condition (\ref{bbP=04}) is satisfied only 
when we have\footnote{
This is only true when the electron number is $L^{d}=|\Lao|$.
We treated only the special models with $\nu=1$ in \cite{92e,93d}, 
where this step can be extended to other electron numbers.
}
\begin{equation}
	(b_{y,\up}b_{z,\dn}-b_{z,\up}b_{y,\dn})\Phi=0,
	\label{bbP=05}
\end{equation}
for any $y,z\in\Lao(x)$ with $y\ne z$ for some $x\in\La'$.

By substituting the expansion (\ref{Phiexpand2}) into the condition 
(\ref{bbP=05}), we find that the coefficients satisfy
 \begin{equation}
 	g(\sigma)=g(\sigma_{y,z}),
 	\label{g=g}
 \end{equation}
 where $\sigma_{y,z}$ is the spin configuration obtained by switching 
 $\sigma(y)$ and $\sigma(z)$ in the original $\sigma$.
The relation (\ref{g=g}) along with the expansion  (\ref{Phiexpand2})
implies that $\Phi$ can be written as 
\begin{equation}
	\Phi=\sum_{M=0}^{L^{d}}\alpha_{M}(S^{-}_{\rm tot})^{M}\UP,
	\label{Phinal}
\end{equation}
with suitable coefficients $\alpha_{M}$.
Here $S^{-}_{\rm tot}=S^{(1)}_{\rm tot}-iS^{(2)}_{\rm tot}$ is the 
spin-lowering operator.
This proves the desired theorem.
\end{proof}

\subsection{Basis for the Many-Electron System}
\label{secBasis}
We shall introduce a basis for describing many-electron problems.

Let $s\in\La$, and let $A\subset\La$ be a subset with $\abs{A}=L^d-1$.
We define
\begin{equation}
\Psik{s,A}=\sum_{x\in\Lao}\eikx\, T_x\sbk{
\ad_{s,\dn}\rbk{\prod_{t\in A}\ad_{t,\up}}\vac},
\label{PsiA}
\end{equation}
where $T_x$ is the translation operator (\ref{Tx}).
The state $\Psik{s,A}$
is an element of the Hilbert space $\calH_k$ (\ref{Hilbk})
of the states with momentum $k$ and a single down-spin
electron.
Clearly $\Psik{s,A}$ with different $(s,A)$ can define the
same state.
For $s\in\Lau$ (with a $u\in\calU$), one can take a unique
$y\in\Lao$ such that $s=u-y$.
Then we have
\begin{equation}
\Psik{s,A}=e^{i\theta}\Psik{u,A+y},
\label{PsiID}
\end{equation}
where $A+y=\set{x+y}{x\in A}$, and $\theta\in{\bf R}$.

Let $\Laob=\Lao\backslash\cbk{o}$.
We define
\begin{equation}
\Ok=\frac{1}{\alpha(k)}\Psik{o,\Laob},
\label{Omega}
\end{equation}
where $\alpha(k)>0$ is a real function of $k$ which will be determined
later in the proof.
We note that $\Ok$ is our approximate spin-wave excitation, which 
plays the central role in our proof.

Finally, we define our basis $\Bk$ for the space $\calH_k$ as
\begin{equation}
\Bk=
\cbk{\Ok}\cup
\set{\Psik{u,A}}
{\mbox{$u\in\calU$, $A\subset\La$ with 
$\abs{A}=L^d-1$, and $(u,A)\ne(o,\Laob)$}}.
\label{Bk}
\end{equation}
\Section{Proof of the Main Theorems}
\label{SecProof}
In the present section, we shall describe the proof of our
main theorems on the stability of ferromagnetism and
the lower bound for the spin-wave dispersion relation.
We make use of various estimates which will be proved
in the latter sections.
\subsection{Basic Lemma}
\label{secblemma}
Let us define 
\begin{equation}
\Htil=\Hhopt+\Hint,
\label{Htil}
\end{equation}
where $\Hhopt$ is the modified hopping Hamiltonian
(\ref{Hhoptil}), and 
$\Hint$ is the standard interaction Hamiltonian (\ref{Hint}).
For basis states $\Phi,\Psi\in\Bk$, we define the matrix elements
$h[\Psi,\Phi]\in{\bf C}$ of the Hamiltonian $\Htil$ above by the 
unique expansion
\begin{equation}
\Htil\Phi=\sum_{\Psi\in\Bk}h[\Psi,\Phi]\Psi.
\label{matrixelement}
\end{equation}
Note that only states from $\Bk$ with a fixed $k$ appear in the
right-hand side of (\ref{matrixelement}) since $\Htil$ is translation
invariant and the momentum $k$ is conserved.

For $\Phi\in\Bk$, we define
\begin{equation}
D[\Phi]={\rm Re}\sbk{h[\Phi,\Phi]}-
\sum_{\Psi\in\Bk\bs\cbk{\Phi}}\abs{h[\Phi,\Psi]}.
\label{DPhi}
\end{equation}
Then we have the following lemma.
The basic statement is elementary and
wellknown (in standard linear algebra), but it serves as a 
basis of our proof.
\begin{lemma}
Let $\ESW$ be the energy of the spin-wave excitation defined in
Section~\ref{secsw}.
Then for each $k\in\calK$, we have
\begin{equation}
\ESW\ge\min_{\Phi\in\Bk}D[\Phi].
\label{ESW>D}
\end{equation}
\label{DPhiLemma}
\end{lemma}
\begin{proof}{Proof}
Let $\widetilde{E}(k)$ be the lowest eigenvalue of $\Htil$ in the 
Hilbert space $\calH_k$ (\ref{Hilbk}).
We first claim that $\ESW\ge\widetilde{E}(k)$.
This is a straightforward consequence of the operator inequality
$H\ge\Htil$ (which follows from (\ref{Htil}) and (\ref{H>Htil}))
and the fact that both $H$ and $\Htil$ commute with $T_x$ ($x\in\Zd$)
and $\Sztot$.

Thus the desired bound (\ref{ESW>D}) follows from the inequality
\begin{equation}
\widetilde{E}(k)\ge\min_{\Phi\in\Bk}D[\Phi],
\label{E>D}
\end{equation}
which is indeed a straightforward consequence of 
a wellknown relation in elementary linear algebra.
To show (\ref{E>D}), let $E$ be an eigenvalue of $\Htil$,
and $\Phi_0\in\calH_k$ be the corresponding eigenstate.
We expand $\Phi_0$ as 
$\Phi_0=\sum_{\Psi\in\Bk}C(\Psi)\,\Psi$ where
$C(\Psi)$ are coefficients.
From (\ref{matrixelement}) and the eigenvalue equation
$\Htil\Phi_0=E\Phi_0$, we find that $C(\Psi)$ satisfy
\begin{equation}
E\,C(\Phi)=\sum_{\Psi\in\Bk}h[\Phi,\Psi]\,C(\Psi)
\end{equation}
for any $\Phi\in\Bk$.
Let $\Phi'\in\Bk$ be the state such that
$\abs{C(\Psi)/C(\Phi')}\le1$ holds for any $\Psi\in\Bk$.
Then we have
\begin{eqnarray}
E&=&\sum_{\Psi\in\Bk}h[\Phi',\Psi]\frac{C(\Psi)}{C(\Phi')}\ret
&\ge&{\rm Re}\sbk{h[\Phi',\Phi']}
-\sum_{\Psi\in\Bk\bs\cbk{\Phi'}}\abs{h[\Phi',\Psi]}\ret
&=&D[\Phi'].
\end{eqnarray}
Since $\widetilde{E}(k)$ is the smallest eigenvalue,
the desired inequality (\ref{E>D}) follows.
\end{proof}

Being a very crude bound, we cannot expect (\ref{ESW>D})
to yield meaningful results unless we use a basis which 
``almost diagonalizes'' the low energy part of the Hamiltonian.
As we shall see below, it turns out that the basis we constructed
in Section~\ref{secBasis} indeed have such properties.

\subsection{Estimates of the Matrix Elements}
\label{secmatrix}
We shall summarize the result of Sections~\ref{Secmatrep}
and \ref{SecMatrixBound}
where we estimate various matrix elements.

Before stating the results, it is convenient to introduce new
labeling of the special states $\Psi\in\Bk$
which have nonvanishing matrix elements $h[\Psi,\Ok]$.
For any $u\in\calU$ and $r\in\Lao$, we define
\begin{equation}
\Phik{u,r}=\sum_{x\in\Lao}\eikx\ad_{x+u,\dn}b_{x+r,\up}\UP,
\label{Phiur}
\end{equation}
where $\UP=\rbk{\prod_{y\in\Lao}\ad_{y,\up}}\vac$ is the
``ferromagnetic ground state''.
By noting that 
\newline
$T_x\sbk{\prod_{y\in\Lao}\ad_{y,\up}}=\prod_{y\in\Lao}\ad_{y,\up}$,
we can relate the state (\ref{Phiur}) with the general state 
$\Psik{s,A}$ (\ref{PsiA}) as
\begin{eqnarray}
	\Phik{u,r}
	&=&
	\sum_{x\in\Lao}\eikx\,
	T_x\sbk{\ad_{u,\dn}b_{r,\up}\rbk{\prod_{y\in\Lao}\ad_{y,\up}}\vac}
	\ret
  &=&
	{\rm sgn}[r]
	\sum_{x\in\Lao}\eikx\,
	T_x\sbk{\ad_{u,\dn}\rbk{\prod_{y\in\Lao\bs\cbk{r}}\ad_{y,\up}}\vac}
	\ret
	&=&
	{\rm sgn}[r]\,\Psik{u,\,\Lao\bs\cbk{r}},
	\label{Phi=Psi}
\end{eqnarray}
where ${\rm sgn}[r]=\pm1$.
By using (\ref{Phi=Psi}), we can rewrite (\ref{Omega}) as
\begin{equation}
\Ok=\frac{1}{\alpha(k)}\Phik{o,o}.
\label{Ok=Phi}
\end{equation}

Let $\Laop\subset\Lao$ be a special subset with the property 
that for any $s,t\in\Lao$ such that $s\ne t$, we have either 
$s-t\in\Laop$ or $t-s\in\Laop$ (and not both).
An example is
\begin{eqnarray}
\Laop&=&\set{x\in\Lao}{0<x_1\le\frac{L-1}{2}}
\cup\set{x\in\Lao}{x_1=0,\ 0<x_2\le\frac{L-1}{2}}\cup\ret
&&
\cup\set{x\in\Lao}{x_1=x_2=0,\ 0<x_3\le\frac{L-1}{2}}\cup\cdots\ret
&&
\cdots\cup
\set{x\in\Lao}{x_1=x_2=\cdots=x_{d-1}=0,\ 0<x_d\le\frac{L-1}{2}}.
\label{La+o}
\end{eqnarray}
For $u\in\calU$, $r\in\La'$, and $s,t\in\Lao$ such that $s-t\in\Laop$,
we define 
\begin{equation}
\Phik{u,r,t,s}=\sum_{x\in\Lao}\eikx\,
\ad_{x+u,\dn}\ad_{x+r,\up}b_{x+t,\up}b_{x+s,\up}\UP.
\label{Phiurts}
\end{equation}

It can be shown that the only states $\Psi\in\Bk$ such that
$h[\Psi,\Ok]\ne0$ can be written in the form 
$\Phik{u,r}$ or $\Phik{u,r,t,s}$ with suitable $u$, $r$, $t$, and $s$.
See Section~\ref{Secmatrep}.

By using the representations (\ref{HtilhopRep}) and (\ref{HintRep})
for the Hamiltonians, we can express the matrix elements 
$h[\Psi,\Phi]$ explicitly in terms of the effective hopping $\tau_{x,y}$
and the effective interaction $\Ut_{y,v;w,z}$.
We leave the derivation to Section~\ref{Secmatrep}, and 
summarize the results as the following lemma.
\begin{lemma}
\label{matrepLemma}
For any $u$, $r$, $t$, and $s$ as in (\ref{Phiur}) or (\ref{Phiurts}), 
we have
\begin{equation}
h[\Ok,\Ok]=E_0+2\sum_{s\in\Lao}\rbk{\sin\frac{k\cdot s}{2}}^2
\Ut_{s,o;s,o},
\label{h1}
\end{equation}
\begin{equation}
h[\Ok,\Phik{u,r}]=\delta_{u,o}\,\alpha(k)(e^{-ik\cdot r}-1)\,\tau_{r,o}
+\alpha(k)\sum_{s\in\Lao}(e^{-ik\cdot r}-e^{-ik\cdot s})\Ut_{s,r;u,s},
\label{h2}
\end{equation}
\begin{equation}
h[\Ok,\Phik{u,r,t,s}]=\alpha(k)(e^{-ik\cdot s}-e^{-ik\cdot t})\Ut_{s,t;u,r},
\label{h3}
\end{equation}
\begin{equation}
h[\Phik{u,r},\Ok]=\delta_{u,o}\frac{1}{\alpha(k)}(e^{ik\cdot r}-1)\tau_{o,r}
+\frac{1}{\alpha(k)}\sum_{s\in\Lao}(e^{ik\cdot r}-e^{ik\cdot s})\Ut_{u,s;s,r},
\label{h4}
\end{equation}
and
\begin{equation}
h[\Phik{u,r,t,s},\Ok]=\frac{1}{\alpha(k)}
(e^{ik\cdot s}-e^{ik\cdot t})\Ut_{u,r;s,t}.
\label{h5}
\end{equation}
It should be noted that these matrix elements are not symmetric,
reflecting that the basis $\Bk$ is not orthonormal.
\end{lemma}

By combining the expressions (\ref{h1})-(\ref{h5}), the bounds 
(\ref{regtau1}), (\ref{regtau2}) for $\tau_{x,y}$,
the representation (\ref{Util}) for $\Ut_{y,v;w,z}$ in terms of
the basis states $\phis{x}{y}$, $\phit{x}{y}$,
and the bounds (\ref{reg1})-(\ref{reg4}), (\ref{reg5})-(\ref{reg8})
for these states, we can derive explicit bounds for the matrix elements
and their sums.
Again we leave all the derivations to Section~\ref{SecMatrixBound},
and summarize the results as the following lemma.
\begin{lemma}
\label{meb1Lemma}
Under the assumptions that $\la\ge\la_0$ and $\akappa\le\kappa_0$,
we have
\begin{equation}
{\rm Re}\sbk{h[\Ok\Ok]}
\ge
E_0+\frac{U}{\la^4}\rbk{1-C_1\akappa-\frac{C_2}{\la}}G(k),
\label{hOO>}
\end{equation}
with $G(k)$ defined in (\ref{Gk}),
\begin{equation}
\sum_{\Psi\in\Bk\bs\cbk{\Ok}}\abs{h[\Ok,\Psi]}
\le
\alpha(k)\rbk{B_1R\akappa t+\frac{C_3U}{\la^2}}\abs{k},
\label{hOPsi<}
\end{equation}
\begin{equation}
\abs{h[\Phik{u,r},\Ok]}
\le
\frac{1}{\alpha(k)}
\rbk{B_1R\akappa t +\frac{C_4U\akappa}{\la^2}+\frac{C_5U}{\la^3}}
\abs{k},
\label{hPur}
\end{equation}
and
\begin{equation}
\abs{h[\Phik{u,r,t,s},\Ok]}
\le
\frac{1}{\alpha(k)}\frac{C_6U}{\la^2}\abs{k}.
\label{hPurts}
\end{equation}
Here $C_i$ ($i=1,2,3,4,5,6$) are positive constants which depend only
on $d$, $\nu$, and $R$.
\end{lemma}

We can perform similar analysis for the matrix elements which do 
not involve the state $\Ok$.
For $\Phi\in\Bk\bs\cbk{\Ok}$, we define
\begin{equation}
\Dt[\Phi]=
{\rm Re}\sbk{h[\Phi,\Phi]}-\sum_{\Psi\in\Bk\bs\cbk{\Phi,\Ok}}
\abs{h[\Phi,\Psi]}.
\label{Dtil}
\end{equation}
Then we prove the following in Section~\ref{secDtil}.
\begin{lemma}
\label{meb2Lemma}
Assume that 
$\la\ge\la_4$, $\akappa\le\kappa_0$,
and $K_3t\akappa\le U\le K_4\la^3t$,
where $\la_4$, $\kappa_0$, $K_3$, and $K_4$ are constants which depend only
on $d$, $\nu$, and $R$.
Then we have
\begin{equation}
\Dt[\Psik{u,A}]\ge E_0+\frac{\la^2t}{2},
\label{Dtil>}
\end{equation}
for any $u\in\calU$ and $A\subset\La$ such that
$\abs{A}=L^d-1$ and $A\cap\La'\ne\emptyset$,
\begin{equation}
\Dt[\Phik{u,r}]\ge E_0+\frac{\la^2t}{2},
\label{Dtilur>}
\end{equation}
for $u\ne o$, and
\begin{equation}
\Dt[\Phik{o,r}]\ge E_0+\frac{U}{2},
\label{Dtilor>}
\end{equation}
for $r\ne o$.
\end{lemma}

\subsection{Proof of Theorem~\protect\ref{Ek>Th}}
We will now prove Theorem~\ref{Ek>Th} for the lower bound 
of the spin-wave excitation energy,
which is one of most important results.
In the proof, we make use of Lemmas~\ref{meb1Lemma}
and \ref{meb2Lemma}.
We will later confirm that the conditions for these Lemmas are
satisfied.

From Lemma~\ref{DPhiLemma},
we find that the desired lower bound (\ref{ESW>})
follows if we show
\begin{equation}
D[\Phi]\ge E_0+F_2\frac{U}{\la^4}G(k),
\label{DPhi>}
\end{equation}
for any $\Phi\in\Bk$.
(Recall that $E_0=\Emin(\Smax)$.)

We shall first verify (\ref{DPhi>}) for $\Phi=\Psik{u,A}\in\Bk$
such that 
\begin{equation}
h[\Psik{u,A},\Ok]=0.
\label{hPO=0}
\end{equation}
Then comparing the definitions (\ref{DPhi}) and (\ref{Dtil}),
we find $D[\Psik{u,A}]=\Dt[\Psik{u,A}]$ for such $\Psik{u,A}$.
We also claim that the condition (\ref{hPO=0}) inevitably implies
$A\cap\La'\ne\emptyset$.
To see this, we note that the converse $A\cap\La'=\emptyset$
means $A=\Lao\bs\cbk{r}$ for some $r\in\Lao$, and hence
$\Psik{u,A}$ is equal to $\pm\Phik{u,r}$.
(See (\ref{Phi=Psi}).)
Therefore we can use the lower bound (\ref{Dtil>}) to find
\begin{equation}
D[\Psik{u,A}]=\Dt[\Psik{u,A}]\ge E_0+\frac{\la^2t}{2}
\ge E_0+F_2\frac{U}{\la^4}G(k),
\label{DPuA>}
\end{equation}
where the final bound is derived by
noting that $G(k)\le2^{2\nu+1}\dnu$ and $F_2\le1$,
and assuming that
\begin{equation}
\frac{\la^6t}{U}\ge2^{\nu+2}\dnu.
\label{cond1}
\end{equation}
Therefore the desired inequality (\ref{DPhi>}) is verified for
$\Phi=\Psik{u,A}$ such that (\ref{hPO=0}) holds.

Next we examine the inequality  (\ref{DPhi>}) for the states
which do not satisfy the condition (\ref{hPO=0}).
They are the states $\Ok$, $\Phik{u,r}$, and $\Phik{u,r,t,s}$
defined in (\ref{Omega}) (see also (\ref{Ok=Phi})), 
(\ref{Phiur}), and (\ref{Phiurts}),
respectively.

As for the state $\Phik{u,r}$, we use the definitions (\ref{DPhi}),
(\ref{Dtil}), and the bounds (\ref{Dtilur>}), (\ref{hPur}) to get
\begin{eqnarray}
D[\Phik{u,r}]
&=&
\Dt[\Phik{u,r}]-\abs{h[\Phik{u,r},\Ok]}
\ret
&\ge&
E_0+\frac{U}{2}-\frac{1}{\alpha(k)}
\rbk{B_1R\akappa t+\frac{C_4U\akappa}{\la^2}+\frac{C_5U}{\la^3}}
\abs{k}.
\label{DPur>}
\end{eqnarray}
Let us choose $\alpha(k)$ as
\begin{equation}
\alpha(k)=\frac{4\abs{k}}{U}
\rbk{B_1R\akappa t+\frac{C_4U\akappa}{\la^2}+\frac{C_5U}{\la^3}}.
\label{alpha}
\end{equation}
Then (\ref{DPur>}) becomes
\begin{equation}
D[\Phik{u,r}]\ge E_0+\frac{U}{4}\ge E_0+F_2\frac{U}{\la^4}G(k).
\label{DPur>2}
\end{equation}
To get the final bound, we have made a further assumption that
\begin{equation}
\la^4\ge2^{2\nu+3}\dnu.
\label{cond2}
\end{equation}
We have shown the desired bound (\ref{DPhi>}) for $\Phi=\Phik{u,r}$.

The state $\Phik{u,r,t,s}$ (where we require $r\in\La'$) satisfies the condition
for the bound (\ref{Dtil>}).
By combining  (\ref{Dtil>}) with the bound (\ref{hPurts}), 
and using (\ref{alpha}), we have
\begin{eqnarray}
D[\Phik{u,r,t,s}]&=&
\Dt[\Phik{u,r,t,s}]-\abs{h[\Phik{u,r,t,s},\Ok]}
\ret
&\ge&
E_0+\frac{\la^2t}{2}-\frac{1}{\alpha(k)}\frac{C_6U}{\la^2}\abs{k}
\ret
&=&
E_0+\frac{\la^2t}{2}
-\frac{C_6U}
{4\la^2\rbk{B_1R\akappa t+\frac{C_4U\akappa}{\la^2}
+\frac{C_5U}{\la^3}}}
\ret
&\ge&
E_0+\frac{\la^2t}{2}-\frac{C_6\la}{4C_5}U
\ret
&\ge&
E_0+\frac{\la^2t}{4}
\ret
&\ge&
E_0+F_2\frac{U}{\la^4}G(k),
\label{DPurts>}
\end{eqnarray}
where, to get the final bound, we required
\begin{equation}
0\le U\le K_2\la t
\label{U<lambdat}
\end{equation}
with $K_2=C_5/C_6$, and
\begin{equation}
\frac{\la^6t}{U}\ge2^{2\nu+3}\dnu.
\label{cond3}
\end{equation}
We have shown the bound (\ref{DPhi>}) for
$\Phi=\Phik{u,r,t,s}$.

Finally we examine the state $\Ok$,
which is our trial state for the elementary spin-wave excitation.
By using the bounds (\ref{hOO>}) and (\ref{hOPsi<}),
and the choice (\ref{alpha}) of $\alpha(k)$,
we get
\begin{eqnarray}
D[\Ok]
&\ge&
E_0+\frac{U}{\la^4}\rbk{1-C_1\akappa-\frac{C_2}{\la}}G(k)
-\alpha(k)\rbk{B_1R\akappa t + \frac{C_3U}{\la^2}}\abs{k}
\ret
&=&
E_0+\frac{U}{\la^4}\rbk{1-C_1\akappa-\frac{C_2}{\la}}G(k)
\ret
&&
-4\frac{U}{\la^4}\rbk{C_3+B_1R\frac{\la^2t\akappa}{U}}
\rbk{C_4\akappa+\frac{C_5}{\la}+B_1R\frac{\la^2t\akappa}{U}}\abs{k}^2
\ret
&\ge&
E_0+
\frac{U}{\la^4}\rbk{1-A_1\akappa-\frac{A_2}{\la}-A_3\frac{\la^2t\akappa}{U}}
G(k)
\ret
&=&
E_0+F_2\frac{U}{\la^4}G(k)
\end{eqnarray}
with suitable positive constants $A_1$, $A_2$, and $A_3$.
Here we used the bound
\begin{equation}
\abs{k}^2\le\pi^2\sum_{i=1}^d\rbk{\sin\frac{k_i}{2}}^2
\le\frac{\pi^2}{4}G(k),
\label{k2<Gk}
\end{equation}
which follows from $\abs{k_i}\le\pi$,
and further assumed that
\begin{equation}
A_3\frac{\la^2t\akappa}{U}\le1.
\label{cond4}
\end{equation}
We have thus confirmed the desired bound (\ref{DPhi>}) for 
all $\Phi\in\Bk$.
This means that the desired lower bound (\ref{ESW>}) for the
spin-wave excitation energy has been proved.

It remains to examine the conditions for the model parameters
assumed in the proof.
The  assumptions made during the proof are
(\ref{cond1}), (\ref{cond2}), (\ref{U<lambdat}), (\ref{cond3}), and
(\ref{cond4}).
Among them  (\ref{U<lambdat}) and (\ref{cond4}) are explicitly
assumed in the statement of the theorem.

Since we shall choose $\la_3$ so that $\la_3\ge\la_4\ge\la_0$,
the conditions about $\la$ and $\kappa$ stated in
Lemmas~\ref{meb1Lemma} and \ref{meb2Lemma}
are satisfied.

Let us set
\begin{equation}
\la_3=\max\cbk{
\la_0,\la_4,\rbk{K_2\,2^{\nu+3}\dnu}^{1/5},\rbk{2^{2\nu+3}\dnu}^{1/4},
(K_2/K_4)^{1/2},(K_3/A_3)^{1/2}
}.
\label{la3}
\end{equation}
From the assumption $\la\ge\la_3$ (with the above $\la_3$)
and the assumed (\ref{U<lambdat}) and (\ref{cond4}),
we can verify that the conditions
(\ref{cond1}), (\ref{cond2}), (\ref{cond3}), and
$K_3\akappa t\le U\le K_4\la^3 t$ 
(which is required in Lemma~\ref{meb2Lemma})
are satisfied.
Finally the conditions $\la\ge\la_0$ and $\la\ge\la_4$ required in
Lemma~\ref{meb1Lemma} and Lemma~\ref{meb2Lemma}, respectively, are
satisfied since $\la\ge\la_4\ge\la_0$. 
This completes the proof of the theorem.
\subsection{Proof of Theorem~\protect\ref{stabilityTh}}
\label{secproofstab}
We now prove our main theorem which states the local stability
of the ferromagnetic ground states.

Theorem~\ref{stabilityTh} follows from the following statement
which has more  general (but more complicated)
conditions.

\begin{lemma}
\label{stLemma}
The local stability 
inequality (\ref{stabineq}) is valid if either i) or ii) below
is satisfied.
\par\noindent
i) $\la\ge\la_3$,
$\akappa\le\kappa_0$,
$A_1\abs{\kappa}+A_2\la^{-1}+A_3\la^2t\akappa U^{-1}<1$,
and 
$0<U\le K_2\la t$,
\par\noindent
ii)  $\la\ge\la_3$,
$\akappa\le\kappa_0$
$A_1\abs{\kappa}+A_2\la^{-1}+A_3\la\akappa (K_2)^{-1}<1$,
and 
$U\ge K_2\la t$.
\par\noindent
The constants 
$A_1$, $A_2$, and $A_3$ are the same as those appeared 
in Theorem~\ref{Ek>Th}.
\end{lemma}
The Lemma actually is the  most natural way of expressing our stability
theorem. 
The conditions (\ref{st1}) and (\ref{st2}) in Theorem~\ref{stabilityTh}
were introduced to give an easily accessible sufficient condition
for the conditions i) or ii) in Lemma~\ref{stLemma}.
\bigno
\begin{proof}{Proof of Theorem~\ref{stabilityTh},
given Lemma~\ref{stLemma}}
We set 
$\la_2=\max\cbk{\la_3,A_2/4}$, 
\newline
$\kappa_1=\min\cbk{\kappa_0,(4A_1)^{-1}}$, $p_1=K_2(4A_3)^{-1}$,
and $K_1=4A_3$.
Suppose that the conditions in Theorem~\ref{stabilityTh}
are satisfied.

We first assume $0\le U\le K_2\la t$.
Then we have
$A_1\akappa\le1/4$, $A_2/\la\le1/4$, and
$A_3\la^2t\akappa/U\le1/4$.
It is obvious that all the conditions in i) are satisfied.

Next we assume $U\ge K_2\la t$.
Again we have 
$A_1\akappa\le1/4$, $A_2/\la\le1/4$, and
$A_3\la \akappa/K_2\le1/4$.
The conditions in ii) are satisfied.
\end{proof}

In what follows we prove Lemma~\ref{stLemma}.

For each state $\Phi$ which is an eigenstate of 
$({\bf S}_{\rm tot})^2$ with $\Stot=\Smax-1$,
we can take its $SU(2)$ rotation $\widetilde{\Phi}$
which satisfies 
$\Sztot\widetilde{\Phi}=(\Smax-1)\widetilde{\Phi}$.
Since $\Phi$ and $\widetilde{\Phi}$ have the same energy,
it suffices to concentrate on the space
\begin{equation}
\calH_{\Smax-1}=\set{\Phi}{\Sztot\Phi=(\Smax-1)\Phi},
\label{spacedecomp}
\end{equation}
and prove the stability theorem.
By using $\calH_k$ defined in (\ref{Hilbk}),
the above space is decomposed as
\begin{equation}
\calH_{\Smax-1}=\bigoplus_{k\in\calK}\calH_k.
\end{equation}

We first assume that the condition i) in Lemma~\ref{stLemma} is satisfied.
Then the assumptions of Theorem~\ref{Ek>Th} are
automatically satisfied, and we also have $F_2>0$.
Thus for any $k\in\calK$ such that
$k\ne o=(0,\ldots,0)$ the lowest energy $\ESW$ 
in the sector $\calH_k$ satisfies
\begin{equation}
\ESW>E_0=\Emin(\Smax).
\end{equation}
Recalling the decomposition (\ref{spacedecomp}),
one finds that this proves the desired bound
\newline
$\Emin(\Smax-1)>\Emin(\Smax)$
except in the sector $\calH_o$.

To deal with the sector $\calH_o$ is not hard.
We note that the state $\Omega(o)$ is written as
\begin{equation}
\Omega(o)=\frac{1}{\alpha(k)}S^-_{\rm tot}\UP,
\label{O=S-P}
\end{equation}
where $S^-_{\rm tot}=S^{(1)}_{\rm tot}-i S^{(1)}_{\rm tot}$ 
is the spin lowering operator.
This means that $\Omega(o)$ is nothing but one of
the ``ferromagnetic ground states'', and has the total
spin $\Stot=\Smax$.
Let $\Emin(\Smax-1,o)$ be the lowest energy in the sector
$\calH_o$ with $\Stot=\Smax-1$.
Then, by repeating the argument in the proof of 
Lemma~\ref{DPhiLemma}, we find that
\begin{equation}
\Emin(\Smax-1,o)\ge\min_{\Phi\in\calB_o\bs\cbk{\Omega(o)}}
D[\Phi].
\end{equation}
The right-hand side can be bounded from below 
by using the inequalities (\ref{DPuA>}), (\ref{DPur>2}),
and (\ref{DPurts>}).
We get  
\begin{equation}
\Emin(\Smax-1,o)
\ge E_0+\min\cbk{\frac{\la^2t}{2},\frac{U}{4},\frac{\la^2t}{4}}
>E_0,
\end{equation}
which completes the proof of the desired local stability inequality
(\ref{stabineq}).

The only remaining task is to prove the inequality (\ref{stabineq})
when the condition ii) in Lemma~\ref{stLemma} is satisfied\footnote{
The following argument has been brought to the author by
Andreas Mielke.
}.
Note that $U$ is not bounded from above in this case.

The key ingredient in the extension is to realize that
$\Emin(\Smax)$ does not depend on $U$, while
$\Emin(\Smax-1)$ is increasing in $U$.
The latter fact follows by noting that $\Hint$ is increasing in
$U$ (as an operator), both $\Hhop$ and $\Hint$ commute with
the total spin operator, and $\Emin(\Smax-1)$ is the lowest
energy in the  sector with the fixed $\Stot$.

Suppose that the condition ii) in the Remark after 
Theorem~\ref{stabilityTh} is satisfied.
Then by setting $U=K_2\la t$, the condition i) in Lemma~\ref{stLemma}
is satisfied, and we have $\Emin(\Smax-1)>\Emin(\Smax)$.
Because of the increasing property of $\Emin(\Smax-1)$,
this inequality remains valid if we increase $U$ with other
parameters kept fixed.
This proves the local stability inequality (\ref{stabineq}).

\Section{Representation of the Matrix Elements}
\label{Secmatrep}
Here we will prove Lemma~\ref{matrepLemma}
about the representation of the matrix elements involving the states
$\Ok$, $\Phik{u,r}$, and $\Phik{u,r,t,s}$.

\subsection{Treatment of the Hopping Hamiltonian}
\label{secMRhop}
By operating $\Hhopt$ in the form (\ref{HtilhopRep}) to the
state $\Phik{o,r}$ (see (\ref{Phiur})),
and using the anticommutation relations
(\ref{ac3}), (\ref{ac4}), we get
\begin{eqnarray}
\Hhopt\Phik{o,r}
&=&
\rbk{
\sumtwo{y,v\in\Lao}{\sigma=\up,\dn}\tau_{y,v}\,
\ad_{y,\sigma}\ad_{v,\sigma}
}
\sum_{p\in\Lao}\eikp\ad_{p,\dn}b_{p+r,\up}\UP
\ret
&=&
-\sum_{p,v\in\Lao}\tau_{p+r,v}\,\eikp
\ad_{p,\dn}b_{v,\up}\UP
\ret
&&
+\sum_{p,y\in\Lao}\tau_{y,y}\,\eikp
\ad_{p,\dn}b_{p+r,\up}\UP
\ret
&&
+\sum_{p,y\in\Lao}\tau_{y,p}\,\eikp
\ad_{y,\dn}b_{p+r,\up}\UP.
\end{eqnarray}
We shall make the change of variables
$p=x$, $v=x+s$ (with $x,s\in\Lao$) in the first term,
and the change of variables 
$y=x$, $p=x+s-r$ (with $x,s\in\Lao$) in the second term.
By also using (\ref{E0}), we have
\begin{eqnarray}
\Hhopt\Phik{o,r}
&=&
E_0\Phik{o,r}
-\sum_{x,s\in\Lao}\tau_{r,s}\eikx\ad_{x,\dn}b_{x+s,\up}\UP
\ret
&&
+\sum_{x,s\in\Lao}\tau_{o,s-r}\eik{(s-r)}\eikx
\ad_{x,\dn}b_{x+s,\up}\UP
\ret
&=&
E_0\Phik{o,r}
+\sum_{s\in\Lao}\tau_{r,s}\rbk{\eik{(s-r)}-1}
\Phik{o,s},
\label{HopPhior}
\end{eqnarray}
where we made use of the translation invariance of $\tau_{x,y}$.
Following the definition (\ref{matrixelement})
of matrix elements, we define the matrix elements $\hhop[\Psi,\Phi]$
by the unique expansion
\begin{equation}
\Hhopt\Phi=\sum_{\Psi\in\Bk}\hhop[\Psi,\Phi]\Psi.
\label{matrixelement2}
\end{equation}
By comparing (\ref{HopPhior}) with this definition, 
we find
\begin{equation}
\hhop[\Phik{o,s},\Phik{o,r}]=
\delta_{r,s}\,E_0+\tau_{r,s}\rbk{\eik{(s-r)}-1}.
\label{hhop0}
\end{equation}
By recalling $\Ok=\alpha(k)^{-1}\,\Phik{o,o}$, (\ref{hhop0})
yields
\begin{equation}
\hhop[\Ok,\Ok]=E_0,
\label{hhop1}
\end{equation}
\begin{equation}
\hhop[\Ok,\Phik{o,r}]=\alpha(k)\rbk{\emik{r}-1}\tau_{r,o},
\label{hhop2}
\end{equation}
and
\begin{equation}
\hhop[\Phik{o,r},\Ok]=\frac{1}{\alpha(k)}\rbk{\eik{r}-1}\tau_{o,r}.
\label{hhop3}
\end{equation}

\subsection{Treatment of the Interaction Hamiltonian}
\label{secMRint}
Before calculating the matrix elements of the interaction
Hamiltonian, we recall the representation (\ref{HintRep}),
and  decompose it as 
$\Hint=\Hinto+\Hintt$ with
\begin{equation}
\Hinto=
\sum_{y\in\Lao}\sum_{v,w,z\in\La}
\Ut_{y,v;w,z}\,\ad_{y,\up}\ad_{v,\dn}b_{w,\dn}b_{z,\up},
\label{Hint1}
\end{equation}
and
\begin{equation}
\Hintt=
\sum_{y\in\La'}\sum_{v,w,z\in\La}
\Ut_{y,v;w,z}\,\ad_{y,\up}\ad_{v,\dn}b_{w,\dn}b_{z,\up}.
\label{Hint2}
\end{equation}
Note that $\Hinto$ and $\Hintt$ are not hermitian.

We apply $\Hinto$ onto $\Phik{u,r}$, and simplify the expression by
using the anticommutation relations (\ref{ac3}), (\ref{ac4}) to get
\begin{eqnarray}
\Hinto\Phik{u,r}&=&
\sumtwo{y\in\Lao}{v,w,z\in\La}
\Ut_{y,v;w,z}\,\ad_{y,\up}\ad_{v,\dn}b_{w,\dn}b_{z,\up}
\sum_{p\in\Lao}\eikp\ad_{p+u,\dn}b_{p+r,\up}\UP
\ret
&=&
\sumtwo{v\in\La}{y,p\in\Lao}\Ut_{y,v;u+p,y}\,\eikp
\ad_{v,\dn}b_{p+r,\up}\UP
\ret&&
-\sumtwo{v\in\La}{z,p\in\Lao}\Ut_{p+r,v;p+u,z}\,\eikp
\ad_{v,\dn}b_{z,\up}\UP.
\label{Hint1Pur}
\end{eqnarray}
We note that $v\in\La$ can be uniquely decomposed as
$v=x+u'$ with $x\in\Lao$ and $u'\in\calU$.
we further make the change of variables $p=x+w-r$,
$y=x+s-r$ (with $w,s\in\Lao$) in the first term (in the right-hand
side of (\ref{Hint1Pur})),
and the change of variables
$z=x+w$, $p=x-s$ (with $w,s\in\Lao$) in the second term.
Then we get
\begin{eqnarray}
\Hinto\Phik{u,r}&=&
\sumtwo{u'\in\calU}{x,s,w\in\Lao}
\Ut_{x+s-r,x+u';x+u+w-r,x+s-r}\,\eik{(x+w-r)}
\ad_{x+u',\dn}b_{x+w,\up}\UP
\ret
&&
-\sumtwo{u'\in\calU}{x,s,w\in\Lao}
\Ut_{x+r-s,x+u';x+u-s,x+w}\,\eik{(x-s)}
\ad_{x+u',\dn}b_{x+w,\up}\UP
\ret
&=&
\sumtwo{u'\in\calU}{s,w\in\Lao}
\rbk{\Ut_{s,u'+r;u+w,s}\,\eik{(w-r)}-\Ut_{r,u'+s;u,w+s}\,\emik{s}}
\Phik{u',w},
\label{Hint1Pur2}
\end{eqnarray}
where we used the translation invariance of $\Ut_{y,v;w,z}$.

We again define the matrix elements $\hint[\Psi,\Phi]$
by the unique expansion
\begin{equation}
\Hint\Phi=\sum_{\Psi\in\Bk}\hint[\Psi,\Phi]\Psi.
\label{matrixelement3}
\end{equation}
Then  we can read off from (\ref{Hint1Pur2}) that
\begin{equation}
\hint[\Phik{u',w},\Phik{u,r}]
=\sum_{s\in\Lao}
\rbk{\Ut_{s,u'+r;u+w,s}\,\eik{(w-r)}-\Ut_{r,u'+s;u,w+s}\,\emik{s}}.
\label{hint1}
\end{equation}

By setting $u'=w=o$ in (\ref{hint1}), we get
\begin{equation}
\hint[\Ok,\Phik{u,r}]
=\alpha(k)\sum_{s\in\Lao}(\emik{r}-\emik{s})\Ut_{s,r;u,s}.
\label{hint2}
\end{equation}
Next, we set $u=r=o$ in (\ref{hint1}) to get
\begin{eqnarray}
\hint[\Phik{u',w},\Ok]
&=&\frac{1}{\alpha(k)}\sum_{s\in\Lao}
\rbk{\Ut_{s,u';w,s}\eik{w}-\Ut_{o,u'+s;o,w+s}\emik{s}}
\ret
&=&\frac{1}{\alpha(k)}\sum_{s\in\Lao}(\eik{w}-\eik{s})
\Ut_{u',s;s,w},
\label{hint3}
\end{eqnarray}
where in the second term, we used the translation invariance
and the symmetry as
\newline
$\Ut_{o,u'+s;o,w+s}=\Ut_{-s,u';-s,w}=\Ut_{u',-s;-s,w}$,
and then replaced $s\rightarrow-s$.
Finally we set $u=u'=r=w=o$ in (\ref{hint1}) to get
\begin{eqnarray}
\hint[\Ok,\Ok]
&=&\sum_{s\in\Lao}\rbk{1-\emik{s}}\Ut_{s,o;s,o}
\ret
&=&\frac{1}{2}\sum_{s\in\Lao}\rbk{2-\emik{s}-\eik{s}}\Ut_{s,o;s,o}
\ret
&=&2\sum_{s\in\Lao}\rbk{\sin\frac{k\cdot s}{2}}^2\Ut_{s,o;s,o},
\label{hint4}
\end{eqnarray}
where we used 
$\Ut_{s,o;s,o}=\Ut_{o,-s;o,-s}=\Ut_{-s,o;-s,o}$ which follows from
the translation invariance and the symmetry of $\Ut_{y,v;w,z}$.
Note that we do not assume any reflection invariance.

We are now ready to prove some of the expressions in 
Lemma~\ref{matrepLemma}.
The expression (\ref{h1}) follows by summing (\ref{hhop1})
and (\ref{hint4}), the expression (\ref{h2}) follows by
summing (\ref{hhop2}) and (\ref{hint2}), and the
expression (\ref{h4}) follows by summing (\ref{hhop3}) and (\ref{hint3}).

We next calculate the action of (\ref{Hint2}) as
\begin{eqnarray}
\Hintt\Phik{o,o}&=&
\sumtwo{y\in\La'}{v,w,z\in\La}
\Ut_{y,v;w,z}\,\ad_{y,\up}\ad_{v,\dn}b_{w,\dn}b_{z,\up}
\sum_{p\in\Lao}
\eikp\ad_{p,\dn}b_{p,\up}\UP
\ret
&=&
-\sumthree{y\in\La'}{v\in\La}{p,z\in\Lao}
\Ut_{u,v;p,z}\,\eikp
\ad_{v,\dn}\ad_{y,\up}b_{z,\up}b_{p,\up}\UP.
\label{Hint2oo}
\end{eqnarray}
In the final expression, we note that the summand is vanishing for $p=z$,
and  decompose the sum over $p,z$ as
\begin{equation}
\sumtwo{p,z\in\Lao}{p\ne z}(\cdots)
=\sumtwo{p,z\in\Lao}{p-z\in\Laop}(\cdots)
+\sumtwo{p,z\in\Lao}{z-p\in\Laop}(\cdots)
\label{pzdecomp}
\end{equation}
where $\Laop$ is defined in (\ref{La+o}).
We then switch 
the variables $z$ and $p$ in the second sum
to get
\begin{equation}
\Hintt\Phik{o,o}=
\sumtwo{p,z\in\Lao}{p-z\in\Laop}\rbk{\eikp-\eik{z}}
\sumtwo{v\in\La}{y\in\La'}
\Ut_{y,v;p,z}\,\ad_{v,\dn}\ad_{y,\up}b_{z,\up}b_{p,\up}\UP.
\label{Hint2oo2}
\end{equation}
We write $v=x+u$ with $x\in\Lao$ and $u\in\calU$,
and make the change of variables
$y=x+r$, $p=x+s$, and $z=x+t$ 
(with $r\in\La'$, $s,t\in\Lao$ such that $s-t\in\Laop$)
to get
\begin{eqnarray}
\Hintt\Phik{o,o}&=&
\sumfour{u\in\calU}{r\in\La'}{s,t\in\Lao}{s-t\in\Laop}
\rbk{\eik{s}-\eik{t}}\Ut_{r,u;s,t}
\sum_{x\in\Lao}\eikx
\ad_{x+u,\dn}\ad_{x+r,\up}b_{x+t,\up}b_{x+s,\up}\UP
\ret
&=&
\sumfour{u\in\calU}{r\in\La'}{s,t\in\Lao}{s-t\in\Laop}
\rbk{\eik{s}-\eik{t}}\Ut_{r,u;s,t}\,
\Phik{u,r,t,s}.
\end{eqnarray}
This leads us to
\begin{equation}
\hint[\Phik{u,r,t,s},\Ok]=\frac{1}{\alpha(k)}
\rbk{\eik{s}-\eik{t}}\Ut_{r,u;s,t},
\end{equation}
which gives the desired expression (\ref{h5}) since there 
are no corresponding
contributions from $\Hhopt$ or $\Hinto$.

To prove the only remaining expression (\ref{h3}), we calculate
\begin{eqnarray}
\Hinto\Phik{u,r,t,s}&=&
\sumtwo{v\in\La}{x\in\Lao}
\Ut_{x+t,v;x+u,x+r}\,\eikx\ad_{v,\dn}b_{x+s,\up}\UP
\ret
&&
-\sumtwo{v\in\La}{x\in\Lao}
\Ut_{x+s,v;x+u,x+r}\,\eikx\ad_{v,\dn}b_{x+t,\up}\UP
\ret
&&+(\mbox{other terms}),
\end{eqnarray}
where (other terms) do not contain any contributions to $\Phik{o,o}$.
Since we are interested in calculating the matrix elements
$\hint[\Phik{o,o},\Phik{u,r,t,s}]$,
we shall pick up only those terms which have some 
contributions to $\Phik{o,o}$.
This allows us to sum only over $v\in\Lao$ instead of $v\in\La$.
We can also consider only $x$ such that $x+s=v$ in the first term, and
$x+t=v$ in the second term.
Then we get
\begin{eqnarray}
\Hinto\Phik{u,r,t,s}&=&
\sum_{v\in\Lao}
\Ut_{v+t-s,v;v+u-s,v+r-s}\,\eik{(v-s)}\ad_{v,\dn}b_{v,\up}\UP
\ret
&&
-\sum_{v\in\Lao}
\Ut_{v+s-t,v;v+u-t,v+r-t}\,\eik{(v-t)}\ad_{v,\dn}b_{v,\up}\UP
\ret
&&+(\mbox{other terms})
\ret
&=&
\Ut_{s,t;u,r}\rbk{\emik{s}-\emik{t}}\Phik{o,o}
\ret
&&+(\mbox{other terms}),
\end{eqnarray}
which implies
\begin{equation}
\hint[\Ok,\Phik{u,r,t,s}]=\alpha(k)\rbk{\emik{s}-\emik{t}}\Ut_{s,t;u,r}.
\end{equation}
Since there are no corresponding contributions from $\Hhopt$ or $\Hinto$,
this gives the desired expression (\ref{h3}).
This completes the proof of Lemma~\ref{matrepLemma}.

\Section{Bounds on the Matrix Elements}
\label{SecMatrixBound}
Here we prove Lemmas~\ref{meb1Lemma} and \ref{meb2Lemma}
which state various bounds for the matrix elements and their sums.
In the proof we shall make use of the properties of the 
localized bases summarized in
Lemmas~\ref{basisLemma}, \ref{tauLemma}, and \ref{dualbasisLemma},
which will be proved in Section~\ref{SecBasis}.
In order to make use of these Lemmas, we have to assume
that $\la\ge\la_0$ and $\akappa\la^{-2}\le r_0$.
The bound for $\la$ is assumed in the statement of Lemma~\ref{meb1Lemma}.
In Lemma~\ref{meb2Lemma}, we assumed the stronger condition
$\la\ge\la_4$.
(We will choose $\la_4$ so that $\la_4\ge\la_0$.)
The bound $\akappa\la^{-2}\le r_0$ follows from the assumption $\akappa\le\kappa_0$
in Lemmas~\ref{meb1Lemma} and \ref{meb2Lemma},
since we shall now set $\kappa_0=(\la_0)^{2}r_0$.

\subsection{Bound for $h[\Ok,\Ok]$}
We first prove the lower bound (\ref{hOO>}) for 
${\rm Re}\sbk{h[\Ok,\Ok]}$.
In fact we prove the stronger estimate
\begin{equation}
\abs{h[\Ok,\Ok]-\rbk{E_0+\frac{U}{\la^4}G(k)}}
\le\frac{U}{\la^4}\rbk{C_1\akappa+\frac{C_2}{\la}}G(k),
\label{hOObound}
\end{equation}
which implies the desired  (\ref{hOO>}).

With the goal (\ref{hOObound}) in mind,
we will bound the quantity
\begin{eqnarray}
\tilde{G}(k)&=&
\frac{\la^4}{U}\rbk{h[\Ok,\Ok]-E_0}
\ret
&=&
2\frac{\la^4}{U}\sum_{s\in\Lao}\rbk{\sin\frac{k\cdot s}{2}}^2
\Ut_{s,o;s,o}
\ret
&=&
2\la^4\sumtwo{s\in\Lao}{x\in\La}
\rbk{\sin\frac{k\cdot s}{2}}^2
\phit{o}{x}\rbks{\phis{o}{x}}\phit{s}{x}\rbks{\phis{s}{x}},
\label{hOO1}
\end{eqnarray}
where we used the expression (\ref{h1}) for the matrix element,  
and the representation (\ref{Util}) for the 
effective interaction.
Let us introduce
\begin{equation}
\etas{x}{y}=\rbks{\phis{x}{y}}-\psis{x}{y},\quad
\etat{x}{y}=\phit{x}{y}-\psis{x}{y},
\end{equation}
where $\psis{x}{y}$ is the localized basis state (\ref{psi1}), (\ref{psi2})
of the flat-band model.
Then (\ref{hOO1}) can be written as
\begin{eqnarray}
\tilde{G}(k)&=&
2\la^4\sumtwo{s\in\Lao}{x\in\La}
\rbk{\sin\frac{k\cdot s}{2}}^2
(\psis{o}{x}+\etat{o}{x})(\psis{o}{x}+\etas{o}{x})
(\psis{s}{x}+\etat{s}{x})(\psis{s}{x}+\etas{s}{x})
\ret
&=&
G_0(k)+G_1(k)+G_2(k)+G_3(k)+G_4(k),
\label{hOO2}
\end{eqnarray}
where $G_i(k)$ denotes the collection of terms which contain the $i$-th power
of $\psi$'s when we expand the left-hand side.
In the following, we shall control $G_i$ for each $i=0,1,2,3$, and $4$.

We first control $G_0(k)$.
It gives the most dominant contribution as
\begin{eqnarray}
G_0(k)&=&
2\la^4\sumtwo{s\in\Lao}{x\in\La}
(\psis{o}{x})^2(\psis{s}{x})^2
\rbk{\sin\frac{k\cdot s}{2}}^2
\ret
&=&
2\sum_{f\in\calF_o}\sum_{g\in\calF_f}
\rbk{\sin\frac{k\cdot(f+g)}{2}}^2
=G(k),
\label{G0}
\end{eqnarray}
where we used the expression (\ref{psi1}) of $\psis{y}{x}$.
See (\ref{Fo}) and (\ref{Ff}) for the definitions of $\calF_o$ and $\calF_f$.

We bound the absolute value of $G_1(k)$.
One of the four terms in $G_1(k)$ is bounded as
\begin{eqnarray}
\abs{
2\la^4\sumtwo{s\in\Lao}{x\in\La}
\rbk{\sin\frac{k\cdot s}{2}}^2
\etat{o}{x}\psis{o}{x}(\psis{s}{x})^2
}
&\le&
2\la^4\frac{B_1\akappa+B_2}{\la^2}\frac{1}{\la^3}
\sum_{f\in\calF_o}\sum_{g\in\calF_f}
\rbk{\sin\frac{k\cdot(f+g)}{2}}^2
\ret
&=&
\frac{B_1\akappa+B_2}{\la}G(k),
\end{eqnarray}
where we used (\ref{reg5}) to get the bound 
$\abs{\etat{o}{x}}\le(B_1\akappa+B_2)/\la^2$.
The other three terms in $G_1(k)$ can be bounded similarly, and we get
\begin{equation}
\abs{G_1(k)}\le\frac{4B_1\akappa+2B_2}{\la}G(k).
\label{G1}
\end{equation}

We bound the absolute value of $G_2(k)$.
One of the six terms in $G_2(k)$ is bounded as
\begin{eqnarray}
&&
\abs{
2\la^4\sumtwo{s\in\Lao}{x\in\La}
\rbk{\sin\frac{k\cdot s}{2}}^2
(\psis{o}{x})^2\etas{s}{x}\etat{s}{x}
}
\ret&&
\le\la^4\sumtwo{s\in\Lao}{x\in\La}(\psis{o}{x})^2
\frac{\abs{k}^2}{2}\abs{s}\abs{\etas{s}{x}}\abs{s}\abs{\etat{s}{x}}
\ret&&\le
\la^4\frac{\abs{k}^2}{2}\rbk{
\sum_{s\in\Lao}\abs{s}\abs{\etas{s}{o}}\abs{s}\abs{\etat{s}{o}}
+\frac{1}{\la^2}\sum_{f\in\calF_o}\sum_{s\in\Lao}
\abs{s}\abs{\etas{s}{f}}\abs{s}\abs{\etat{s}{f}}
}
\ret&&\le
\la^4\frac{\abs{k}^2}{2}
\rbk{\sum_{s\in\Lao}\abs{s}\abs{\etas{s}{o}}}
\rbk{\sum_{s\in\Lao}\abs{s}\abs{\etat{s}{o}}}
\ret&&\quad
+
2\abs{k}^2\la^2\sum_{f\in\calF_o}
\rbk{\sum_{s\in\Lao}\abs{s-f}\abs{\etas{s}{f}}}
\rbk{\sum_{s\in\Lao}\abs{s-f}\abs{\etat{s}{f}}}
\ret&&\le
\abs{k}^2\la^4\rbk{\frac{1}{2}+\frac{2\abs{\calF_o}}{\la^2}}
\frac{B_1R\akappa}{\la^2}\frac{B_1R\akappa+B_2}{\la^2},
\label{G21}
\end{eqnarray}
where we used $\abs{\sin(k\cdot s/2)}\le\abs{k}\abs{s}/2$,
$\abs{s}\le2\abs{s-f}$, and the bounds (\ref{reg4}) and (\ref{reg8})
on the summability of the basis states.
Another term 
$\abs{2\la^4\sum_{s\in\Lao,x\in\La}
\rbk{\sin\frac{k\cdot s}{2}}^2
(\psis{s}{x})^2\etas{o}{x}\etat{o}{x}}$
can be bounded by the same quantity as in (\ref{G21}).

The remaining four terms in $G_2(k)$ have the 
common structure
\begin{equation}
\abs{
2\la^4\sumtwo{s\in\Lao}{x\in\La}
\rbk{\sin\frac{k\cdot s}{2}}^2
\psis{o}{x}\psis{s}{x}\etab{o}{x}\etab{s}{x}
}
=
2\la^2
\sumtwo{f\in\calF_o}{g\in\calF_f}
\abs{\etab{o}{f}}\abs{\etab{f+g}{f}}\rbk{\sin\frac{k\cdot(f+g)}{2}}^2,
\label{G22}
\end{equation}
where $\bar{\eta}$ denotes either
$\eta$ or $\tilde{\eta}$.
(The four terms are obtained by assigning $\eta$ or $\tilde{\eta}$
with each $\bar{\eta}$.)
We can bound $\abs{\etab{o}{f}}$ and $\abs{\etab{f+g}{f}}$
using (\ref{reg1}) or (\ref{reg5}) depending on whether 
$\bar{\eta}=\eta$ or $\tilde{\eta}$.
By summing the resulting bounds and (\ref{G21}), we get
\begin{eqnarray}
\abs{G_2(k)}
&\le&
\cbk{
\rbk{1+\frac{2\abs{\calF_o}}{\la^2}}B_1R\akappa
\rbk{B_1R\akappa+B_2}
}\abs{k}^2
+\la^2\rbk{\frac{4B_1\akappa+2B_2}{\la^2}}^2G(k)
\ret
&\le&
\rbk{C_1\akappa+\frac{B_3}{\la^2}}G(k),
\label{G23}
\end{eqnarray}
with constants $C_1$ and $B_3$ depending only on $d$, $\nu$, and $R$.
Here we used the assumed bounds $\la\ge\la_0$ and $\akappa\le\kappa_0$,
as well as the bound (\ref{k2<Gk}) to bound $\abs{k}^2$ by $G(k)$.

The quantities 
$G_3(k)$ and $G_4(k)$ which contain higher powers of $\eta$ or $\tilde{\eta}$
can be bounded in a similar (in fact easier) manner, and we get
\begin{equation}
\abs{G_3(k)}\le\frac{B_4\akappa}{\la^2}G(k),
\label{G3}
\end{equation}
and
\begin{equation}
\abs{G_4(k)}\le\frac{B_5}{\la^4}G(k)
\label{G4}
\end{equation}
with constants $B_4$ and $B_5$ which depend only on $d$, $\nu$, and $R$.

By summing up (\ref{G0}), (\ref{G1}), (\ref{G23}), (\ref{G3}), and (\ref{G4}),
and comparing the result with (\ref{hOO1}) and (\ref{hOO2}),
we finally get
\begin{eqnarray}
\abs{\tilde{G}(k)-G(k)}
&\le&
\rbk{
1-\frac{4B_1\akappa+2B_2}{\la}-C_1\akappa-\frac{B_3}{\la^2}
-\frac{B_4\akappa}{\la^2}-\frac{B_5}{\la^4}
}
G(k)
\ret
&\le&
\rbk{1-C_1\akappa-\frac{C_2}{\la}}G(k)
\end{eqnarray}
with a constant $C_2$ which depend only on $d$, $\nu$, and $R$.
This is nothing but the desired (\ref{hOObound}).

\subsection{Bound for $\sum\abs{h[\Ok,\Psi]}$}
We shall prove the bound (\ref{hOPsi<}) for the sum of the 
off-diagonal matrix elements $h[\Ok,\Psi]$
stated in Lemma~\ref{meb1Lemma}.
We first note that, since $\Psi$ with $h[\Ok,\Psi]\ne0$ is either
of the form $\Phik{u,r}$ of (\ref{Phiur}) or $\Phik{u,r,t,s}$ of
(\ref{Phiurts}), we can write the desired quantity as
\begin{equation}
\sum_{\Psi\in\Bk\bs\cbk{\Ok}}\abs{h[\Ok,\Psi]}
=
\sumthree{u\in\calU}{r\in\Lao}{(u,r)\ne(o,o)}
\abs{h[\Ok,\Phik{u,r}]}
+\sumfour{u\in\calU}{r\in\La'}{s,t\in\Lao}{(s-t\in\Laop)}
\abs{h[\Ok,\Phik{u,r,t,s}}.
\label{hOPsi1}
\end{equation}

To bound the first term in the right-hand side of (\ref{hOPsi1}),
we use the expression (\ref{h2}) for the matrix element to get 
\begin{eqnarray}
&&\sumthree{u\in\calU}{r\in\Lao}{(u,r)\ne(o,o)}
\abs{h[\Ok,\Phik{u,r}]}
\ret&&\le
\sum_{r\in\Lao}\alpha(k)\abs{\rbk{\emik{r}-1}\tau_{r,o}}
+\sumthree{u\in\calU}{r,s\in\Lao}{(u,r)\ne(o,o)}
\alpha(k)\abs{\rbk{\emik{r}-\emik{s}}\Ut_{s,r;u,s}}
\ret&&\le
\alpha(k)\abs{k}\sum_{r\in\Lao}\abs{r}\abs{\tau_{r,o}}
+\alpha(k)\abs{k}\sumthree{u\in\calU}{r,s\in\Lao}{(u,r)\ne(o,o)}
\abs{r-s}\abs{\Ut_{s,r;u,s}}.
\label{hOPsi2}
\end{eqnarray}
The first term in the right-hand side is readily bounded by
$\alpha(k)B_1Rt\akappa\abs{k}$ from the summability (\ref{regtau2})
of $\tau_{x,y}$.
To bound the second term, we use the representation (\ref{Util})
for $\Ut_{s,r;u,s}$ and the bound $\abs{r-s}\le\abs{r-x}+\abs{s-x}$
to get
\begin{eqnarray}
&&
\sumthree{u\in\calU}{r,s\in\Lao}{(u,r)\ne(o,o)}
\abs{r-s}\abs{\Ut_{s,r;u,s}}
\ret&&\le
U\sumthree{r,s\in\Lao}{u\in\calU}{x\in\La}
\rbk{\abs{r-x}+\abs{s-x}}
\abs{\phit{s}{x}\phit{r}{x}}\abs{\phis{u}{x}\phis{s}{x}}
\ret&&\le
U\sumtwo{x\in\La}{u\in\calU}
\cbk{
\abs{\phis{u}{x}}
\rbk{\sum_{s\in\Lao}\abs{\phit{s}{x}}}
\rbk{\sum_{r\in\Lao}\abs{r-x}\abs{\phit{r}{x}}}
\rbk{\sum_{s\in\Lao}\abs{\phis{s}{x}}}
}
\ret&&
+U\sumtwo{x\in\La}{u\in\calU}
\cbk{
\abs{\phis{u}{x}}
\rbk{\sum_{s\in\Lao}\abs{s-x}\abs{\phit{s}{x}}}
\rbk{\sum_{r\in\Lao}\abs{\phit{r}{x}}}
\rbk{\sum_{s\in\Lao}\abs{\phis{s}{x}}}
}
\ret&&
=2U\sumtwo{x\in\La}{u\in\calU}
\cbk{
\abs{\phis{u}{x}}
\rbk{\sum_{s\in\Lao}\abs{s-x}\abs{\phit{s}{x}}}
\rbk{\sum_{s\in\Lao}\abs{\phit{s}{x}}}
\rbk{\sum_{s\in\Lao}\abs{\phis{s}{x}}}
}
\ret&&
=2U\sumtwo{x\in\La'}{u\in\calU}\cbk{\cdots}
+2U\sumtwo{x\in\Lao}{u\in\calU}\cbk{\cdots}
\ret&&
\le
2U\sumtwo{x\in\La'}{u\in\calU}
\abs{\phis{u}{x}}
\rbk{2^\nu\frac{\sqrt{\nu}}{2\la}+\frac{B_1R\akappa+B_2}{\la^2}}
\rbk{\frac{2^\nu}{\la}+\frac{B_1\akappa+B_2}{\la^2}}
\rbk{\frac{2^\nu}{\la}+\frac{B_1\akappa}{\la^2}}
\ret&&
+2U\sumtwo{x\in\Lao}{u\in\calU}
\abs{\phis{u}{x}}
\rbk{\frac{B_1R\akappa+B_2}{\la^2}}
\rbk{1+\frac{B_1\akappa+B_2}{\la^2}}
\rbk{1+\frac{B_1\akappa}{\la^2}}
\ret&&
\le B_6\frac{U}{\la^2},
\label{hOPsi3}
\end{eqnarray}
where the constant $B_6$ depends only on $d$, $\nu$, and $R$.
We have used the expressions (\ref{psi1}), (\ref{psi2}) for $\psis{u}{x}$,
and the bounds (\ref{reg3}), (\ref{reg7}), (\ref{reg8}), and (\ref{reg1})
for the sum of the basis states.

Next we bound the second term in the right-hand side of (\ref{hOPsi1}).
we again use the expression (\ref{h3}) and the representation (\ref{Util})
to get
\begin{eqnarray}
\sumfour{u\in\calU}{r\in\La'}{s,t\in\Lao}{(s-t\in\Laop)}
\abs{h[\Ok,\Phik{u,r,t,s}]}
&\le&
\alpha(k)\sumthree{u\in\calU}{r\in\La'}{s,t\in\Lao}
\abs{\rbk{\emik{s}-\emik{t}}\Ut_{s,t;u,r}}
\ret&\le&
\alpha(k)\abs{k}U
\sumfour{x\in\La}{u\in\calU}{r\in\La'}{s,t\in\Lao}
\abs{s-t}
\abs{\phit{s}{x}\phit{t}{x}}\abs{\phis{u}{x}\phis{r}{x}}
\label{hOPsi3d}
\end{eqnarray}
To bound the sum, we again use 
$\abs{s-t}\le\abs{s-x}+\abs{t-x}$,
and the symmetry between $s$ and $t$ as we did in 
(\ref{hOPsi3}) to get
\begin{eqnarray}
&&
\sumfour{x\in\La}{u\in\calU}{r\in\La'}{s,t\in\Lao}
\abs{s-t}
\abs{\phit{s}{x}\phit{t}{x}}\abs{\phis{u}{x}\phis{r}{x}}
\ret&&\le
2\sumtwo{x\in\La}{u\in\calU}
\cbk{
\abs{\phis{u}{x}}
\rbk{\sum_{s\in\Lao}\abs{s-x}\abs{\phit{s}{x}}}
\rbk{\sum_{t\in\Lao}\abs{\phit{t}{x}}}
\rbk{\sum_{r\in\La'}\abs{\phis{r}{x}}}
}
\ret&&
=2\sumtwo{x\in\La'}{u\in\calU}\cbk{\cdots}
+2\sumtwo{x\in\Lao}{u\in\calU}\cbk{\cdots}
\ret&&
\le
2\sumtwo{x\in\La'}{u\in\calU}
\abs{\phis{u}{x}}
\rbk{2^\nu\frac{\sqrt{\nu}}{2\la}+\frac{B_1R\akappa+B_2}{\la^2}}
\rbk{\frac{2^\nu}{\la}+\frac{B_1\akappa+B_2}{\la^2}}
\rbk{1+\frac{B_1\akappa}{\la^2}}
\ret&&
\quad
+2\sumtwo{x\in\Lao}{u\in\calU}
\abs{\phis{u}{x}}
\rbk{\frac{B_1R\akappa+B_2}{\la^2}}
\rbk{1+\frac{B_1R\akappa+B_2}{\la^2}}
\rbk{\frac{\abs{\calF_o}}{\la}+\frac{B_1\akappa}{\la^2}}
\ret&&
\le B_7\frac{1}{\la^2}.
\label{hOPsi4}
\end{eqnarray}

By combining (\ref{hOPsi1})-(\ref{hOPsi4}), 
we finally get the desired bound (\ref{hOPsi<})
with $C_3=B_6+B_7$.
\subsection{Bounds for the Other Matrix Elements}
Here we prove the bounds (\ref{hPur}) and (\ref{hPurts})
stated in Lemma~\ref{meb1Lemma}.

Instead of proving (\ref{hPur}) for fixed $u\in\calU$
and $r\in\La_o$ with $(u,r)\ne(o,o)$,
we prove the bound for their sum
\begin{equation}
\sumthree{u\in\calU}{r\in\Lao}{(u,r)\ne(o,o)}
\abs{h[\Phik{u,r},\Ok]}
\le
\frac{1}{\alpha(k)}
\rbk{B_1R\akappa t +\frac{C_4U\akappa}{\la^2}+\frac{C_5U}{\la^3}}
\abs{k},
\label{hPurbound}
\end{equation}
which clearly implies the desired (\ref{hPur}).
By using the expression (\ref{h4}) for the matrix element, we have
\begin{eqnarray}
&&
\sumthree{u\in\calU}{r\in\Lao}{(u,r)\ne(o,o)}
\abs{h[\Phik{u,r},\Ok]}
\ret&&
\le
\frac{1}{\alpha(k)}
\sum_{r\in\Lao}\abs{\rbk{\eik{r}-1}\tau_{o,r}}
+
\frac{1}{\alpha(k)}
\sumtwo{u\in\calU}{r,s\in\Lao}
\abs{\rbk{\eik{r}-\eik{s}}\Ut_{u,s;s,r}}
\ret&&
\le
\frac{1}{\alpha(k)}
B_1R\akappa t\abs{k}
+
\frac{\abs{k}}{\alpha(k)}
\sumtwo{u\in\calU}{r,s\in\Lao}
\abs{r-s}\abs{\Ut_{u,s;s,r}},
\label{hPO1}
\end{eqnarray}
where we used the summability (\ref{regtau2}) of $\abs{\tau_{o,r}}$.
The second sum can be treated in exactly the same manner as we did
for the similar sum in (\ref{hOPsi3}).
As a result, we get
\begin{eqnarray}
&&
\sumthree{u\in\calU}{r,s\in\Lao}{x\in\La}
\abs{r-s}\abs{\Ut_{u,s;s,r}}
\ret&&
\le
U\sumthree{u\in\calU}{r,s\in\Lao}{x\in\La}
\abs{r-s}
\abs{\phit{u}{x}\phit{s}{x}}\abs{\phis{s}{x}\phis{r}{x}}
\ret&&\le
2U\sumtwo{x\in\La}{u\in\calU}
\cbk{
\abs{\phit{u}{x}}
\rbk{\sum_{s\in\Lao}\abs{s-x}\abs{\phis{s}{x}}}
\rbk{\sum_{s\in\Lao}\abs{\phis{s}{x}}}
\rbk{\sum_{s\in\Lao}\abs{\phit{s}{x}}}
}.
\label{hPO2}
\end{eqnarray}
This is the same as the fifth line in (\ref{hOPsi3}),
except that $\phi$ and $\phitil$ are switched.
Because of the drastic
difference in the localization properties of
the states $\phi$ and $\phitil$, this results in the remarkable
difference between 
$\sum\abs{h[\Ok,\Phik{u,r}]}$ and $\sum\abs{h[\Phik{u,r},\Ok]}$.
Again by decomposing the sum over $x$ as
$\sum_{x\in\La}\cbk{\cdots}=\sum_{x\in\La'}\cbk{\cdots}
+\sum_{x\in\Lao}\cbk{\cdots}$,
and using the expression (\ref{psi1}) for $\psi^{(u)}$ and the bounds
(\ref{reg3}), (\ref{reg4}), and (\ref{reg7})
for the sum of the basis states,
we can further bound (\ref{hPO2}) as
\begin{eqnarray}
&&
\sumthree{u\in\calU}{r,s\in\Lao}{x\in\La}
\abs{r-s}\abs{\Ut_{u,s;s,r}}
\ret&&
\le
2U\sumtwo{x\in\La'}{u\in\calU}
\abs{\phis{u}{x}}
\rbk{2^\nu\frac{\sqrt{\nu}}{2\la}+\frac{B_1\akappa}{\la^2}}
\rbk{\frac{2^\nu}{\la}+\frac{B_1\akappa}{\la^2}}
\rbk{\frac{2^\nu}{\la}+\frac{B_1\akappa+B_2}{\la^2}}
\ret&&
\quad
+2U\sumtwo{x\in\Lao}{u\in\calU}
\abs{\phis{u}{x}}
\rbk{\frac{B_1R\akappa}{\la^2}}
\rbk{1+\frac{B_1\akappa}{\la^2}}
\rbk{1+\frac{B_1\akappa+B_2}{\la^2}}
\ret&&
\le 
\frac{C_4U\akappa}{\la^2}+\frac{C_5U}{\la^3},
\label{hPO3}
\end{eqnarray}
where $C_4$ and $C_5$ are constants.
The desired (\ref{hPurbound}) follows from 
(\ref{hPO1}) and (\ref{hPO3}).

Next we show the bound (\ref{hPurts}) for $h[\Phik{u,r,t,s},\Ok]$.
It is done in the similar manner as we bounded 
$h[\Ok,\Phik{u,r,t,s}]$ in (\ref{hOPsi3d}) and (\ref{hOPsi4}).
From (\ref{h3}) and (\ref{Util}), we have
\begin{eqnarray}
&&
\abs{h[\Phik{u,r,t,s},\Ok]}
\ret&&
\le
\frac{1}{\alpha(k)}\abs{\eik{s}-\eik{t}}\abs{\Ut_{s,t;u,r}}
\ret&&
\le
\frac{\abs{k}U}{\alpha(k)}
\sum_{x\in\La}
\abs{s-t}
\abs{\phit{s}{x}\phit{t}{x}}\abs{\phis{u}{x}\phis{r}{x}}
\ret&&
\le
\frac{2\abs{k}U}{\alpha(k)}
\sum_{x\in\La}
\abs{\phis{u}{x}}\abs{\phis{r}{x}}
\rbk{\sum_{s\in\Lao}\abs{s-x}\abs{\phit{s}{x}}}
\rbk{\sum_{t\in\Lao}\abs{\phit{t}{x}}}
\ret&&
\le
\frac{2\abs{k}U}{\alpha(k)}
\sum_{x\in\La'}
\abs{\phis{u}{x}}\abs{\phis{r}{x}}
\rbk{2^\nu\frac{\sqrt{\nu}}{2\la}+\frac{B_1\akappa+B_2}{\la^2}}
\rbk{\frac{2^\nu}{\la}+\frac{B_1\akappa+B_2}{\la^2}}
\ret&&
\quad+
\frac{2\abs{k}U}{\alpha(k)}
\sum_{x\in\Lao}
\abs{\phis{u}{x}}\abs{\phis{r}{x}}
\rbk{\frac{B_1\akappa+B_2}{\la^2}}
\rbk{1+\frac{B_1\akappa+B_2}{\la^2}}
\ret&&
\le
\frac{1}{\alpha(k)}\frac{C_6U}{\la^2}\abs{k},
\end{eqnarray}
which is the desired (\ref{hPurts}).

This completes the proof of Lemma~\ref{meb1Lemma}.
\subsection{Proof of Lemma~\protect\ref{meb2Lemma}}
\label{secDtil}
We shall prove Lemma~\ref{meb2Lemma} which controls the
sum $\Dt[\Psik{u,A}]$ of the matrix elements.
We recall that the assumptions for this lemma is different from those
for Lemma~\ref{meb1Lemma}.

By using the representation (\ref{HintRep}) of the interaction Hamiltonian
and the definition (\ref{PsiA}) of the basis state $\Psik{u,A}$,
we find
\begin{eqnarray}
&&
\Hint\Psik{u,A}
\ret&&
=
\sum_{x\in\Lao}\eikx\,
T_x\sbk{
\sum_{y,v,z\in\La}\Ut_{y,v;u,z}\,
\ad_{y,\up}\ad_{v,\dn}b_{u,\dn}b_{z,\up}\ad_{u,\dn}
\rbk{\prod_{t\in A}\ad_{t,\up}}\vac
}
\ret&&
=
\sum_{x\in\Lao}\eikx\,
T_x\sbk{
\sum_{y,v,z\in\La}\Ut_{y,v;u,z}\,
{\rm sgn}[y,z;A]\,
\ad_{v,\dn}\rbk{\prod_{t\in A_{z\rightarrow y}}\ad_{t,\up}}\vac
}
\ret&&
=
\sum_{y,v,z\in\La}{\rm sgn}[y,z;A]\,\Ut_{y,v;u,z}
\Psik{v,A_{z\rightarrow y}},
\label{Hintmat1}
\end{eqnarray}
where we have used the translation invariance of $\Ut_{y,v;u,z}$.
The set $A_{z\rightarrow y}$ is obtained by replacing the site
$z$ in $A$ with $y$, and ${\rm sgn}[y,z;A]=\pm1$ comes from the
reordering of the fermion operators.
The matrix element $\hint[\Psik{u',A'},\Psik{u,A}]$
can be (in principle) obtained
 from (\ref{Hintmat1}) if we take into account the 
identification (\ref{PsiID}) between the basis states
and rewrite $\Psik{v,A_{z\rightarrow y}}$
in terms of some $\Psik{u',A'}\in\Bk$.
But here we take a slightly different strategy.

By $\htint[\cdots,\cdots]$ let us denote the pseudo matrix elements
which are directly read off from (\ref{Hintmat1}) without taking
into account the identification (\ref{PsiID}).
We immediately find from (\ref{Hintmat1}) that 
\begin{equation}
\htint[\Psik{v,A_{z\rightarrow y}},\Psik{u,A}]
=
{\rm sgn}[y,z;A]\,\Ut_{y,v;u,z},
\label{htint1}
\end{equation}
and, by a suitable replacement of symbols, that
\begin{equation}
\htint[\Psik{u,A},\Psik{v,A_{z\rightarrow y}}]
=
{\rm sgn}[y,z;A]\,\Ut_{u,z;y,v}.
\label{htint2}
\end{equation}

Since some of the diagonal elements in the true matrix
elements $\hint[\cdots,\cdots]$ are treated as off-diagonal
elements in the pseudo matrix elements $\htint[\cdots,\cdots]$,
we observe that
\begin{equation}
\sum_{\Phi\in\Bk\bs\cbk{\Psik{u,A},\Ok}}
\abs{\hint[\Psik{u,A},\Phi]}
\le
\sum_{\Phi\in\Bk\bs\cbk{\Psik{u,A},\Ok}}
\abs{\htint[\Psik{u,A},\Phi]},
\label{hintoff}
\end{equation}
and
\begin{eqnarray}
&&{\rm Re}\sbk{\hint[\Psik{u,A},\Psik{u,A}]}
\ret
&&
\ge
{\rm Re}\sbk{\htint[\Psik{u,A},\Psik{u,A}]}
-\sum_{\Phi\in\Bk\bs\cbk{\Psik{u,A},\Ok}}
\abs{\htint[\Psik{u,A},\Phi]},
\label{hintdiag}
\end{eqnarray}
for any $u\in\calU$, $A\subset\La$ with $\abs{A}=L^d-1$
such that $(u,A)\ne(o,\Laob)$, where $\Laob=\Lao\bs\cbk{o}$.
From (\ref{hintoff}) and (\ref{hintdiag}), we can bound
the contribution to $\Dt[\Psik{u,A}]$ (\ref{Dtil})
from the interaction
Hamiltonian as
\begin{eqnarray}
&&\Dt_{\rm int}[\Psik{u,A}]
\ret
&&
=
{\rm Re}\sbk{\hint[\Psik{u,A},\Psik{u,A}]}
-
\sum_{\Phi\in\Bk\bs\cbk{\Psik{u,A},\Ok}}
\abs{\hint[\Psik{u,A},\Phi]}
\ret
&&
\ge
{\rm Re}\cbk{\htint[\Psik{u,A},\Psik{u,A}]}
-2\sum_{\Phi\in\Bk\bs\cbk{\Psik{u,A},\Ok}}
\abs{\htint[\Psik{u,A},\Phi]},
\label{Dtilint}
\end{eqnarray}
for any $(u,A)\ne(o,\Laob)$.

By using (\ref{htint2}), the sum in the right-hand side of
(\ref{Dtilint}) can be evaluated as
\begin{eqnarray}
\sum_{\Phi\in\Bk\bs\cbk{\Psik{u,A},\Ok}}
\abs{\htint[\Psik{u,A},\Phi]}
&\le&
\sumtwo{y,v,z\in\La}{(v,y)\ne(u,z)}\abs{\Ut_{u,z;y,v}}
\ret
&\le&
 U\sumtwo{x,y,v,z\in\La}{(v,y)\ne(u,z)}
\abs{\phit{u}{x}\phit{z}{x}}\abs{\phis{y}{x}\phis{v}{x}},
\label{hintoff2}
\end{eqnarray}
where we used the representation (\ref{Util}) for the
effective interaction $\Ut_{u,z;y,v}$.
We further use the bounds
(\ref{psi1}), (\ref{psi2}), (\ref{reg3}), and (\ref{reg7})
for the sum of the 
localized basis states to bound (\ref{hintoff2}) as
\begin{eqnarray}
&&
\sum_{\Phi\in\Bk\bs\cbk{\Psik{u,A},\Ok}}
\abs{\htint[\Psik{u,A},\Phi]}
\ret&&
\le
U\sum_{x\in\La}\cbk{
\abs{\phit{u}{x}}
\rbk{\sum_{z\in\La}\abs{\phit{z}{x}}}
\rbk{\sum_{y\in\La}\abs{\phis{y}{x}}}
\rbk{\sum_{v\in\La}\abs{\phis{v}{x}}}
}
-U\rbk{\phit{u}{u}}^2\rbk{\phis{u}{u}}^2
\ret&&
\le
U\rbk{\max_{u'\in\calU}\sum_{x\in\La}\abs{\phit{u'}{x}}}^2
\rbk{\max_{u'\in\calU}\sum_{x\in\La}\abs{\phis{u'}{x}}}^2
-U\rbk{\phit{u}{u}}^2\rbk{\phis{u}{u}}^2
\ret&&
\le
U\Biggl\{
\rbk{1+\frac{\abs{\calF_o}}{\la}+\frac{B_1\akappa+B_2}{\la^2}}^2
\rbk{1+\frac{\abs{\calF_o}}{\la}+\frac{B_1\akappa}{\la^2}}
\ret&&
\hspace{1.2cm}-
\rbk{1-\frac{B_1\akappa+B_2}{\la^2}}^2
\rbk{1-\frac{B_1\akappa}{\la^2}}^2
\Biggr\}
\ret&&
\le
B_8\frac{U}{\la},
\label{hintoff3}
\end{eqnarray}
where the constant $B_8$ depends only of $d$, $\nu$, and $R$.
We have made use of the bounds
$\la\ge\la_4$ and $\akappa\le\kappa_0$.

By using (\ref{htint1}) and the representation (\ref{Util}) for
$\Ut_{u,y;u,y}$,
the diagonal element of $\htint[\cdots,\cdots]$ is written as
\begin{equation}
\htint[\Psik{u,A},\Psik{u,A}]
=
\sum_{y\in A}\Ut_{u,y;u,y}
=
U\sumtwo{x\in\La}{y\in A}
\phit{u}{x}\phit{y}{x}\rbks{\phis{u}{x}\phis{y}{x}}.
\end{equation}
Again by using the properties (\ref{reg3}) and (\ref{reg7})
of the basis states, we have
\begin{eqnarray}
&&
{\rm Re}\sbk{\htint[\Psik{u,A},\Psik{u,A}]}
\ret&&
\ge
\chi[u\in A]U\,{\rm Re}\sbk{\rbk{\phit{u}{u}}^2\cbk{\rbks{\phis{u}{u}}}^2}
-
U\sumthree{x\in\La}{y\in A}{(x,y)\ne(u,u)}
\abs{\phit{u}{x}\phit{y}{x}}\abs{\phis{u}{x}\phis{y}{x}}
\ret&&
\ge
\chi[u\in A]U\,{\rm Re}\sbk{\rbk{\phit{u}{u}}^2\cbk{\rbks{\phis{u}{u}}}^2}
\ret&&\quad
-
U\cbk{
\rbk{\sum_{x\in\La}\abs{\phit{u}{x}}}
\rbk{\max_{x'\in\La}\sum_{y\in\La}\abs{\phit{y}{x'}}}
\rbk{\sum_{x\in\La}\abs{\phis{u}{x}}}
\rbk{\max_{x'\in\La}\sum_{y\in\La}\abs{\phis{y}{x'}}}
-\abs{\phit{u}{u}}^2\abs{\phis{u}{u}}^2
}
\ret&&
\ge
\chi[u\in A]U-B_9\frac{U}{\la},
\label{hintdiag2}
\end{eqnarray}
where $\chi[\cdots]$ is the indicator function with
$\chi[\mbox{true event}]=1$ and $\chi[\mbox{false event}]=0$.

Substituting (\ref{hintoff3}) and (\ref{hintdiag2}) into
(\ref{Dtilint}), we get
\begin{equation}
\Dt_{\rm int}[\Psik{u,A}]
\ge
\chi[u\in A]U-B_{10}\frac{U}{\la}.
\label{hintsum}
\end{equation}

Next we examine the matrix elements of the modified hopping 
Hamiltonian $\Hhopt$.
By using the representation (\ref{HtilhopRep})
and the definition (\ref{PsiA}), we get
\begin{eqnarray}
&&
\Hhopt\Psik{u,A}
\ret&&
=\frac{3}{4}\la^2t\abs{(A\cup\cbk{u})\cap\La'}\Psik{u,A}
\ret&&
\quad
+\sum_{x\in\Lao}\eikx\, T_x\sbk{
\sumtwo{y,z\in\Lao}{\sigma=\up,\dn}
\tau_{y,z}\,\ad_{y,\sigma}b_{z,\sigma}\ad_{u,\dn}
\rbk{\prod_{t\in A}\ad_{t,\up}}\vac
}
\ret&&
=\rbk{\ep_0\abs{(A\cup\cbk{u})\cap\Lao}+
\frac{3}{4}\la^2t\abs{(A\cup\cbk{u})\cap\La'}}
\Psik{u,A}
\ret&&
\quad
+\chi[u\in\Lao]\sum_{y\in\Lao\bs\cbk{u}}\tau_{y,u}
\Psik{y,A}
+\sumtwo{z\in A\cap\Lao}{y\in\Lao\bs A}
\tau_{y,z}\,{\rm sgn}[y,z;A]
\Psik{u,A_{z\rightarrow y}},
\label{HhopuA}
\end{eqnarray}
where we wrote $\ep_0=\tau_{y,y}$ for $y\in\Lao$.
Note that $\chi[u\in\Lao]=\delta_{u,o}$ as long as $u\in\calU$.

From (\ref{HhopuA}), we can read off the matrix elements
of $\Hhopt$ as
\begin{equation}
\hhop[\Psik{u,A},\Psik{u,A}]
=
\ep_0\abs{(A\cup\cbk{u})\cap\Lao}+
\frac{3}{4}\la^2t\abs{(A\cup\cbk{u})\cap\La'},
\label{hhopdiag}
\end{equation}
\begin{equation}
\hhop[\Psik{u,A},\Psik{u,A_{z\rightarrow y}}]
=
{\rm sgn}[y,z;A]\,\chi[z,y\in\Lao]\,\tau_{z,y},
\label{hhopoff1}
\end{equation}
and
\begin{equation}
\hhop[\Psik{o,A},\Psik{y,A}]=\tau_{o,y},
\label{hhopoff2}
\end{equation}
where $\Psik{y,A}$ in (\ref{hhopoff2}) should be properly
interpreted as a state in $\Bk$ using the identification (\ref{PsiID}).
We did not define pseudo matrix elements here since 
$\hhop[\Psik{o,A},\Psik{y,A}]$ does not contain any 
diagonal elements.

Let us use (\ref{hhopoff1}) and (\ref{hhopoff2})
to evaluate the sum of the off-diagonal matrix elements as
\begin{eqnarray}
\sum_{\Phi\in\Bk\bs\cbk{\Psik{u,A},\Ok}}
\abs{\hhop[\Psik{u,A},\Phi]}
&\le&
\sumtwo{z\in A\cap\Lao}{y\in\Lao\bs A}\abs{\tau_{z,y}}
+\delta_{u,o}\sum_{y\in\Lao\bs\cbk{o}}\abs{\tau_{o,y}}
\ret
&\le&
\rbk{\abs{\Lao\bs A}+\delta_{u,o}}B_1t\akappa
\ret
&=&
\rbk{\abs{A\cap\La'}+1+\delta_{u,o}}B_1t\akappa,
\label{hhopoff3}
\end{eqnarray}
where we used the bound (\ref{regtau1}) for the sum of the effective
hopping $\tau_{z,y}$.
The identity
$\abs{\Lao\bs A}=\abs{A\cap\La'}+1$
follows from $\abs{A}+1=\abs{\Lao}=L^d$.

By combining (\ref{hhopdiag}) and (\ref{hhopoff3}),
we can evaluate the contribution of $\Dt[\Psik{u,A}]$ from
the hopping Hamiltonian as
\begin{eqnarray}
\Dt_{\rm hop}[\Psik{u,A}]
&=&
\hhop[\Psik{u,A},\Psik{u,A}]
-\sum_{\Phi\in\Bk\bs\cbk{\Psik{u,A},\Ok}}
\abs{\hhop[\Psik{u,A},\Phi]}
\ret
&\ge&
\ep_0\rbk{L^d-1-\abs{A\cap\La'}+\delta_{u,o}}
+\frac{3}{4}\la^2t\rbk{\abs{A\cap\La'}+1-\delta_{u,o}}
\ret
&&
-\rbk{\abs{A\cap\La'}+1+\delta_{u,o}}B_1t\akappa.
\label{hhopsum}
\end{eqnarray}

By summing up the contributions (\ref{hintsum}) 
and (\ref{hhopsum}) from $\Hint$ and $\Hhopt$, respectively,
we can finally bound the desired quantity
$\Dt[\Psik{u,A}]$ (\ref{Dtil}) as
\begin{eqnarray}
&&
\Dt[\Psik{u,A}]
=
\Dt_{\rm int}[\Psik{u,A}]+\Dt_{\rm hop}[\Psik{u,A}]
\ret
&&\ge
E_0+\chi[u\in A]U-B_{10}\frac{U}{\la}
\ret
&&
+\rbk{\frac{3}{4}\la^2t-\ep_0-B_1t\akappa}\rbk{\abs{A\cap\La'}+1}
+\rbk{-\frac{3}{4}\la^2t+\ep_0-B_1t\akappa}\delta_{u,o},
\label{Dtil2}
\end{eqnarray}
where we noted that $E_0=\sum_{x\in\Lao}\tau_{x,x}=L^d\ep_0$.
See (\ref{E0}).

The desired bounds (\ref{Dtil>}), (\ref{Dtilur>}), and (\ref{Dtilor>})
are derived by investigating the bound (\ref{Dtil2}) in each situation.
We first consider the case $A\cap\La'\ne\emptyset$.
Noting that
$\abs{A\cap\La'}\ge1$, $\delta_{u,o}\le1$,
$\chi[u\in A]\ge0$, and $\ep_0\le B_1t\akappa$,
we find from the basic bound (\ref{Dtil2}) that
\begin{eqnarray}
\Dt[\Psik{u,A}]
&\ge&
E_0-B_{10}\frac{U}{\la}+\frac{3}{4}\la^2t-4B_1t\akappa
\ret
&\ge&
E_0+\frac{1}{2}\la^2t,
\end{eqnarray}
which is the desired bound (\ref{Dtil>}).
To get the final inequality, we have here assumed that
\begin{equation}
B_{10}\frac{U}{\la}\le\frac{1}{8}\la^2t,\quad
4B_1t\akappa\le\frac{1}{8}\la^2t.
\label{ass1}
\end{equation}

We then turn to the case  $A\cap\La'=\emptyset$.
Then the state $\Psik{u,A}$ is nothing but the state 
$\Phik{u,r}$ defined in (\ref{Phiur}).
When $u\ne o$, the basic bound (\ref{Dtil2}) with
$\abs{A\cap\La'}=0$ and $\delta_{u,o}=0$ yields
\begin{eqnarray}
\Dt[\Phik{u,r}]
&\ge&
E_0-B_{10}\frac{U}{\la}+\frac{3}{4}\la^2t-2B_1t\akappa
\ret
&\ge&
E_0+\frac{1}{2}\la^2t,
\end{eqnarray}
which is the desired bound (\ref{Dtilur>}).
We again used (\ref{ass1}).

Finally when $u=o$ and $r\ne o$, we find that
$\chi[u\in A]=1$ since the state $\phi^{(o)}$ is doubly occupied.
Thus the basic bound (\ref{Dtil2}) yields
\begin{eqnarray}
\Dt[\Phik{o,r}]
&\ge&
E_0+U-B_{10}\frac{U}{\la}-2B_1t\akappa
\ret
&\ge&
E_0+\frac{U}{2},
\end{eqnarray}
which is the desired (\ref{Dtilor>}).
To get the final inequality, we have assumed
\begin{equation}
2B_1t\akappa\le\frac{U}{4},\quad
B_{10}\frac{U}{\la}\le\frac{U}{2}.
\label{ass2}
\end{equation}

It only remains to examine the conditions for the parameters.
We shall set $K_3=8B_1$, $K_4=(8B_{10})^{-1}$, and
\begin{equation}
\la_4=\max\cbk{\la_0,2B_{10},\sqrt{32B_1\kappa_0}},
\label{la4}
\end{equation}
and make the requirements as in the statement of Lemma~\ref{meb2Lemma}.
Then the conditions (\ref{ass1}) and (\ref{ass2}) are easily checked to be 
satisfied.
Lemma~\ref{meb2Lemma} has been proved.

\Section{Upper Bound for the Spin-Wave Energy}
\label{SecSW<}
We will here prove Theorem~\ref{Ek<Th} which states
the upper bound (\ref{ESW<}) for the energy $\ESW$ for the
elementary spin-wave excitation with the wave number
vector $k\in\calK$.
In contrast to the corresponding lower bound, the upper bound
can be proved by employing the standard variational argument.
The new idea here is to use the state $\Ok$ (\ref{Omega})
as a trial state.
In the proof, we shall make use of Lemmas~\ref{basisLemma}
and \ref{dualbasisLemma} about the localized basis states,
and some estimates about the matrix elements proved in 
Section~\ref{SecMatrixBound} during the proof of
Lemma~\ref{meb1Lemma}.
The assumption made in the statement of Theorem~\ref{Ek<Th}
guarantees that we can make use of these results.
(See the beginning of Section~\ref{SecMatrixBound}.)

Since we have $\Ok\in\calH_k$, the lowest energy $\ESW$
in the space $\calH_k$ satisfies the variational inequality
\begin{equation}
\ESW\le\frac{\rbk{\Ok,H\,\Ok}}{\rbk{\Ok,\Ok}},
\label{Evar}
\end{equation}
where $(.,.)$ denotes the inner product.
Recalling the definitions of $\Htil$ (see (\ref{Hhoptil}) and (\ref{Htil})),
and the matrix elements (\ref{matrixelement}), we can write
\begin{eqnarray}
H\Ok
&=&
\Htil\Ok
\ret
&=&h[\Ok,\Ok]\Ok
+\sumthree{u\in\calU}{r\in\Lao}{(u,r)\ne(o,o)}h[\Phik{u,r},\Ok]\Phik{u,r}
\ret
&&
+\sumfour{u\in\calU}{r\in\La'}{s,t\in\Lao}{(s-t\in\Laop)}
h[\Phik{u,r,t,s},\Ok]\Phik{u,r,t,s}.
\end{eqnarray}
By noting that $(\Phik{u,r},\Ok)=0$ if $r\ne o$, and
$(\Phik{u,r,t,s},\Ok)=0$, we find
\begin{equation}
(\Ok,H\,\Ok)
=h[\Ok,\Ok](\Ok,\Ok)+
\sumtwo{r\in\Lao}{r\ne o}h[\Phik{o,r},\Ok](\Ok,\Phik{o,r}).
\label{O,HO}
\end{equation}

Recalling the definition (\ref{Phiur}) of $\Phik{o,r}$,
and noting that $\Ok=\alpha(k)^{-1}\Phik{o,o}$, we have
\begin{eqnarray}
(\Phik{o,o},\Phik{o,r})
&=&
\rbk{
\sum_{x\in\Lao}\eikx\ad_{x,\dn}b_{x,\up}\UP,
\sum_{y\in\Lao}\eik{y}\ad_{y,\dn}b_{y+r,\up}\UP
}
\ret
&=&
L^d\sum_{x\in\Lao}\emik{x}
\rbk{\UP,\bd_{x,\up}a_{x,\dn}\ad_{o,\dn}b_{r,\up}\UP},
\label{OP}
\end{eqnarray}
where we made use of the translation invariance to replace
$y$ by $o$.
Note that we have encountered the operators $a$ and $\bd$ for
the first time in the present paper.
Going back to the definitions (\ref{aDef}), (\ref{bDef}), we get
the anticommutation relations
\begin{equation}
\cbk{\ad_{x,\sigma},a_{y,\tau}}=\Gm{x,y}\,\delta_{\sigma,\tau},
\label{aa}
\end{equation}
and
\begin{equation}
\cbk{\bd_{x,\sigma},b_{y,\tau}}=(G^{-1})_{x,y}\,\delta_{\sigma,\tau},
\label{bb}
\end{equation}
where the Gramm matrix $G$ is given by
\begin{equation}
\Gm{x,y}=\sum_{z\in\La}\rbks{\phis{x}{z}}\phis{y}{z},
\label{Gxy}
\end{equation}
and its inverse is
\begin{equation}
\Gi{x,y}=\sum_{z\in\La}\rbks{\phit{x}{z}}\phit{y}{z}.
\label{Gixy}
\end{equation}
That (\ref{Gixy}) correctly defines inverse of $G$ can be easily verified
by using the duality relations (\ref{duality1}) and (\ref{duality2}).
The complicated anticommutation relations (\ref{aa}) and (\ref{bb})
are major drawback of the use of the non-orthogonal basis.

By using (\ref{aa}) and (\ref{bb}), we can further evaluate (\ref{OP}) as
\begin{equation}
(\Phik{o,o},\Phik{o,r})=
L^d\rbk{\sum_{x\in\Lao}\emik{x}\Gm{o,x}\Gi{x,r}}
(\UP,\UP).
\label{OP2}
\end{equation}
As for the expectation value in the right-hand side of (\ref{Evar}), we use
(\ref{O,HO}) and (\ref{OP2}) to get
\begin{eqnarray}
&&
\frac{\rbk{\Ok,H\,\Ok}}{\rbk{\Ok,\Ok}}
\ret&&
=h[\Ok,\Ok]
+\frac{
\sum_{r\in\Lao (r\ne o)}
\alpha(k)^{-1}h[\Phik{o,r},\Ok](\Phik{o,o},\Phik{o,r})
}{
\alpha(k)^{-2}(\Phik{o,o},\Phik{o,o})
}
\ret&&
=h[\Ok,\Ok]
\ret&&
\quad+\frac{
\sum_{r\in\Lao (r\ne o)}
\alpha(k)h[\Phik{o,r},\Ok]
\sum_{x\in\Lao}\emik{x}\Gm{o,x}\Gi{x,r}
}{
\sum_{x\in\Lao}\emik{x}\Gm{o,x}\Gi{x,o}
}.
\label{Evar2}
\end{eqnarray}
Since the first term in the right-hand side is already 
controlled by the bound
(\ref{hOObound}), we only need to bound the second term.

We start from the denominator of the final term in (\ref{Evar2}).
By noting that 
\newline
$\sum_{x\in\Lao}\Gm{o,x}\Gi{x,o}=1$, we have
\begin{eqnarray}
&&
\abs{1-\sum_{x\in\Lao}\emik{x}\Gm{o,x}\Gi{x,o}}
\ret&&
=
\abs{\sum_{x\in\Lao}(1-\emik{x})\Gm{o,x}\Gi{x,o}}
\ret&&
\le
\abs{k}\sum_{x\in\Lao}\abs{x}\abs{\Gm{o,x}\Gi{x,o}}
\ret&&
\le
\abs{k}\sumtwo{x\in\Lao}{y,z\in\La}\rbk{\abs{x-y}+\abs{y}}
\abs{\phis{o}{y}}\abs{\phis{x}{y}}\abs{\phit{x}{z}}\abs{\phit{o}{z}}
\ret&&
\le
\abs{k}2
\rbk{\frac{B_1R\akappa}{\la^2}}
\rbk{1+\frac{\abs{\calF_o}}{\la}+\frac{B_1\akappa}{\la^2}}
\rbk{1+\frac{\abs{\calF_o}}{\la}+\frac{B_1\akappa+B_2}{\la^2}}^2
\ret&&
\le
B_{11}\frac{\akappa}{\la^2},
\label{Evar3}
\end{eqnarray}
where we used (\ref{psi1}), (\ref{psi2}), (\ref{reg2}), (\ref{reg1}),
and (\ref{reg5}).
We also noted that $\la\ge\la_0$, $\akappa\le\kappa_0$,
and $\abs{k}\le\sqrt{d}\pi$.
Thus we get
\begin{equation}
\abs{\rbk{\sum_{x\in\Lao}\emik{x}\Gm{o,x}\Gi{x,o}}^{-1}}
\le
1+B_{12}\frac{\akappa}{\la^2}.
\label{Evar3b}
\end{equation}

We now control the numerator of the final term in (\ref{Evar2}).
By noting that
\newline
$\sum_{x\in\Lao}\Gm{o,x}\Gi{x,r}=0$ for $r\ne o$, we get
\begin{eqnarray}
&&
\abs{
\sum_{r\in\Lao\bs\cbk{o}}
\alpha(k)h[\Phik{o,r},\Ok]\sum_{x\in\Lao}\emik{x}\Gm{o,x}\Gi{x,r}
}
\ret&&
=\abs{
\sum_{r\in\Lao\bs\cbk{o}}
\alpha(k)h[\Phik{o,r},\Ok]\sum_{x\in\Lao}(\emik{x}-1)\Gm{o,x}\Gi{x,r}
}
\ret&&
\le
\alpha(k)
\rbk{\sum_{r\in\Lao\bs\cbk{o}}\abs{h[\Phik{o,r},\Ok]}}
\abs{k}
\rbk{\sum_{x,r\in\Lao}\abs{\Gm{o,x}\Gi{x,r}}}
\ret&&
\le
\rbk{B_1Rt\akappa+\frac{C_4U\akappa}{\la^2}+\frac{C_5U}{\la^3}}
\rbk{\sum_{x,r\in\Lao}\abs{x}\abs{\Gm{o,x}\Gi{x,r}}}
\abs{k}^2,
\label{Evar4}
\end{eqnarray}
where we used (\ref{hPurbound}) to control the sum of the matrix
elements $h[\Phik{o,r},\Ok]$.
The remaining factor can be bounded as
\begin{eqnarray}
\sum_{x,r\in\Lao}\abs{x}\abs{\Gm{o,x}\Gi{x,r}}
&\le&
\sumtwo{x,r\in\Lao}{y,z\in\La}\rbk{\abs{x-y}+\abs{y}}
\abs{\phis{o}{y}}\abs{\phis{x}{y}}\abs{\phit{x}{z}}\abs{\phit{r}{z}}
\ret
&\le&
B_{13}\frac{\akappa}{\la^2}.
\label{Evar5}
\end{eqnarray}

By collecting (\ref{Evar}), (\ref{Evar2}), (\ref{Evar3b}), (\ref{Evar4}),
and (\ref{Evar5}), and by using the bound (\ref{hOObound}) for
the matrix element $h[\Ok,\Ok]$,
and the bound (\ref{k2<Gk}) for $G(k)$, we finally get
\begin{eqnarray}
\ESW
&\le&
E_0+\frac{U}{\la^4}G(k)
+\frac{U}{\la^4}\rbk{C_1\akappa+\frac{C_2}{\la}}G(k)
\ret
&&
+B_{14}\rbk{B_1Rt\akappa+\frac{C_4U\akappa}{\la^2}+\frac{C_5U}{\la^3}}
\frac{\akappa}{\la^2}G(k)
\ret
&\le&
E_0
+\frac{U}{\la^4}
\rbk{1+\frac{A_4}{\la}+A_5\la\akappa
+\frac{A_6\la^2t\akappa^2}{U}}
G(k),
\end{eqnarray}
which is the desired (\ref{ESW<}).

\Section{Construction of the Localized Bases}
\label{SecBasis}
In the present section, we shall explicitly construct the localized
bases $\cbk{\phi^{(x)}}_{x\in\La}$,  $\cbk{\phitil^{(x)}}_{x\in\La}$,
and the dispersion relation $\ep_1(k)$,
and prove the summability stated in 
Lemmas~\ref{basisLemma}, \ref{dualbasisLemma}, and \ref{tauLemma}.

The main problem treated here is a
perturbation theory in the finite-dimensional eigenvalue problem
(\ref{SchVec1}), 
where the unperturbed problem has an energy gap.
It is well-established that such a finite-dimensional perturbation theory 
can be controlled in a perfectly rigorous manner \cite{Kato82,ReedSimon78}.

However there are some subtle points specific to the present problem.
Here we are treating the set of
eigenvalue problems indexed by the parameter $k\in\calK$.
Moreover it is essential for us to explicitly construct (unnormalized)
eigenvectors which are especially chosen to have ``nice'' $k$-dependence.
We found that, for this purpose, it is better to directly deal with the
Rayleigh-Schr\"{o}dinger perturbation theory in an  
explicit manner, rather than to make use of the general theory 
\cite{Kato82,ReedSimon78}.
Unfortunately such an analysis of perturbation theory requires us rather 
involved technical estimates which are summarized in this
lengthy section.
\subsection{States in the $k$-Space Representation}
\label{seckspace}
The basic starting point in the construction of the bases is the
Schr\"{o}dinger equation written in the form of (\ref{SchVec1}), which is
\begin{equation}
\ep\vv=\rbk{\la^2t\, \Mt + \kappa t \,\Qt}\vv,
\label{SchVec12}
\end{equation}
where $\vv=(\vk{u})_{u\in\calU}$ is a $b$-dimensional vector.
The $b\times b$ matrices $\Mt=(\Mk{u,u'})_{u,u'\in\calU}$ and 
$\Qt=(\Qk{u,u'})_{u,u'\in\calU}$ are defined in (\ref{Mkdef})
and (\ref{Qdef}), respectively.
For a fixed $k\in\calK$ (see (\ref{Kdef}) for the definition of the 
space $\calK$), (\ref{SchVec12}) is an eigenvalue equation of a $b\times b$
matrix.  
Here $b=\abs{\calU}=\dnu+1$ is the number of bands.
From a solution $\vv$ of (\ref{SchVec12}) for some $k$, we can construct the 
corresponding Bloch state in the real space by
\begin{equation}
\phi_x=\eikx\vk{\mu(x)},
\label{Blochstate}
\end{equation}
where $\mu(x)$ denotes the unique element in $\calU$ such that 
$x\in\La_{\mu(x)}$.
The Bloch state $\phi=(\phi_x)_{x\in\La}$ 
becomes an eigenstate of the original Schr\"{o}dinger equation
(\ref{Sch1}) with the energy eigenvalue $\epsilon$.

One of our major tasks in the following subsections is to construct,
for each $k\in\calK$,
a vector $\vov=(\vok{u})_{u\in\calU}$ which satisfies
\begin{equation}
\ep_1(k)\vov=\rbk{\la^2t\, \Mt + \kappa t \,\Qt}\vov,
\label{SchVec2}
\end{equation}
where $\ep_1(k)$ is the lowest eigenvalue for each $k$.
In other words, $\ep_1(k)$ is the dispersion relation of the lowest band.
Thus the Bloch state $\phi=(\phi_x)_{x\in\La}$ 
constructed from $\vov$ according to (\ref{Blochstate})
is an element of the Hilbert space $\Hsingo$ (see (\ref{Hsingdecomp})) 
for the lowest band.
In our construction, we do not normalize the vector $\vov$.
We rather try to get a $\vov$ which has a ``nice'' $k$-dependence so that we
finally get sharply localized basis states.

For the moment, we assume that the desired $\vov$ is defined,
and introduce other related vectors.
For each $e\in\calU'(=\calU\bs\cbk{o})$, we define a vector
$\vvv{e}=(\vks{e}{u})_{u\in\calU}$, so that
the Bloch state (\ref{Blochstate}) constructed from $\vvv{e}$ 
belongs to the Hilbert space $\Hsingp$ (see (\ref{Hsingdecomp}))
for the higher bands.
For this to be the case, it suffices to have 
orthogonality\footnote{
$(.,.)$ denotes the standard inner product in the $b$-dimensional
linear space.
For ${\bf v}=(v_u)_{u\in\calU}$ and ${\bf w}=(w_u)_{u\in\calU}$,
we define $({\bf v},{\bf w})=\sum_{u\in\calU}\rbks{v_u}w_u$.
}
$(\vvv{e},\vov)=0$ for each $k\in\calK$.
The vectors $\vvv{e}$ are defined in terms of $\vov$ as
\begin{equation}
\vks{e}{u}=
\cases{
-\vok{e}&if $u=o$;\cr
\rbks{\vok{o}}&if $u=e$;\cr
0&otherwise.\cr
}
\label{veuk}
\end{equation}
The required orthogonality is readily verified from the definition.
It is also found that, for each $k$, the vectors $\vvv{e}$ with $e\in\calU'$
are linearly independent with each other.
Therefore the collection $\cbk{\vvv{u}}_{u\in\calU}$ for a fixed $k$ forms
a  basis of ${\bf C}^b$.

We also introduce the dual of the basis $\cbk{\vvv{u}}_{u\in\calU}$.
For each $k\in\calK$, we define the Gramm matrix $\Gt$ by
\begin{equation}
\rbk{\Gt}_{u,u'}=(\vvv{u},\vvv{u'}),
\label{GrammDef}
\end{equation}
for $u,u'\in\calU$.
Since the vectors $\vvv{u}$ with $u\in\calU$ are linearly independent,
the corresponding Gramm matrix is invertible.
We define the dual vectors by
\begin{equation}
\vt{u}=\sum_{u'\in\calU}\rbk{\Gt^{-1}}_{u',u}\vvv{u'},
\label{dualvDef}
\end{equation}
for each $u\in\calU$.
We again write the components of the dual vectors as
$\vt{u}=(\vkt{u}{w})_{w\in\calU}$.
By definition, we have
\begin{equation}
\rbk{\vt{u},\vvv{u'}}=\delta_{u,u'},
\label{duality3}
\end{equation}
and
\begin{equation}
\sum_{w\in\calU}\rbks{\vkt{w}{u}}\vks{w}{u'}=\delta_{u,u'},
\label{duality4}
\end{equation}
for any $u,u'\in\calU$ and for any $k\in\calK$.
\subsection{Construction of the Localized Basis States}
\label{secrealspace}
Since we have introduced the vectors (states) in the $k$-space
representation, let us describe how we construct the 
desired localized basis states.
For $x\in\La$, we denote by $\mu(x)$ the unique site in the unit cell
$\calU$ such that $x\in\La_{\mu(x)}$.

For $x,y\in\La$, we define
\begin{equation}
\phis{y}{x}=(2\pi)^{-d}\int dk\,\eik{(x-y)}\,
\vks{\mu(y)}{\mu(x)},
\label{phiDef}
\end{equation}
and
\begin{equation}
\phit{y}{x}=(2\pi)^{-d}\int dk\,\eik{(x-y)}\,
\vkt{\mu(y)}{\mu(x)},
\label{phitilDef}
\end{equation}
where $\int dk(\cdots)$ is a shorthand for the sum
$(2\pi/L)^d\sum_{k\in\calK}(\cdots)$.

Let us prove the duality relation (\ref{duality1}).
By using the definitions (\ref{phiDef}) and (\ref{phitilDef}),
and (uniquely) decomposing $y\in\La$ as $y=z+u$ with $z\in\Lao$
and $u\in\calU$, we get
\begin{eqnarray}
&&
\sum_{y\in\La}\rbks{\phit{x}{y}}\phis{x'}{y}
\ret&&
=\sum_{y\in\La}(2\pi)^{-2d}\int dk\,dk'\,
e^{-ik\cdot(y-x)+ik'\cdot(y-x')}
\rbks{\vkt{\mu(x)}{\mu(y)}}v^{(\mu(x'))}_{\mu(y)}(k')
\ret&&
=
\sumtwo{u\in\calU}{z\in\Lao}
(2\pi)^{-2d}\int dk\,dk'\,
e^{-i(k-k')\cdot z-i(k-k')\cdot u+ik\cdot x-ik'\cdot x'}
\rbks{\vkt{\mu(x)}{u}}v^{(\mu(x'))}_{u}(k')
\ret&&
=(2\pi)^{-d}\int dk\,\eik{(x-x')}
\rbk{\vt{\mu(x)},\vvv{\mu(x')}}
\ret&&
=(2\pi)^{-d}\int dk\,\eik{(x-x')}\chi[x-x'\in\Lao]
\ret&&
=\delta_{x,x'},
\end{eqnarray}
where we used the duality relation (\ref{duality3})
for ${\bf v}$ and $\tilde{\bf v}$.
We have also noted that 
$\delta_{\mu(x),\mu(x')}=\chi[x-x'\in\Lao]$
with the indicator function $\chi[\mbox{true}]=1$,
$\chi[\mbox{false}]=0$.
The other duality relation (\ref{duality2}) follows from the general argument about
the uniqueness of inverse matrix, or can be
shown in the similar
manner by using the corresponding relation (\ref{duality4}).

In Lemma~\ref{basisLemma}, we claimed that the sets
$\cbk{\phi^{(x)}}_{x\in\Lao}$ and $\cbk{\phi^{(x)}}_{x\in\La'}$
form  bases of the Hilbert space $\Hsingo$ and $\Hsingp$, respectively.
Note that, in (\ref{phiDef}), $\phis{y}{x}$ is constructed as a superposition
of various Bloch states $\eikx\vks{\mu(y)}{\mu(x)}$ of the form 
(\ref{Blochstate}).
This means that $\phi^{(x)}\in\Hsingo$ if $x\in\Lao$ and 
$\phi^{(x)}\in\Hsingp$ if $x\in\La'$.
To prove the completeness of each basis, it therefore suffices to show that the union
$\cbk{\phi^{(x)}}_{x\in\La}$ is a  basis of the whole Hilbert space
$\Hsing$.
But the desired completeness follows readily from the duality relation (\ref{duality2}).
The same argument shows the corresponding claim about the completeness of the
dual bases stated in Lemma~\ref{dualbasisLemma}.

Finally we investigate the action of the modified hopping matrix
$\Ttil=(\ttil_{x,y})_{x,y\in\La}$ (see (\ref{Ttil})) on the basis states.
Noting the the Bloch state is given by (\ref{Blochstate}),
the Schr\"{o}dinger equation (\ref{Sch1}) and (\ref{Ttil}) imply 
\begin{equation}
\sum_{x'\in\La}\ttil_{x,x'}\rbk{\eik{x'}\vks{u}{\mu(x')}}
=\cases{
\ep_1(k)\rbk{\eik{x}\vks{o}{\mu(x)}}&if $u=o$;\cr
\frac{3}{4}\la^2t\rbk{\eik{x}\vks{u}{\mu(x)}}&if $u\in\calU'$.\cr
}
\label{Sch5}
\end{equation}
From (\ref{phiDef}) and (\ref{Sch5}), we get for $x\in\La$ that
\begin{eqnarray}
\sum_{z'\in\La}\ttil_{z,z'}\phis{x}{z'}
&=&
(2\pi)^{-d}\int dk\sum_{z'\in\La}\ttil_{z,z'}\,
\eik{(z'-x)}\,\vok{\mu(z')}
\ret
&=&
(2\pi)^{-d}\int dk\,\ep_1(k)\,\eik{(z-x)}\,\vok{\mu(z)}
\ret
&=&
\sum_{y\in\Lao}\tau_{y,x}\,\phis{y}{z},
\end{eqnarray}
which is nothing but (\ref{Ttilphi2}) with $\tau_{y,z}$ defined as
in (\ref{tauDef}).
The relation (\ref{Ttilphi1}) follows easily from (\ref{Sch5}).
\subsection{Basic Setup of Perturbation Theory}
\label{secsetup}
In the following construction of various vectors, we treat $k\in\calK$
as a fixed parameter.
The $k$-dependence of the vectors will play nontrivial roles only 
in the final Section~\ref{secreg}.

Let us first set $\kappa=0$ (corresponding to the flat-band model)
in the Schr\"{o}dinger equation (\ref{SchVec12}).
The eigenvector $\wv=(\wk{u})_{u\in\calU}$ with the lowest
eigenvalue $\ep=0$ is given by
\begin{equation}
\wk{u}=\cases{
1&if $u=o$;\cr
-{C_u(k)}/{\la}&if $u\in\calU'$,\cr
}
\label{wDef}
\end{equation} 
where $C_u(k)$ is defined in(\ref{Ckdef}).
We will construct our $\vov$ by the standard Rayleigh-Schr\"{o}dinger
perturbation theory so that it coincides with $\wv$ if $\kappa=0$.

For a fixed $k\in\calK$, we denote by 
$\Pt=(\Pk{u,u'})_{u,u'\in\calU}$ the orthogonal projection 
(in the linear space ${\bf C}^b$)
onto the vector $\wv$.
From (\ref{wDef}), we explicitly have
\begin{equation}
\Pk{u,u'}=\rbk{1+\frac{A(k)}{\la^2}}^{-1}\times
\cases{
1&if $u=u'=o$;\cr
-{C_u(k)}/{\la}&if $u\in\calU'$, $u'=o$;\cr
-{C_{u'}(k)}/{\la}&if $u=o$, $u'\in\calU'$;\cr
{C_u(k)C_{u'}(k)}/{\la^2}&if $u,u'\in\calU'$,\cr
}
\label{PDef}
\end{equation}
where $A(k)$ is defined in (\ref{Akdef}).

By comparing (\ref{SchVec12}) and (\ref{flatdispersion}), we find
that the matrix $\Mt$ (with a fixed $k$) has simple eigenvalues $0$,
$1+A(k)/\la^2$, and $(b-2)$-fold degenerate eigenvalue $1$.
Since $\wv$ is the eigenvector corresponding to the eigenvalue $0$,
the matrix $\Mt+\Pt$ has eigenvalues not less than $1$, and hence is
invertible.
We define
\begin{equation}
\Wt=(\Mt+\Pt)^{-1}.
\label{WDef}
\end{equation}
From (\ref{Mkdef}) and (\ref{PDef}), we find\footnote{
For arbitrary vectors ${\bf v}=(v_u)_{u\in\calU}$ and
${\bf w}=(w_u)_{u\in\calU}$, we define their Kronecker product as
${\bf v}\otimes{\bf w}=(v_uw_{u'})_{u,u'\in\calU}$ which can be 
regarded as a $b\times b$ matrix.
}
\begin{equation}
\Mt+\Pt=\It+\rbk{1+\frac{A(k)}{\la^2}}^{-1}
\av\otimes\av
\label{M+Prep}
\end{equation}
where $\It$ is the identity matrix, and the vector $\av=(\ak{u})_{u\in\calU}$ 
is defined as
\begin{equation}
\ak{u}=\cases{
{A(k)}/{\la^2}&if $u=o$;\cr
{C_u(k)}/{\la}&if $u\in\calU'$.\cr
}
\label{avDef}
\end{equation}
By using the representation (\ref{M+Prep}), and the general formula
\begin{equation}
\rbk{\It+\alpha\,{\bf v}\otimes{\bf v}}^{-1}
=
\It-\frac{\alpha}{\alpha({\bf v},{\bf v})+1}\,{\bf v}\otimes{\bf v},
\label{inversematrix}
\end{equation}
we find from (\ref{WDef}) that
\begin{equation}
\Wt=\It-\rbk{1+\frac{A(k)}{\la^2}}^{-2}\av\otimes\av,
\label{Wrep}
\end{equation}
where we noted $(\av,\av)=\cbk{A(k)/\la^2}+\cbk{A(k)/\la^2}^2$.

Following the philosophy of the Rayleigh-Schr\"{o}dinger perturbation theory,
we are going to express the eigenvector of (\ref{SchVec12}) (for
a fixed $k$) with the lowest
eigenvalue $\ep_1(k)$ as a power series in $\kappa$ as
\begin{equation}
\vov=\sum_{n=0}^\infty\kappa^n\vo{n},
\label{vexp}
\end{equation}
where $\vo{n}$ is a vector independent of $\kappa$.
We require
\begin{equation}
\vo{0}=\wv,
\label{v0=w}
\end{equation}
and
\begin{equation}
\rbk{\wv,\vo{n}}=0,
\label{wv=0}
\end{equation}
for any $n\ge1$.
We also express the eigenvalue as
\begin{equation}
\ep_1(k)=t\sum_{n=1}^\infty\kappa^n e_n(k),
\label{eexp}
\end{equation}
where the $0$-th order is vanishing since we have $\ep_1(k)=0$
when $\kappa=0$ (which corresponds to the flat-band model).

By substituting the expression (\ref{vexp}) into the Schr\"{o}dinger equation
(\ref{SchVec12}), and collecting the terms with the $n$-th power of $\kappa$,
we get 
\begin{equation}
\sumtwo{j,\ell\ge0}{(j+\ell=n-1)}
e_{j+1}(k)\,\vo{\ell}
=
\la^2\,\Mt\,\vo{n}+\Qt\,\vo{n-1},
\label{pertbasic}
\end{equation}
for any $n\ge1$.
In the present and the next subsection, summations like the above 
are always taken over integers (unless otherwise mentioned). 
The relation (\ref{pertbasic}) is the basis of our perturbation theory.
By taking the inner product with $\vo{0}=\wv$ in (\ref{pertbasic}),
we get
\begin{equation}
e_n(k)=\frac{(\wv,\Qt\vo{n-1})}{(\wv,\wv)}.
\label{alpharep}
\end{equation}
For $n\ge1$, we have $\Pt\vo{n}=0$ because of (\ref{wv=0}).
Thus, by using (\ref{WDef}), we can write
\begin{equation}
\Wt\Mt\vo{n}=\cbk{\Mt+\Pt}^{-1}\cbk{\Mt+\Pt}\vo{n}=\vo{n},
\label{WMv}
\end{equation}
for $n\ge1$.
Applying $\Wt$ from the left of (\ref{pertbasic}) and using
(\ref{WMv}), we get the recursion relation
\begin{equation}
\vo{n}=-\frac{1}{\la^2}\Wt\Qt\vo{n-1}
+\frac{1}{\la^2}\sumtwo{j,\ell\ge0}{(j+\ell=n-1)}
\frac{(\wv,\Qt\vo{j})}{(\wv,\wv)}\Wt\vo{\ell},
\label{basicrec}
\end{equation}
where we have substituted (\ref{alpharep}) for $e_j(k)$.
Since the right-hand side of (\ref{basicrec}) only contains
$\vo{m}$ with $m<n$,
we can in principle determine $\vo{n}$ with any $n$ by using 
(\ref{basicrec}) recursively.

Let us rewrite the recursion relation (\ref{basicrec}) in a more explicit form.
By substituting $(\wv,\wv)=1+(A(k)/\la^2)$ (which follows from (\ref{wDef})
and (\ref{Akdef})), and the explicit form (\ref{Wrep}) of $\Wt$,
we find that (\ref{basicrec}) becomes
\begin{equation}
\vo{n}={\bf V}^{(1)}_n+{\bf V}^{(2)}_n+{\bf V}^{(3)}_n+{\bf V}^{(4)}_n,
\label{rec2}
\end{equation}
with 
\begin{equation}
{\bf V}^{(1)}_n=-\frac{1}{\la^2}\Qt\,\vo{n-1},
\label{Vn1}
\end{equation}
\begin{equation}
{\bf V}^{(2)}_n=\frac{1}{\la^2}
\sum_{h=0}^\infty(h+1)\rbk{-\frac{A(k)}{\la^2}}^h
(\av,\Qt\vo{n-1})\,\av,
\label{Vn2}
\end{equation}
\begin{equation}
{\bf V}^{(3)}_n=\frac{1}{\la^2}
\sum_{h=0}^\infty\rbk{-\frac{A(k)}{\la^2}}^h
\sumtwo{j,\ell\ge0}{(j+\ell=n-1)}
(\wv,\Qt\vo{j})\,\vo{\ell},
\label{Vn3}
\end{equation}
and
\begin{equation}
{\bf V}^{(4)}_n=-\frac{1}{\la^2}
\sum_{h=0}^\infty\frac{(h+1)(h+2)}{2}
\rbk{-\frac{A(k)}{\la^2}}^h
\sumtwo{j,\ell\ge0}{(j+\ell=n-1)}
(\wv,\Qt\vo{j})(\av,\vo{\ell})\,\av.
\label{Vn4}
\end{equation}
\subsection{Recursive Bounds for the Perturbation Coefficients}
\label{secrecbound}
Let us construct the vector $\vov$ as in the expression (\ref{vexp})
by using the recursion relations (\ref{basicrec}), (\ref{rec2}),
along with the initial condition (\ref{v0=w}).
The construction proceeds in an inductive manner.
We first {\em assume} that the $u$-component (where $u\in\calU$)
of the vector $\vo{n}$ can be written as
\begin{eqnarray}
\rbk{\vo{n}}_u &=&
\rbk{\frac{1}{\la^2}}^n\sum_{m=0}^\infty\rbk{\frac{1}{\la}}^m
\sumthree{(s_i,t_i)\in\calU\times\calU}{{\rm with\ } i=1,\ldots,n}
{{\rm s.t.\ } (s_i,t_i)\le(s_{i+1},t_{i+1})}
\sumthree{u_j\in\calU'}{{\rm with\ } j=1,\ldots,m}
{{\rm s.t.\ } u_j\le u_{j+1}}\times
\ret
&&
\times
\alpha_1(u;\cbk{(s_i,t_i)},\cbk{u_j})
\rbk{\prod_{i=1}^n\Qk{s_i,t_i}}
\rbk{\prod_{j=1}^m\Ck{u_j}},
\label{vnbasic}
\end{eqnarray}
with $k$-independent coefficients $\alpha_1(u;\cbk{(s_i,t_i)},\cbk{u_j})$.
In (\ref{vnbasic}), the summation over $\cbk{(s_i,t_i)}_{i=1,\ldots,n}$
and $\cbk{u_j}_{j=1,\ldots,m}$ are restricted to the combinations
which satisfy $(s_i,t_i)\le(s_{i+1},t_{i+1})$ and
$u_j\le u_{j+1}$, respectively.
Here we have introduced an arbitrary complete ordering in the sets
$\calU\times\calU$ and $\calU'$.

Let us define
\begin{equation}
\alt_1(n,m)=\supthree{u}{\cbk{(s_i,t_i)}_{i=1,\ldots,n}}
{\cbk{u_j}_{j=1,\ldots,m}}
\abs{\alpha_1(u;\cbk{(s_i,t_i)},\cbk{u_j})},
\label{alpha1}
\end{equation}
where the sup is taken over all the possible combinations that appear
in (\ref{vnbasic}) with the given $m$.
The quantity $\alt_1(n,m)$ plays the essential role in our inductive proof.

From (\ref{wDef}), it is obvious that $\vo{0}=\wv$ can be written in the 
form (\ref{vnbasic}).
We also find that the recursion relation (\ref{basicrec}) ``preserves'' the
form (\ref{vnbasic}) since the recursion essentially consists of 
multiplications by $\Ck{u}$ (or $A(k)=\sum_{u\in\calU'}\cbk{\Ck{u}}^2$)
and the matrix elements of $\Qt$.
See (\ref{rec2}) and (\ref{Vn1})-(\ref{Vn4}).
This observation determines $\al_1(\cdots)$ uniquely, and formally ``proves''
the validity of the representation (\ref{vnbasic}) if one neglects the problem
of convergence.

Let us turn to the harder problem of controlling $\alt_1(n,m)$
inductively and proving convergence of the sum in (\ref{vnbasic}).
Our strategy is to substitute the expression (\ref{vnbasic})
for $\vo{1},\ldots,\vo{n-1}$ into the right-hand side of the recursion relation
(\ref{basicrec}), reorganize the resulting expressions for $\vo{n}$ so the it
becomes the form of (\ref{vnbasic}), and finally express the coefficients
$\al_1$ for $\vo{n}$ in terms of $\al_1$ for  $\vo{1},\ldots,\vo{n-1}$.
The final expression leads us to an upper bound for $\alt_1(n,m)$ in terms
of $\alt(n',m')$ with $n'<n$ and $m'\le m$.
See (\ref{alpha1j}), (\ref{alpha11}), (\ref{alpha122}), (\ref{alpha13}),
and (\ref{alpha14}).

The above procedure is easy to describe, but is too complicated to be 
executed explicitly.
We shall take a slightly less complicated way, where we skip the intermediate
calculations and directly get the final upper bounds for $\alt_1(n,m)$.
To avoid too much complication, we write the desired upper bound as
\begin{equation}
\alt_1(n,m)\le\sum_{j=1}^4\alt_1^{(j)}(n,m),
\label{alpha1j}
\end{equation}
where $\alt_1^{(j)}(n,m)$ are suitable upper bounds for the contributions to
$\al_1(n,m)$ from ${\bf V}_n^{(j)}$ in the recursion formula (\ref{rec2}).

To  bound the contribution from ${\bf V}_n^{(1)}$ (\ref{Vn1})
and get an upper bound $\altn{1}{1}$, 
we assume that $\vo{n}$
is written as (\ref{vnbasic}), and then ask which $\Qk{s_i,t_i}$
in (\ref{vnbasic})
comes form the $\Qt$ which explicitly appears in the right-hand side of 
(\ref{Vn1}).
Since there are at most $b^2$ different $\Qk{s,t}$'s, we can set
\begin{equation}
\altn{1}{1}(n,m)=b^2\,\altn{1}{1}(n-1,m).
\label{alpha11}
\end{equation}
 
To bound the contribution from ${\bf V}_n^{(2)}$ (\ref{Vn2}),
we note that one of the components of $\av$ (\ref{avDef}) is 
$A(k)/\la^2=\sum_{u\in\calU'}\cbk{\Ck{u}/\la}^2$, and
$(b-1)$-components of $\av$ are of the form $\Ck{u}/\la$.
By considering all the possible combinations of these components, we can set
\begin{eqnarray}
\altn{2}{1}(n,m)&=&
b^2\sum_{h=0}^\infty
(h+1)(b-1)^h
\Bigl\{
(b-1)^{2h+4}(b-1)^2\,\alt_1(n-1,m-2h-4)
\ret
&&
\hspace{2cm}
+(b-1)^{2h+3}(b-1)^2\,\alt_1(n-1,m-2h-3)
\ret&&
\hspace{2cm}
+(b-1)^{2h+2}(b-1)\,\alt_1(n-1,m-2h-2)
\Bigr\}.
\label{alpha12}
\end{eqnarray}
The prefactor $b^2$ appears for the same reason as in (\ref{alpha11}).
The factors $(b-1)^{2h+4}$, $(b-1)^{2h+3}$, and $(b-1)^{2h+2}$
are the upper bounds for the number of ways to identify $\Ck{u}$'s in
(\ref{vnbasic}) as coming from $A(k)$ or $\av$ in the right-hand side
of (\ref{Vn2}).
Since $A(k)$ contains products of two $\Ck{u}$'s, we have the common
factor $(b-1)^{2h}$.
For convenience, we reorganize (\ref{alpha12}) as
\begin{equation}
\altn{2}{1}(n,m)=\sum_{h,p\ge0}
b^2(b+1)^{h+m-p}(h+1)\,\alt_1(n-1,p)
\sum_{\mu=2,3,4}\xin{2}{\mu}\delta_{2h+p,m-\mu},
\label{alpha122}
\end{equation}
with $\xin{2}{2}=b-1$, and $\xin{2}{3}=\xin{2}{4}=(b-1)^2$.

The next term ${\bf V}_n^{(3)}$ (\ref{Vn3}) contains two $\vo{n'}$ vectors.
This means that we need to identify $\Ck{u}$'s in $\vo{n}$ (in the form
(\ref{vnbasic})) as either 1)~coming from $A(k)$ or $\wv$ explicitly
contained in (\ref{Vn3}), 2)~coming from $\vo{j}$, or
3)~coming from $\vo{\ell}$.
Identifications of $\Ck{u}$'s into the classes 2) and 3) requires a new
combinatoric estimate.
We need to count the number of ways to decompose $(p+q)$ objects
into $p$ objects and $q$ objects.
There are $(b-1)$ different kinds of objects, and we do not
distinguish between the objects of the same kind.
(Of course, the objects are $\Ck{u}$'s.)
A crude upper bound for the desired combinatoric number is obtained by considering
what are the possible contents of $p$ objects.
This observation shows that the desired number is bounded from above by
\begin{equation}
{p+(b-1)-1\choose (b-1)-1}
\le
\frac{\cbk{p+(b-1)-1}^{(b-1)-1}}{\cbk{(b-1)-1}!}.
\end{equation}
Since there is a similar estimate with $p$ replaced by $q$, the desired combinatoric
number is bounded from above by the quantity $F(b-1;p,q)$, where
\begin{equation}
F(g;p,q)=\min\cbk{
\frac{(p+g-1)^{g-1}}{(g-1)!},\frac{(q+g-1)^{g-1}}{(g-1)!}
}
\label{FDef}
\end{equation}
An analogous combinatoric problem arises when we identify $\Qk{s,t}$'s in
$\vo{n}$ as coming from either $\Qt$, $\vo{j}$, or $\vo{\ell}$ in the right-hand
side of (\ref{Vn3}).
Consequently we have the following upper bound for the contribution from
${\bf V}_n^{(3)}$;
\begin{eqnarray}
\altn{3}{1}(n,m)&=&\sum_{h,p,q\ge0}\sumtwo{j,\ell\ge0}{(j+\ell=n-1)}
b^2(b-1)^{h+m-(p+q)}F(b^2;j,\ell)F(b-1;p,q)\times
\ret
&&
\times
\alt_1(j,p)\,\alt_1(\ell,q)
\sum_{\mu=0,1}\xin{3}{\mu}\delta_{2h+p+q,m-\mu},
\label{alpha13}
\end{eqnarray}
with $\xin{3}{0}=1$ and $\xin{3}{1}=b-1$.

Finally the contribution from ${\bf V}_n^{(4)}$ (\ref{Vn4}) can be bounded in a similar
manner as
\begin{eqnarray}
&&\altn{4}{1}(n,m)=
\sum_{h,p,q\ge0}\sumtwo{j,\ell\ge0}{(j+\ell=n-1)}
b^2\frac{(h+1)(h+2)}{2}(b-1)^{h+m-(p+q)}\times
\ret&&
\quad\times
F(b^2;j,\ell)F(b-1;p,q)\,\alt_1(j,p)\,\alt_1(\ell,q)
\sum_{\mu=2,3,4,5}\xin{4}{\mu}\delta_{2h+p+q,m-\mu},
\label{alpha14}
\end{eqnarray}
with $\xin{4}{2}=(b-1)$, $\xin{4}{3}=b(b-1)$, $\xin{4}{4}=b(b-1)^2$,
and $\xin{4}{5}=(b-1)^3$.
\subsection{Upper Bounds for the Coefficients}
\label{secinduction}
We shall carry out an inductive proof of upper bounds for $\alt_1(n,m)$.
Our inductive assumption is that there are constants $\beta$ and $\gamma$
(determined later) which depend only on the band number $b$, 
and we have\footnote{
The factors $(n'+1)^{b^2+1}$ and $(m'+1)^b$ are indispensable for carrying
out our inductive proof.
We do not mean, however, 
that (\ref{theass}) represents the correct asymptotic
behavior of $\alt_1(n',m')$.
}
\begin{equation}
\alt_1(n',m')\le\frac{\beta^{n'}}{(n'+1)^{b^2+1}}\frac{\gamma^{m'}}{(m'+1)^b},
\label{theass}
\end{equation}
for any nonnegative $n'$ and $m'$, such that $n'<n$ and $m'\le m$.
Our goal is to prove  the same bound for $n'=n$ and $m'=m$.
Since $\vo{0}=\wv$, we have $\alt_1(0,0)=1$
(by comparing (\ref{wDef}), (\ref{vnbasic}), and (\ref{alpha1})),
which clearly satisfies the assumption (\ref{theass}) provided that 
$\beta,\gamma\ge1$.

In what follows, we shall bound each of $\altn{j}{1}(n,m)$
in (\ref{alpha1j}) by using the assumption (\ref{theass}).
We start from $\altn{1}{1}(n,m)$.
Since the right-hand side of (\ref{alpha11}) contains only $\alt_1(n',m')$
with $n'<n$ and $m'\le m$, we can use the assumption (\ref{theass}) to get
\begin{eqnarray}
\altn{1}{1}(n,m)&\le&b^2\frac{\beta^{n-1}}{n^{b^2+1}}
\frac{\gamma^m}{(m+1)^b}
\ret
&=&
\frac{\beta^n}{(n+1)^{b^2+1}}\frac{\gamma^m}{(m+1)^b}
\times
\cbk{\frac{b^2}{\beta}\rbk{\frac{n+1}{n}}^{b^2+1}}.
\label{alpha11b}
\end{eqnarray}
Note that we have factored out the desired quantity in front.

Next we investigate $\altn{2}{1}(n,m)$ by substituting the assumption
(\ref{theass}) into (\ref{alpha122}).
Again we factor out the desired quantity to get
\begin{eqnarray}
&&
\altn{2}{1}(n,m)
\ret&&
\le
\sum_{h,p\ge0}b^2(b+1)^{h+m-p}(h+1)
\frac{\beta^{n-1}}{n^{b^2+1}}\frac{\gamma^p}{(p+1)^b}
\sum_{\mu=2,3,4}\xin{2}{\mu}\delta_{2h+p,m-\mu}
\ret&&
=\frac{\beta^n}{(n+1)^{b^2+1}}\frac{\gamma^m}{(m+1)^b}
\times
\cbk{\frac{b^2}{\beta}\rbk{\frac{n+1}{n}}^{b^2+1}}\times
\ret&&
\times
\cbk{
\sum_{h,p\ge0}(h+1)\rbk{\frac{m+1}{p+1}}^b(b-1)^h
\rbk{\frac{b-1}{\gamma}}^{m-p}
\sum_{\mu=2,3,4}\xin{2}{\mu}\delta_{2h+p,m-\mu}
}.
\label{alpha123}
\end{eqnarray}
We bound the sum over $h$ and $p$ as
\begin{eqnarray}
&&
\sum_{h,p\ge0}(h+1)\rbk{\frac{m+1}{p+1}}^b(b-1)^h
\rbk{\frac{b-1}{\gamma}}^{m-p}
\sum_{\mu=2,3,4}\xin{2}{\mu}\delta_{2h+p,m-\mu}
\ret&&
=\sum_{\mu=2,3,4}\xin{2}{\mu}\rbk{\frac{b-1}{\gamma}}^\mu
\quad
\sum_{h=0}^{\sbk{(m-\mu)/2}}
\rbk{\frac{m+1}{m-2h-\alpha}}^b(h+1)\cbk{\frac{\rbk{b-1}^3}{\gamma^2}}^h
\ret&&
\le
\sum_{\mu=2,3,4}\xin{2}{\mu}\rbk{\frac{b-1}{\gamma}}^\mu
\left[
\sum_{h=0}^{\sbk{(m-2\mu-1)/4}}2^b(h+1)
\cbk{\frac{\rbk{b-1}^3}{\gamma^2}}^h
\right.
\ret&&
\quad\left.
+\sum_{h=\sbk{(m-2\mu-1)/4}+1}^\infty
(m+1)^b(h+1)\cbk{\frac{\rbk{b-1}^3}{\gamma^2}}^h
\right]
\ret&&
\le
\sum_{\mu=2,3,4}\xin{2}{\mu}
\rbk{\frac{b-1}{\gamma}}^\mu
\rbk{1-\frac{(b-1)^3}{\gamma}}^{-2}
\times
\ret&&
\times
\Biggl\{
2^b+ 2^b\rbk{\frac{m-2\mu-1}{4}+1}
\rbk{\frac{(b-1)^3}{\gamma^2}}^{(m-2\mu+7)/4}
\ret&&\quad
+(m+1)^b\rbk{\frac{m-2\mu-1}{4}+2}
\rbk{\frac{(b-1)^3}{\gamma^2}}^{(m-2h+3)/4}
\Biggr\}
\ret&&
\le
2^{b+1}\sum_{\mu=2,3,4}\xin{2}{\mu}\rbk{\frac{b-1}{\gamma}}^\mu,
\label{alpha124}
\end{eqnarray}
where $[\cdots]$ is the Gauss symbol.
The final inequality in (\ref{alpha124}) is valid for sufficiently large $\gamma$.
By substituting (\ref{alpha124}) into (\ref{alpha123}), we get
\begin{eqnarray}
\altn{2}{1}(n,m)
&\le&
\frac{\beta^n}{(n+1)^{b^2+1}}\frac{\gamma^m}{(m+1)^b}
\times
\cbk{\frac{b^2}{\beta}\rbk{\frac{n+1}{n}}^{b^2+1}}
\times
\ret&&
\times
2^{b+1}\cbk{
\frac{(b-1)^3}{\gamma^2}+\frac{(b-1)^5}{\gamma^3}+\frac{(b-1)^6}{\gamma^4}
}.
\label{alpha12b}
\end{eqnarray}

We postpone the estimate of $\altn{3}{1}(n,m)$, and treat $\altn{4}{1}(n,m)$.
Again by substituting the inductive assumption (\ref{theass}) into
(\ref{alpha14}), we get
\begin{eqnarray}
&&
\altn{4}{1}(n,m)
\ret&&
\le
\sum_{h,p,q\ge0}\sumtwo{j,\ell\ge0}{(j+\ell=n-1)}
\Biggl\{
b^2\frac{(h+1)(h+2)}{2}(b-1)^{h+m-(p+q)}
F(b^2;j,\ell)F(b-1;p,q)\times
\ret&&
\quad
\times
\frac{\beta^j}{(j+1)^{b^2+1}}\frac{\gamma^p}{(p+1)^b}
\frac{\beta^\ell}{(\ell+1)^{b^2+1}}\frac{\gamma^q}{(q+1)^b}
\sum_{\mu=2,3,4,5}\xin{4}{\mu}\delta_{2h+p+q,m-\mu}
\Biggr\}
\ret&&
=\frac{\beta^n}{(n+1)^{b^2+1}}\frac{\gamma^m}{(m+1)^b}
\times\frac{b^2}{\beta}\,S_1S_2,
\label{alpha142}
\end{eqnarray}
with 
\begin{equation}
S_1=\sumtwo{j,\ell\ge0}{(j+\ell=n-1)}
\rbk{\frac{n+1}{(j+1)(\ell+1)}}^{b^2+1}F(b^2;j,\ell),
\label{S1}
\end{equation}
and
\begin{eqnarray}
S_2&=&\sum_{h,p,q\ge0}
\Biggl\{
\frac{(h+1)(h+2)}{2}(b-1)^j\rbk{\frac{b-1}{\gamma}}^{m-(p+q)}
\times
\ret&&
\times
\rbk{\frac{m+1}{(p+1)(q+1)}}^bF(b-1;p,q)
\sum_{\mu=2,3,4,5}\xin{4}{\mu}\delta_{2h+p+q,m-\mu}
\Biggr\}.
\label{S2}
\end{eqnarray}

We first bound $S_1$.
By using the symmetry between $j$ and $\ell$ in (\ref{S1})
and in $F(b^2;j,\ell)$, we have
\begin{equation}
S_1\le2\sum_{j=0}^{[(n-1)/2]}
\rbk{\frac{n+1}{(j+1)(n-j)}}^{b^2+1}
\frac{(j+b^2-1)^{b^2-1}}{(b^2-1)!}.
\label{S1b1}
\end{equation} 
By noting that the bounds $n-j\ge(n+1)/2$ and 
$(j+b^2-1)\le(b^2-1)(j+1)$ hold within the range of the summation,
we can further bound $S_1$ as
\begin{eqnarray}
S_1&\le&2\cbk{\frac{n+1}{(n+1)/2}}^{b^2+1}\quad
\sum_{j=0}^{[(n-1)/2]}
\frac{(b^2-1)^{b^2-1}(j+1)^{b^2-1}}{(j+1)^{b^2+1}(b^2-1)!}
\ret
&\le&
2\,2^{b^2+1}\frac{(b^2-1)^{b^2-1}}{(b^2-1)!}\sum_{j=0}^\infty(j+1)^{-2}
\ret
&=&J(b^2),
\label{S1b2}
\end{eqnarray}
where we introduced
\begin{equation}
J(g)=\frac{\pi^2}{3}2^{g+1}\frac{(g-1)^{g-1}}{(g-1)!}.
\label{GDef}
\end{equation}
The quantity $S_2$ (\ref{S2}) can be bounded by combining the techniques used
in the bounds (\ref{alpha123}), (\ref{alpha124}), and in the bounds (\ref{S1b1}),
(\ref{S1b2}).  The resulting bound is
\begin{equation}
S_2\le J(b-1)\,2^{b+1}\cbk{
\frac{(b-1)^3}{\gamma^2}+\frac{b(b-1)^4}{\gamma^3}
+\frac{b(b-1)^6}{\gamma^4}+\frac{(b-1)^8}{\gamma^5}
}.
\label{S2b}
\end{equation}
By substituting (\ref{S1b2}) and (\ref{S2b}) into (\ref{alpha142}),
we finally get
\begin{eqnarray}
\altn{4}{1}(n,m)&\le&
\frac{\beta^n}{(n+1)^{b^2+1}}\frac{\gamma^m}{(m+1)^b}
\times\frac{b^2}{\beta}\,J(b^2)J(b-1)
\times\ret
&&
\times
2^{b+1}\cbk{
\frac{(b-1)^3}{\gamma^2}+\frac{b(b-1)^4}{\gamma^3}
+\frac{b(b-1)^6}{\gamma^4}+\frac{(b-1)^8}{\gamma^5}
}.
\label{alpha14b}
\end{eqnarray}

The quantity $\altn{3}{1}(n,m)$ (\ref{alpha13}) can be bounded in the
same manner as $\altn{4}{1}$.
The resulting bound is
\begin{equation}
\altn{3}{1}(n,m)\le
\frac{\beta^n}{(n+1)^{b^2+1}}\frac{\gamma^m}{(m+1)^b}
\times\frac{b^2}{\beta}\,J(b^2)J(b-1)\,
2^{b+1}\rbk{1+\frac{(b-1)^2}{\gamma}}.
\label{alpha13b}
\end{equation}

Finally, by recalling (\ref{alpha1j}),
we sum up the bounds
(\ref{alpha11b}), (\ref{alpha12b}), (\ref{alpha14b}), 
and (\ref{alpha13b}) to bound  $\alt_1(n,m)$ as
\begin{eqnarray}
&&
\alt_1(n,m)
\le
\frac{\beta^n}{(n+1)^{b^2+1}}\frac{\gamma^m}{(m+1)^b}
\times
\ret&&
\times\frac{b^2}{\beta}
\left[
2^{b^2+1}+
2^{b^2+b+2}\cbk{
\frac{(b-1)^3}{\gamma^2}+\frac{(b-1)^5}{\gamma^3}+\frac{(b-1)^6}{\gamma^4}
}+\right.
\ret&&
\left.
+2^{b+1}J(b^2)J(b-1)
\cbk{
1+\frac{(b-1)^2}{\gamma}+
\frac{(b-1)^3}{\gamma^2}+\frac{b(b-1)^4}{\gamma^3}
+\frac{b(b-1)^6}{\gamma^4}+\frac{(b-1)^8}{\gamma^5}
}
\right]
\ret
&&\le
\frac{\beta^n}{(n+1)^{b^2+1}}\frac{\gamma^m}{(m+1)^b},
\label{alpha1b}
\end{eqnarray}
where the final bound holds for sufficiently large $\beta$
and $\gamma$.
Note that how large these constants should be depend only on
the band number $b$.
Since the bound (\ref{alpha1b}) has precisely the same form as the
inductive assumption (\ref{theass}), we have proved that
$\alt_1(n',m')$ satisfies the bound (\ref{theass}) for any $n',m'\ge0$.
\subsection{Construction of the Vector $\vov$}
We are now ready to construct the ground state vector $\vov$, which played
essential role in our construction in Sections~\ref{seckspace}
and \ref{secrealspace}.
By substituting the series (\ref{vnbasic}) into (\ref{vexp}), we
get the following power series expression for $\vov=(\vok{u})_{u\in\calU}$.
\begin{eqnarray}
\vok{u}&=&
\wk{u}+
\sum_{n=1}^\infty\rbk{\frac{\kappa}{\la^2}}^n
\sum_{m=0}^\infty\rbk{\frac{1}{\la}}^m
\sumthree{(s_i,t_i)\in\calU\times\calU}{{\rm with\ } i=1,\ldots,n}
{{\rm s.t.\ } (s_i,t_i)\le(s_{i+1},t_{i+1})}
\sumthree{u_j\in\calU'}{{\rm with\ } j=1,\ldots,m}
{{\rm s.t.\ } u_j\le u_{j+1}}
\times
\ret
&&\times
\al_1(u;\cbk{(s_i,t_i)},\cbk{u_j})
\rbk{\prod_{i=1}^n\Qk{s_i,t_i}}
\rbk{\prod_{j=1}^m\Ck{u_j}}.
\label{vexp2}
\end{eqnarray}
See the discussion following (\ref{vnbasic}) for the range of the summations.
To investigate the convergence of (\ref{vexp2}), we note that
(\ref{Ckdef}) and (\ref{Ff}) imply
\begin{equation}
\abs{\Ck{u}}\le\abs{\calF_u}=2^\nu,
\label{Ckb}
\end{equation}
and 
\begin{equation}
\abs{\Qk{s,t}}\le1.
\label{Qkb}
\end{equation}
By using the above two bounds, the definition (\ref{alpha1}) of $\alt_1(n,m)$,
and the basic bound (\ref{theass}),
and by noting that the numbers of possible combinations of 
$\cbk{(s_i,t_i)}_{i=1,\ldots,n}$ and $\cbk{u_j}_{j=1,\ldots,m}$
are bounded from above by $b^{2n}$ and $(b-1)^m$, respectively,
we find that the absolute value of the summand in (\ref{vexp2})
for each pair of $n$ and $m$ is bounded from above by
\begin{equation}
\rbk{\frac{\akappa}{\la^2}}^n\rbk{\frac{1}{\la}}^m
b^{2n}(b-1)^m
\frac{\beta^n}{(n+1)^{b^2+1}}\frac{\gamma^m}{(m+1)^b}2^{\nu m}.
\label{absv<}
\end{equation}
The quantity (\ref{absv<}) is summable in $n$ and $m$ provided that 
$(\akappa/\la^2)b^2\beta<1$ and 
\newline
$\la^{-1}(b-1)\gamma2^\nu<1$.
If this is the case, the infinite sum in (\ref{vexp2}) is absolutely convergent.
This completes our construction of the ground state vector $\vov$.

Let us summarize the present result as the following lemma.
\begin{lemma}
\label{v0Lemma}
There exist positive constants $\beta$ and $\gamma$ which depend only on
the band number $b$.
When the parameters $\la$ and $\kappa$ satisfy
\begin{equation}
\frac{\akappa}{\la^2}<\frac{1}{b^2\beta},
\label{k/lam<1}
\end{equation}
and
\begin{equation}
\la>2^\nu(b-1)\gamma,
\label{lam>1}
\end{equation}
the ground state vector $\vov$ (characterized by (\ref{SchVec2}))
is expressed by the absolutely convergent sum (\ref{vexp2}).
The coefficients $\al_1(u,\{s_i,t_i\},\{u_j\})$ in (\ref{vexp2})
are independent of $k$, and satisfy the bound
\begin{equation}
\supthree{u}{\{(s_i,t_i)\}_{i=1,\ldots,n}}{\{u_j\}_{j=1,\ldots,m}}
\abs{\al_1(u,\{s_i,t_i\},\{u_j\})}
\le
\frac{\beta^n}{(n+1)^{b^2+1}}\frac{\gamma^m}{(m+1)^b},
\label{a1bound}
\end{equation}
for any $n$, $m$.
\end{lemma}

\subsection{Dispersion Relation}
Let us investigate the dispersion relation $\ep_1(k)$ for the lowest band,
which appears, e.g., in (\ref{SchVec2}).
By substituting the expression (\ref{alpharep}) into the formal expansion
(\ref{eexp}) for $\ep_1(k)$, we find
\begin{equation}
\ep_1(k)=t\sum_{n=1}^\infty
\kappa^n\,
\frac{(\wv,\Qt\vo{n-1})}{(\wv,\wv)}.
\label{e=v}
\end{equation}
Since $\vo{n-1}$ is expressed as the convergent expansion (\ref{vnbasic}),
it is apparent from (\ref{e=v}) that there is a similar power series expansion
for $\ep_1(k)$.

In fact, by substituting the expansion (\ref{vnbasic}) into (\ref{e=v}),
and performing some estimates similar to those in Section~\ref{secrecbound},
we get the expansion
\begin{eqnarray}
\ep_1(k)
&=&
\la^2t
\sum_{n=1}^\infty\rbk{\frac{\kappa}{\la^2}}^n
\sum_{m=0}^\infty\rbk{\frac{1}{\la}}^m
\sumthree{(s_i,t_i)\in\calU\times\calU}{{\rm with\ } i=1,\ldots,n}
{{\rm s.t.\ } (s_i,t_i)\le(s_{i+1},t_{i+1})}
\sumthree{u_j\in\calU'}{{\rm with\ } j=1,\ldots,m}
{{\rm s.t.\ } u_j\le u_{j+1}}
\times\ret&&\times
\al_2(\cbk{(s_i,t_i)},\cbk{u_j})
\rbk{\prod_{i=1}^n\Qk{s_i,t_i}}
\rbk{\prod_{j=1}^m\Ck{u_j}},
\label{ekexp}
\end{eqnarray}
where the coefficient $\al_2(\cbk{(s_i,t_i)},\cbk{u_j})$ is independent of
$k$.
The range of the sums over $\{s_i,t_i\}$ and $\{u_j\}$ are the same as those
in (\ref{vnbasic}), (\ref{vexp2}).
For any $n\ge2$ and $m\ge0$, the coefficient $\al_2(\cbk{(s_i,t_i)},\cbk{u_j})$
in (\ref{ekexp}) satisfies the bound
\begin{eqnarray}
&&
\suptwo{\{(s_i,t_i)\}_{i=1,\ldots,n}}{\{u_j\}_{j=1,\ldots,m}}
\abs{\al_2(\{s_i,t_i\},\{u_j\})}
\ret
&&
\le b^2\sum_{h,p\ge0}(b-1)^h(b-1)^{m-p}
\rbk{\sum_{\mu=0,1}\xi_\mu\,\delta_{2h+p,m-\mu}}
\alt_1(n-1,p),
\label{alpha2<1}
\end{eqnarray}
with $\xi_0=1$, $\xi_1=b-1$.
Substituting the bound (\ref{theass}) for $\alt_1$ into the right-hand side
of (\ref{alpha2<1}) and performing estimates similar to those in
Section~\ref{secinduction}, we find for $n\ge1$ and $m\ge0$ that\footnote{
The estimate for $n=1$ follows from explicit calculation.
}
\begin{equation}
\suptwo{\{(s_i,t_i)\}_{i=1,\ldots,n}}{\{u_j\}_{j=1,\ldots,m}}
\abs{\al_2(\{s_i,t_i\},\{u_j\})}
\le
\frac{\beta^n}{(n+1)^{b^2+1}}\frac{\gamma^m}{(m+1)^b},
\label{alpha2<}
\end{equation}
again for sufficiently large $\beta$ and $\gamma$.

The bound (\ref{alpha2<}), along with (\ref{Ckb}) and (\ref{Qkb}), proves the 
convergence of the sum (\ref{ekexp}) for $\la$ and $\kappa$ satisfying
the conditions (\ref{k/lam<1}) and (\ref{lam>1}).
\subsection{Dual Vectors}
We shall develop power series expansions for the dual vectors $\vt{u}$
(with $u\in\calU$) defined in (\ref{dualvDef}).
By recalling the definition (\ref{veuk}) of the vector $\vvv{e}$ with $e\in\calU'$,
the components of the Gramm matrix $G(k)$ (\ref{GrammDef}) can be expressed as
$(G(k))_{o,o}=\abs{\vov}^2$,
$(G(k))_{o,e}=(G(k))_{e,o}=0$,
$(G(k))_{e,e'}=\vok{e}\rbks{\vok{e'}}$, and
$(G(k))_{e,e}=\abs{\vok{o}}^2+\abs{\vok{e}}^2$,
where $e,e'\in\calU'$ and $e\ne e'$.
Thus the $b\times b$ matrix $G(k)$ can be compactly written in the form
\begin{equation}
G(k) =
\rbk{\matrix{
|\vov|^2&0&\ldots&0\cr
0&&&\cr
\vdots&&\Ht&\cr
0&&&\cr
}},
\label{Grep}
\end{equation}
where the $(b-1)\times(b-1)$ matrix $\Ht$ is given by\footnote{
$\Ht$ and the identity matrix $\It$ in (\ref{HkDef}) and (\ref{Hinv}) are 
the only $(b-1)\times(b-1)$ matrices that appear in the present paper.
Similarly $\gv$ is the only $(b-1)$-dimensional vector.
}
\begin{equation}
\Ht=|\,\vok{o}|^2\,\It+\gv\otimes\gv^*,
\label{HkDef}
\end{equation}
with the $(b-1)$-dimensional vectors 
$\gv=(\vok{e})_{e\in\calU'}$
and $\gv^*=\rbk{\rbks{\vok{e}}}_{e\in\calU'}$.

It is evident from (\ref{Grep}) that the inverse of the Gramm matrix is written as
\begin{equation}
G(k)^{-1}= \rbk{\matrix{
|\vov|^{-2}&0&\ldots&0\cr
0&&&\cr
\vdots&&\Ht^{-1}&\cr
0&&&\cr
}},
\label{Ginvrep}
\end{equation}
As for the inverse of $\Ht$, we use the general formula (\ref{inversematrix})
to get
\begin{equation}
\Ht^{-1}=\frac{1}{\abs{\vok{o}}^2}
\rbk{\It-\frac{1}{\abs{\vov}^2}\,\gv\otimes\gv^*}.
\label{Hinv}
\end{equation}

By substituting (\ref{Ginvrep}) and (\ref{Hinv}) to the definition (\ref{dualvDef})
of the dual vectors, we get
\begin{equation}
\vt{o}=\frac{1}{\abs{\vov}^2}\vov,
\label{vdo}
\end{equation}
and
\begin{equation}
\vt{e}=\frac{1}{\abs{\vok{o}}^2}\,\vvv{e}
-\frac{\rbks{\vok{e}}}{\abs{\vov}^2}
\sum_{e'\in\calU'}\vok{e'}\,\vvv{e'},
\label{vde}
\end{equation}
where $e\in\calU'$.
We again denote the components of the dual vectors as
$\vt{u}=(\vkt{u}{u'})_{u'\in\calU}$.
By using (\ref{veuk}), these two equations lead us to the following expressions
for the components of the dual vectors in terms of the components of the 
ground state vector $\vov$;
\begin{equation}
\vkt{o}{u}=\abs{\vov}^{-2}\vok{u},
\label{vdou}
\end{equation}
for $u\in\calU$,
\begin{equation}
\vkt{e}{o}=-\abs{\vov}^{-2}\rbks{\vok{e}},
\label{vdeo}
\end{equation}
\begin{equation}
\vkt{e}{e'}=-\abs{\vov}^{-2}\rbks{\vok{o}}\rbks{\vok{e}}
\vok{e'},
\label{vdeep}
\end{equation}
and
\begin{equation}
\vkt{e}{e}=\frac{1}{\vok{o}}
-\frac{\rbks{\vok{o}}\abs{\vok{e}}^2}{\abs{\vov}^2},
\label{vdee}
\end{equation}
where $e,e'\in\calU'$ and $e\ne e'$.

Recalling that each $\vok{u}$ admits the power series expansion (\ref{vexp2}),
it is clear from the expressions (\ref{vdou}), (\ref{vdeo}), (\ref{vdeep}), and
(\ref{vdee}) that there are similar expansions for the dual vectors.
In order to control these expansions, we substitute (\ref{vexp2}) into
(\ref{vdou}), (\ref{vdeo}), (\ref{vdeep}),  (\ref{vdee}), and reorganize the
resulting expressions into transparent series expansions.
Then by using the bounds (\ref{a1bound}) for the coefficients in (\ref{vexp2}),
we can control the coefficients of the new expansions for the dual vectors.
Unfortunately this straightforward procedure turns out to be rather tedious 
to carry out in practice.
We shall omit the details here since the required estimates are quite similar
to those in Sections~\ref{secrecbound} and \ref{secinduction}.

The resulting series expansions for the dual vectors can be written as
\begin{eqnarray}
\vkt{u}{u'}
&=&
\wkt{u}{u'}+
\sumtwo{g,\ell\ge0}{(g+\ell\ge1)}\rbk{\frac{\kappa}{\la^2}}^{g+\ell}
\sum_{m\ge0}\rbk{\frac{1}{\la}}^m
\times
\ret&&
\times
\sumthree{(q_h,r_h)\in\calU\times\calU}
{{\rm with\ } h=1,\ldots,g}
{{\rm s.t.\ } (q_h,r_h)\le(q_{h+1},r_{h+1})}
\quad
\sumthree{(s_i,t_i)\in\calU\times\calU}
{{\rm with\ } i=1,\ldots,\ell}
{{\rm s.t.\ } (s_i,t_i)\le(s_{i+1},t_{i+1})}
\quad
\sumthree{u_j\in\calU'}{{\rm with\ } j=1,\ldots,m}
{{\rm s.t.\ } u_j\le u_{j+1}}
\times
\ret&&
\times
\al_3(u,u';\cbk{(q_h,r_h)},\cbk{(s_i,t_i)},\cbk{u_j})
\times
\ret&&
\times
\cbk{\prod_{h=1}^g\Qk{q_h,r_h}}
\cbk{\prod_{i=1}^\ell\rbks{\Qk{s_i,t_i}}}
\cbk{\prod_{j=1}^m\Ck{u_j}}.
\label{vdexp}
\end{eqnarray}
We have introduced the dual vectors $\wt{u}=(\wkt{u}{u'})_{u'\in\calU}$
with $u\in\calU$ for the model with $\kappa=0$ (i.e., the flat-band model).
By using (\ref{vdou}), (\ref{vdeo}), (\ref{vdeep}), and (\ref{vdee}) with $\vov$
replaced by $\wv$, and the definition (\ref{wDef}) of $\wv$, we find
\begin{equation}
\wkt{u}{u'}=
\left\{
\begin{array}{ll}
\rbk{1+{A(k)}/{\la^2}}^{-1}\times
\left\{
\begin{array}{l}
1\\
-{\Ck{u'}}/{\la}\\
{\Ck{u}}/{\la}\\
-{\Ck{u}\Ck{u'}}/{\la^2}
\end{array}
\right.
&
\begin{array}{l}
\mbox{if $u=u'=o$;}\\
\mbox{if $u=o$, $u'\in\calU'$;}\\
\mbox{if $u\in\calU'$, $u'=o$;}\\
\mbox{if $u,u'\in\calU'$, $u\ne u'$;}
\end{array}
\\
1-\rbk{1+{A(k)}/{\la^2}}^{-1}{(\Ck{u})^2}/{\la^2}
&
\begin{array}{l}
\mbox{if $u=u'\in\calU'$,}
\end{array}
\end{array}
\right.
\label{wtil}
\end{equation}
where we have used (\ref{Akdef}).

The coefficients $\al_3$ in the expansion (\ref{vdexp}) can be shown to satisfy for
each $n\ge1$, $m\ge0$ the bound
\begin{equation}
\suptwo{g,\ell\ge0}{(g+\ell=n)}
\supfour{u,u'}{\{q_h,r_h\}_{h=1,\ldots,g}}
{\{s_i,t_i\}_{i=1,\ldots,\ell}}{\{u_j\}_{j=1,\ldots,m}}
\abs{\al_3(u,u';\cbk{(q_h,r_h)},\cbk{(s_i,t_i)},\cbk{u_j})}
\le
C_b\,\tilde{\beta}^n\,\tilde{\gamma}^m,
\label{a3bound}
\end{equation}
where
\begin{equation}
\tilde{\beta}=8b^4\beta,\quad\tilde{\gamma}=8(b-1)\gamma,
\label{btgt}
\end{equation}
and $C_b$ is a constant which depends only on $b$.

By using the bounds (\ref{a3bound}), (\ref{Ckb}), and (\ref{Qkb}),
we can show that the power series (\ref{vdexp}) for the dual vectors
converge provided that
\begin{equation}
\frac{\akappa}{\la^2}\le r_0=\frac{\theta}{b^2\tilde{\beta}},
\label{k/lam<2}
\end{equation}
and
\begin{equation}
\la\ge\la_0=\frac{2^\nu(b-1)\tilde{\gamma}}{\theta},
\label{lam>2}
\end{equation}
with a constant $0<\theta<1$, which we shall now fix.
Note that the conditions (\ref{k/lam<1}) and (\ref{lam>1})
(required for the convergence of the series for $\vov$ and $\ep_1(k)$)
are automatically satisfied if we assume the above (\ref{k/lam<2})
and (\ref{lam>2}).
This completes our construction of the basis states.
\subsection{Summability of the Basis States}
\label{secreg}
It only remains to prove the summability of the basis states
$\phi^{(x)}$, $\phitil^{(x)}$, and the effective hopping $\tau_{y,x}$
stated in Lemmas~\ref{basisLemma}, \ref{dualbasisLemma}, and
\ref{tauLemma}.
It turns out that these bounds are natural consequences of the series expansions
(\ref{vexp2}), (\ref{ekexp}), and (\ref{vdexp}).

Let us look at the proofs of the bounds (\ref{reg1}) and (\ref{reg2}) in detail.
We first recall that, for $x\in\Lao$, the strictly localized basis states
$\psi^{(x)}$ (defined by (\ref{psi1}) and (\ref{psi2})) is written in terms of
$\wv$ (\ref{wDef}) as
\begin{equation}
\psis{x}{y}=(2\pi)^{-d}\int dk\,\emik{(x-y)}w_{\mu(y)}(k).
\label{psi=w}
\end{equation}
See Section~\ref{secrealspace} for the notations.
Note that (\ref{psi=w}) is a special case of (\ref{phiDef}).
From (\ref{phiDef}), (\ref{psi=w}), and the expansion (\ref{vexp2}),
we find for $x\in\Lao$ that
\begin{eqnarray}
\phis{x}{y}-\psis{x}{y}
&=&
(2\pi)^{-d}\int dk\,\emik{(x-y)}\rbk{\vok{\mu(y)}-w_{\mu(y)}(k)}
\ret
&=&
\sum_{n=1}^\infty\rbk{\frac{\kappa}{\la^2}}^n
\sum_{m=0}^\infty\rbk{\frac{1}{\la}}^m
\sumthree{(s_i,t_i)\in\calU\times\calU}{{\rm with\ } i=1,\ldots,n}
{{\rm s.t.\ } (s_i,t_i)\le(s_{i+1},t_{i+1})}
\sumthree{u_j\in\calU'}{{\rm with\ } j=1,\ldots,m}
{{\rm s.t.\ } u_j\le u_{j+1}}\times
\ret&&\times
\al_1(\mu(y);\cbk{(s_i,t_i)},\cbk{u_j})\,
I_{x,y}(\cbk{(s_i,t_i)},\cbk{u_j}),
\label{phi-psi}
\end{eqnarray}
with
\begin{equation}
I_{x,y}(\cbk{(s_i,t_i)},\cbk{u_j})
=(2\pi)^{-d}\int dk\,\emik{(x-y)}
\rbk{\prod_{i=1}^n\Qk{s_i,t_i}}
\rbk{\prod_{j=1}^m\Ck{u_j}}.
\label{Ixy}
\end{equation}
Recalling the definitions (\ref{Qdef}) and (\ref{Ckdef}), we find that
\begin{equation}
\sum_{y\in\La}\abs{I_{x,y}(\cbk{(s_i,t_i)},\cbk{u_j})}
\le
\cbk{\max_{z\in\La}\rbk{\frac{1}{t}\sum_{w\in\La}|t'_{z,w}|}}^n
\rbk{2^\nu}^m
\le 2^{\nu m},
\label{sumI<}
\end{equation}
where we used (\ref{tp<t}) and $\abs{\calF_f}=2^\nu$.
Similarly we have
\begin{eqnarray}
&&
\sum_{y\in\La}\abs{x-y}\abs{I_{x,y}(\cbk{(s_i,t_i)},\cbk{u_j})}
\ret&&
\le n
\cbk{\max_{z\in\La}\rbk{\frac{1}{t}\sum_{w\in\La}|t'_{z,w}|}}^{n-1}
\cbk{\max_{z\in\La}\rbk{\frac{1}{t}\sum_{w\in\La}|w-z|\,|t'_{z,w}|}}
\rbk{2^\nu}^m
\ret&&
\quad
+m\cbk{\max_{z\in\La}\rbk{\frac{1}{t}\sum_{w\in\La}|t'_{z,w}|}}^n
\rbk{2^\nu}^{m-1}\rbk{\sum_{g\in\calF_f}|g|}
\ret&&
\le\rbk{nR+m\frac{\sqrt{\nu}}{2}}2^{\nu m}
\ret&&
\le(n+m)R\,2^{\nu m},
\label{sumxI<}
\end{eqnarray}
where we used (\ref{txp<tR}) and noted that $|g|=\sqrt{\nu}/2$ for $g\in\calF_f$.
In the final step, we used the assumption $\sqrt{\nu}/2\le R$ introduced right
after (\ref{txp<tR}).

We substitute the bound (\ref{sumI<}) for $I_{x,y}$, and the bound (\ref{a1bound}) for 
$\al_1$ to (\ref{phi-psi}) to get
\begin{eqnarray}
&&
\sum_{y\in\La}\abs{\phis{x}{y}-\psis{x}{y}}
\ret&&
\le
\sum_{n=1}^\infty\sum_{m=0}^\infty
\rbk{\frac{\akappa}{\la^2}}^n\rbk{\frac{1}{\la}}^m
b^{2n}(b-1)^m
\frac{\beta^n}{(n+1)^{b^2+1}}\frac{\gamma^m}{(m+1)^b}
2^{\nu m}
\ret&&
\le
\cbk{\sum_{n=1}^\infty\rbk{\frac{\akappa}{\la^2}b^2\beta}^n}
\cbk{\sum_{m=0}^\infty\rbk{\frac{1}{\la}2^\nu(b-1)\gamma}^m}
\ret&&
\le B_1\frac{\akappa}{\la^2}
\label{phi-psi2}
\end{eqnarray}
for $\kappa$ and $\la$ satisfying (\ref{k/lam<2}) and (\ref{lam>2})
(or (\ref{k/lam<1}) and (\ref{lam>1})).
The constant $B_1$ will be fixed later.

Similarly we use (\ref{sumxI<}) to get
\begin{eqnarray}
&&
\sum_{y\in\La}\abs{x-y}\abs{\phis{x}{y}-\psis{x}{y}}
\ret&&
\le
\sum_{n=1}^\infty\sum_{m=0}^\infty
\rbk{\frac{\akappa}{\la^2}}^n\rbk{\frac{1}{\la}}^m
b^{2n}(b-1)^m
\frac{\beta^n}{(n+1)^{b^2+1}}\frac{\gamma^m}{(m+1)^b}
(n+m)R\,2^{\nu m}
\ret&&
\le B_1R\frac{\akappa}{\la^2},
\end{eqnarray}
which is the desired bound (\ref{reg2}).
The bounds (\ref{reg1}), (\ref{reg2}) for $x\in\La'$ as well as the remaining
bounds (\ref{reg3}), (\ref{reg4}) follow in the same manner.

The bounds (\ref{regtau1}) and (\ref{regtau2}) for the effective hopping $\tau_{x,y}$
stated in Lemma~\ref{tauLemma} are proved in exactly the same manner by using the
definition (\ref{tauDef}), the expansion (\ref{ekexp}) for $\ep_1(k)$, and the
bounds (\ref{alpha2<}) for the coefficients.

The bounds (\ref{reg5}),  (\ref{reg6}),  (\ref{reg7}), and  (\ref{reg8}) for the dual
basis states can also be shown in the same spirit.
A major difference is that $\phitil^{(x)}$ does not coincide with the 
strictly localized state $\psi^{(x)}$ when $\kappa=0$.
To control this situation, we note
\begin{equation}
\abs{\phit{x}{y}-\psis{x}{y}}\le
\abs{\phit{x}{y}-\psit{x}{y}}+\abs{\psit{x}{y}-\psis{x}{y}},
\label{pt-ps}
\end{equation}
where $\psit{x}{y}$ (which is the dual basis states for $\kappa=0$) is defined as
\begin{equation}
\psit{x}{y}=(2\pi)^{-d}\int dk\,\emik{(x-y)}\,\wkt{\mu(x)}{\mu(y)},
\label{pstDef}
\end{equation}
with $\wt{u}$ defined in (\ref{wtil}).
By using the series expansion (\ref{vdexp}), we can control the term
$\abs{\phit{x}{y}-\psit{x}{y}}$ in exactly the same way as we controlled
$\abs{\phis{x}{y}-\psis{x}{y}}$ in the above.
Consequently, we get
\begin{equation}
\sum_{x\in\La}\abs{\phit{x}{y}-\psit{x}{y}}
\le B_1\frac{\akappa}{\la^2}, 
\quad
\sum_{y\in\La}\abs{\phit{x}{y}-\psit{x}{y}}
\le B_1\frac{\akappa}{\la^2},
\label{reg9}
\end{equation}
and 
\begin{equation}
\sum_{x\in\La}\abs{x-y}\abs{\phit{x}{y}-\psit{x}{y}}
\le B_1R\frac{\akappa}{\la^2},
\quad 
\sum_{y\in\La}\abs{x-y}\abs{\phit{x}{y}-\psit{x}{y}}
\le B_1R\frac{\akappa}{\la^2}.
\label{reg10}
\end{equation}
At this stage, we fix the constant $B_1$ so that the bounds
(\ref{reg1})-(\ref{reg4}), (\ref{reg9}), and (\ref{reg10}) are
simultaneously satisfied\footnote{
Of course it is possible to state the bounds (\ref{reg1})-(\ref{reg4}) with
smaller $B_1$ than in (\ref{reg9}) or (\ref{reg10}).
We have unified the coefficients as much as possible to make the 
formulas less complicated.
}.
Note that $B_1$ depends only on the band number $b$.

To control the second term in (\ref{pt-ps}), we first note
\begin{equation}
\psit{x}{y}-\psis{x}{y}=
(2\pi)^{-d}\int dk\,\emik{(x-y)}\,\zetas{\mu(x)}{\mu(y)},
\label{pst-ps}
\end{equation}
with
\begin{equation}
\zetas{u}{u'}=
\left\{
\begin{array}{ll}
{
-\rbk{{A(k)}/{\la^2}}\rbk{1+{A(k)}/{\la^2}}^{-1}\times
\left\{
\begin{array}{l}
1\\-{\Ck{u'}}/{\la}\\{\Ck{u}}/{\la}
\end{array}
\right.
}
&
{
\begin{array}{l}
\mbox{if $u=u'=o$;}\\
\mbox{if $u=o$, $u'\in\calU'$;}\\
\mbox{if $u\in\calU'$, $u'=o$;}
\end{array}
}
\\
-\rbk{{\Ck{u}\Ck{u'}}/{\la^2}}\rbk{1+{A(k)}/{\la^2}}^{-1}
&
\begin{array}{l}
\mbox{if $u,u'\in\calU'$.}
\end{array}
\end{array}
\right.
\label{zetaDef}
\end{equation}
The expressions (\ref{pst-ps}), (\ref{zetaDef}) are straightforward consequences
of (\ref{pstDef}), (\ref{wtil}), (\ref{wDef}), and (\ref{veuk}).
By expanding $\cbk{1+(A(k)/\la^2)}^{-1}$ in (\ref{zetaDef}), 
$\psit{x}{y}-\psis{x}{y}$ can be expressed as a power series of $\la^{-2}$.
By analyzing the series, it is easily shown that, for $\la\ge\la_0$, the
summations
$\sum_{y\in\La}\abs{\psit{x}{y}-\psis{x}{y}}$,
$\sum_{y\in\La}\abs{x-y}\abs{\psit{x}{y}-\psis{x}{y}}$,
$\sum_{x\in\La}\abs{\psit{x}{y}-\psis{x}{y}}$,
and
$\sum_{x\in\La}\abs{x-y}\abs{\psit{x}{y}-\psis{x}{y}}$
are all bounded from above by $B_2/\la^2$, where $B_2$ is a constant
which depend only on the band number $b$.
By combining these bounds with (\ref{pt-ps}), (\ref{reg9}), and (\ref{reg10}), 
we get the desired bounds
(\ref{reg5}), (\ref{reg6}), (\ref{reg7}), and (\ref{reg8}).
\par\bigskip
{\small
I wish to thank Tohru Koma and Andreas Mielke 
for stimulating discussions and important suggestions which made
the present work possible.
I also thank
Hideo Aoki, 
Yasuhiro Hatsugai,
Arisato Kawabata,
Kenn Kubo,
Koichi Kusakabe, 
Elliott Lieb,
Hiroshi Mano,
and
Dieter Vollhardt
for useful discussions on various related topics.
Finally I thank 
Tom Kennedy
and
Masanori Yamanaka
for useful comments on the paper.}


\bibliography{myWorks,CM}
\bibliographystyle{myabbrv}
\end{document}

%% file: stabMacro.tex
\newcommand{\ret}{\nonumber \\}
\newcommand{\bigno}{\par\bigskip\noindent}
\newcommand{\Section}[1]%
{\section{#1}\setcounter{equation}{0}%
\setcounter{theorem}{0}}
\newtheorem{theorem}{Theorem}
\newtheorem{lemma}[theorem]{Lemma}
\newtheorem{coro}[theorem]{Corollary}

\newenvironment{proof}[1]%
{\par\noindent{\em #1:\ }}%
{~\rule{2mm}{2mm}\par\bigskip}
\newcommand{\Remark}{\bigno{\bf Remark:}\ }
\newcommand{\Rem}{\bigno{\bf Remark:}\ }
\newcommand{\abs}[1]{\left|#1\right|}
\newcommand{\norm}[1]{\left\Vert#1\right\Vert}
\newcommand{\rbk}[1]{\left(#1\right)}
\newcommand{\sbk}[1]{\left[#1\right]}
\newcommand{\cbk}[1]{\left\{#1\right\}}
\newcommand{\set}[2]{\left\{#1\,\Bigl|\,#2\right\}}
\newcommand{\sumtwo}[2]%
{\mathop{\sum_{#1}}_{#2}}
\newcommand{\sumthree}[3]%
{\mathop{\mathop{\sum_{#1}}_{#2}}_{#3}}
\newcommand{\sumfour}[4]%
{\mathop{\mathop{\mathop{\sum_{#1}}_{#2}}_{#3}}_{#4}} 
\newcommand{\suptwo}[2]%
{\mathop{\sup_{#1}}_{#2}}
\newcommand{\supthree}[3]%
{\mathop{\mathop{\sup_{#1}}_{#2}}_{#3}}
\newcommand{\supfour}[4]%
{\mathop{\mathop{\mathop{\sup_{#1}}_{#2}}_{#3}}_{#4}} 
\newcommand{\inftwo}[2]%
{\mathop{\inf_{#1}}_{#2}}
\newcommand{\infthree}[3]%
{\mathop{\mathop{\inf_{#1}}_{#2}}_{#3}}
\newcommand{\inffour}[4]%
{\mathop{\mathop{\mathop{\inf_{#1}}_{#2}}_{#3}}_{#4}} 

\newcommand\calB{{\cal B}}

\newcommand\calF{{\cal F}}

\newcommand\calH{{\cal H}}

\newcommand\calK{{\cal K}}

\newcommand\calU{{\cal U}}

\newcommand{\Zd}{{\bf Z}^d}
\newcommand{\La}{\Lambda}
\newcommand{\up}{\uparrow}
\newcommand{\dn}{\downarrow}
\newcommand{\bs}{\backslash}
\renewcommand{\kappa}{\rho}
\newcommand{\Lao}{\Lambda_o}
\newcommand{\Lau}{\Lambda_u}
\newcommand{\Laop}{\Lambda_o^+}
\newcommand{\Laob}{\bar{\Lambda}_o}
\newcommand{\dnu}{{d\choose \nu}}
\newcommand{\akappa}{\abs{\kappa}}
\newcommand{\cxs}{c^\dagger_{x,\sigma}}
\newcommand{\axs}{c_{x,\sigma}}
\newcommand{\cys}{c^\dagger_{y,\sigma}}
\newcommand{\ays}{c_{y,\sigma}}
\newcommand{\cyt}{c^\dagger_{y,\tau}}
\newcommand{\ayt}{c_{y,\tau}}
\newcommand{\cd}{c^\dagger}
\newcommand{\ad}{a^\dagger}
\newcommand{\bd}{b^\dagger}
\newcommand{\adxs}{a^\dagger_{x,\sigma}}
\newcommand{\bxs}{b_{x,\sigma}}
\newcommand{\nxs}{n_{x,\sigma}}
\newcommand{\Emin}{E_{\rm min}}
\newcommand{\ESW}{E_{\rm SW}(k)}
\newcommand{\Hhop}{H_{\rm hop}}
\newcommand{\Htil}{\widetilde{H}}
\newcommand{\Hhopt}{\widetilde{H}_{\rm hop}}
\newcommand{\Hhopo}{H_{\rm hop}^{(0)}}
\newcommand{\Hint}{H_{\rm int}}
\newcommand{\Hinto}{H_{\rm int}^{(1)}}
\newcommand{\Hintt}{H_{\rm int}^{(2)}}
\newcommand{\Hpert}{H'_{\rm hop}}
\newcommand{\Stot}{S_{\rm tot}}
\newcommand{\Sztot}{S_{\rm tot}^{(3)}}
\newcommand{\Smax}{S_{\rm max}}
\newcommand{\Ttil}{\widetilde{T}}
\newcommand{\ttil}{\tilde{t}}
\newcommand{\ntil}{\tilde{n}}
\newcommand{\St}[2]{\widetilde{S}^{(#1)}_{#2}}
\newcommand{\hhop}{h_{\rm hop}}
\newcommand{\hint}{h_{\rm int}}
\newcommand{\htint}{\tilde{h}_{\rm int}}
\newcommand{\vac}{\Phi_{\rm vac}}
\newcommand{\UP}{\Phi_\uparrow}
\newcommand{\Psik}[1]{\Psi_{#1}(k)}
\newcommand{\Phik}[1]{\Phi_{#1}(k)}
\newcommand{\Ok}{\Omega(k)}
\newcommand{\Bk}{\calB_k}
\newcommand{\Po}{P_0}
\newcommand{\Hsing}{\calH_{\rm single}}
\newcommand{\Hsingo}{\calH_{\rm single}^{(1)}}
\newcommand{\Hsingp}{\calH'_{\rm single}}
\newcommand{\eik}[1]{e^{ik\cdot{#1}}}
\newcommand{\emik}[1]{e^{-ik\cdot{#1}}}
\newcommand{\eikx}{e^{ik\cdot x}}
\newcommand{\eikp}{e^{ik\cdot p}}
\newcommand{\la}{\lambda}
\renewcommand{\epsilon}{\varepsilon}
\newcommand{\ep}{\varepsilon}
\newcommand{\Ut}{\widetilde{U}}
\newcommand{\Dt}{\widetilde{D}}
\newcommand{\al}{\alpha}
\newcommand{\alt}{\tilde{\alpha}}
\newcommand{\altn}[2]{\tilde{\alpha}^{(#1)}_{#2}}
\newcommand{\xin}[2]{\xi^{(#1)}_{#2}}
\renewcommand{\phi}{\varphi}
\newcommand{\phitil}{\tilde{\varphi}}
\newcommand{\psis}[2]{\psi^{(#1)}_{#2}}
\newcommand{\psit}[2]{\tilde{\psi}^{(#1)}_{#2}}
\newcommand{\phis}[2]{\varphi^{(#1)}_{#2}}
\newcommand{\phit}[2]{\tilde{\varphi}^{(#1)}_{#2}}
\newcommand{\etas}[2]{\eta^{(#1)}_{#2}}
\newcommand{\etat}[2]{\tilde{\eta}^{(#1)}_{#2}}
\newcommand{\etab}[2]{\bar{\eta}^{(#1)}_{#2}}
\newcommand{\zetas}[2]{\zeta^{(#1)}_{#2}}
\newcommand{\Gammat}[2]{\widetilde{\Gamma}^{(#1)}_{#2}}
\newcommand{\vk}[1]{v_{#1}(k)}
\newcommand{\vv}{{\bf v}(k)}
\newcommand{\wk}[1]{w_{#1}(k)}
\newcommand{\wv}{{\bf w}(k)}
\newcommand{\vov}{{\bf v}^{(o)}(k)}
\newcommand{\vok}[1]{v^{(o)}_{#1}(k)}
\newcommand{\vo}[1]{{\bf v}^{(o)}_{#1}(k)}
\newcommand{\vvv}[1]{{\bf v}^{(#1)}(k)}
\newcommand{\vks}[2]{v^{(#1)}_{#2}(k)}
\newcommand{\vt}[1]{\tilde{\bf v}^{(#1)}(k)}
\newcommand{\vkt}[2]{\tilde{v}^{(#1)}_{#2}(k)}
\newcommand{\wt}[1]{\tilde{\bf w}^{(#1)}(k)}
\newcommand{\wkt}[2]{\tilde{w}^{(#1)}_{#2}(k)}
\newcommand{\ak}[1]{a_{#1}(k)}
\newcommand{\av}{{\bf a}(k)}
\newcommand{\gv}{{\bf g}(k)}
\newcommand{\It}{{\sf I}}
\newcommand{\Mt}{{\sf M}(k)}
\newcommand{\Qt}{{\sf Q}(k)}
\newcommand{\Gt}{{\sf G}(k)}
\newcommand{\Mk}[1]{M_{#1}(k)}
\newcommand{\Qk}[1]{Q_{#1}(k)}
\newcommand{\Pt}{{\sf P}(k)}
\newcommand{\Pk}[1]{P_{#1}(k)}
\newcommand{\Wt}{{\sf W}(k)}

\newcommand{\Ht}{{\sf H}(k)}
\newcommand{\Gm}[1]{(G)_{#1}}
\newcommand{\Gi}[1]{(G^{-1})_{#1}}
\newcommand{\Ck}[1]{C_{#1}(k)}
\newcommand{\rbks}[1]{\left(#1\right)^*}
\newcommand{\captionA}{
The one-dimensional lattice studied in Section~\protect\ref{Sec1d}.
We identify the left most site with the right most site to
get a closed chain.
The black dots represent sites in $\Lao$ (metallic atoms), 
and the gray dots
represent sites in $\La'$ (oxygen atoms).
There are two types of hopping $t$, $s$, and on-site (one-body)
potential $V$.
In addition we have on-site Coulomb repulsion $U>0$ at each site.
There are $2L$ sites in the lattice, and we put $L$ electrons in the
system.
(Here $L=5$.)
For the flat-band models characterized by $s=\la t$, $V=(\la^2-2)t$
with $\la>0$, $t>0$, the ground states of the models are proved to be
ferromagnetic.
Here we prove the local stability of ferromagnetism for models
obtained by adding small perturbations to the flat-band models.
}
\newcommand{\captionB}{
The dispersion relations $\ep_1(k)$, $\ep_2(k)$ in the 
one-dimensional models.
a)~The flat-band model with $\la=2$, $\kappa=0$.
b)~The perturbed model with $\la=2$, $\kappa=0.2$, which has
a nearly flat lower band and an energy gap between the two bands.
}
\newcommand{\captionC}{
A schematic picture of the state $\Gamma_x$ (\protect\ref{Gammax}) which appears 
in the definition of $\Ok$ (\protect\ref{Omega1d}).
Since $\Gamma_x$ is constructed by using sharply localized basis states
for the lower band,
it costs small Coulomb repulsion energy and small kinetic energy.
In the state $\Ok$, the down-spin propagates with momentum $k$
and further reduces the total energy to 
$\simeq E_0+(4U/\la^4)\cbk{\sin(k/2)}^2$.
}
\newcommand{\captionD}{
When the ``exchange'' Hamiltonian $H_1$ acts on the state $\Gamma_y$,
four terms are generated.
Two of them are the same as $\Gamma_y$, while the electronic
spins are exchanged in the other two states.
The process illustrated here can be regarded as the ultimate origin of 
ferromagnetism in the present model.
}
\newcommand{\captionE}{
When the (effective hopping) operator 
\protect\newline
$\sum_{x\in\La,\,\sigma=\up,\dn}
\Ut_{x+1,x+1;x+1,x}\,
\ad_{x+1,\sigma}\ad_{x+1,-\sigma}b_{x+1,-\sigma}b_{x,\sigma}$
acts on $\Gamma_y$, two states with a bound pair of a hole and a doubly
occupied site are generated.
Note that the two resulting states are related 
through the translation by a 
distance $1$.
This process is the major source of instability against antiferromagnetism.
}
\newcommand{\captionF}{
Schematic pictures of the states $\Ok$, $\Xi_+(k)$, and $\Theta(k)$,
and the matrix elements between them.
As for the off-diagonal matrix elements, we only present the main
part of their absolute values.
Note that the matrix elements between $\Ok$ and $\Xi_+(k)$ are highly
asymmetric (apart from the artificial asymmetry factor $\alpha(k)$).
The small out-going matrix elements from $\Ok$ indicates that the state
$\Ok$ is a good trial state for the spin-wave excitation.
}
\newcommand{\captionG}{
The lattice $\La$ in two dimensions ($d=2$) with
a)~$\nu=1$, and b)~$\nu=2$.
The black dots are sites in $\Lao$ and the gray dots are
sites in $\La'$.
One may interpret black sits as metallic atoms and
gray sites as oxygen atoms.
}
\newcommand{\captionH}{
The dispersion relation for the three-band model with $d=2$, $\nu=1$.
We have set $t>0$, $\la=2$, and $\kappa=0$ to get a flat-band model.
There are two flat bands, and one cosine band.
}
\newcommand{\captionI}{
The dispersion relation for the three-band model with $d=2$, $\nu=1$.
The perturbation is given by
$t'_{x,x}=t$ if $x\in\calF_{(0,1/2)}$,
$t'_{x,x}=-t$ if $x\in\calF_{(1/2,0)}$,
$t'_{x,y}=t'_{y,x}=t$ if $x\in\calF_o$ and $y=x+(1,1)$,
and 
$t'_{x,y}=0$ otherwise.
We have set $t>0$, $\la=2$, and $\kappa=0.7$.
Note that the two lower bands become dispersive, and there appears
a gap between the second and the third bands.
}